\documentclass[twoside,11pt]{article}

\usepackage{jmlr2e}

\usepackage[linesnumbered, ruled]{algorithm2e}
\usepackage[colorlinks=false]{hyperref}
\usepackage{Definitions}
\usepackage{empheq}
\usepackage{graphicx,tikz}
\usepackage[noend]{algorithmic}
\usepackage{subcaption}

\newcommand{\continmax}{\textsc{ConTinEst}\xspace}
\newcommand{\influmax}{\textsc{Influmax}\xspace}
\newcommand{\netrate}{\textsc{NetRate}\xspace}
\newcommand{\spm}{\textsc{SP1M}\xspace}
\newcommand{\pmia}{\textsc{PMIA}\xspace}

\newtheorem{claim}{Claim}

\newcommand{\Graph}{\Gcal}
\newcommand{\Node}{\Vcal}

\newcommand{\Edge}{\Ecal}

\newcommand{\Item}{\Lcal}
\newcommand{\Ground}{\Zcal}
\newcommand{\nGround}{N}
\newcommand{\Ind}{\Ical}

\newcommand{\budget}{b}
\newcommand{\attention}{u}

\newcommand{\Greedy}{G}
\newcommand{\g}{g}
\newcommand{\Optimal}{O}

\newcommand{\cur}{c}

\newcommand{\budgetmax}{\textsc{BudgetMax}\xspace}
\newcommand{\gdegree}{{GreedyDegree}\xspace}

\ShortHeadings{Continuous-Time Influence Maximization of Multiple Items}{Du, Liang, Balcan, Gomez-Rodriguez, Zha and Song}
\firstpageno{1}

\begin{document}

\title{Scalable Influence Maximization for Multiple Products in \\ Continuous-Time Diffusion Networks}
\author{\name Nan Du \email dunan@google.com \\
  \addr Google Research, 
        1600 Amphitheatre Pkwy, Mountain View, CA 94043	
        \AND
  \name Yingyu Liang \email yingyul@cs.princeton.edu \\
  \addr Department of Computer Science, Princeton University, 
        Princeton, NJ 08540
  \AND
  \name Maria-Florina Balcan \email ninamf@cs.cmu.edu \\
  \addr School of Computer Science, 
        Carnegie Mellon University, 
        Pittsburgh, PA 15213
  \AND  
  \name Manuel Gomez-Rodriguez \email manuelgr@mpi-sws.org \\
  \addr MPI for Software Systems, 
            Kaiserslautern, Germany 67663
  \AND
  \name Hongyuan Zha \email zha@cc.gatech.edu\\
  \name Le Song \email lsong@cc.gatech.edu\\
  \addr College of Computing, 
        Georgia Institute of Technology, 
        Atlanta, GA 30332
       }

\editor{}

\maketitle

\begin{abstract}

A typical viral marketing model identifies influential users in a social network to maximize a single product adoption assuming unlimited user attention, campaign budgets, and time. In reality, multiple products need campaigns, users have limited attention, convincing users incurs costs, and advertisers have limited budgets and expect the adoptions to be maximized soon. Facing these user, monetary, and timing constraints, we formulate the problem as a submodular maximization task in a continuous-time diffusion model under the intersection of one matroid and multiple knapsack constraints. We propose a randomized algorithm estimating the user influence\footnote{Partial results in the paper on influence estimation have been published in a conference paper: Nan Du, Le Song, Manuel Gomez-Rodriguez, and Hongyuan Zha. Scalable influence estimation in continuous time diffusion networks. In Advances in Neural Information Processing Systems 26, 2013.} in a network ($|\Vcal|$ nodes, $|\Ecal|$ edges) to an accuracy of $\epsilon$ with $n=\Ocal(1/\epsilon^2)$ randomizations and $\tilde\Ocal(n|\Ecal|+n|\Vcal|)$ computations. By exploiting the influence estimation algorithm as a subroutine, we develop an adaptive threshold greedy algorithm achieving an approximation factor $k_a/(2+2 k)$ of the optimal when $k_a$ out of the $k$ knapsack constraints are active. Extensive experiments on networks of millions of nodes demonstrate that the proposed algorithms achieve the state-of-the-art in terms of effectiveness and scalability.

\end{abstract}


\section{Introduction}

Online social networks play an important role in the promotion of new products, the spread of news, the success of political campaigns, and the diffusion of technological innovations. In these contexts, the influence maximization problem (or viral marketing problem) typically has the following flavor: identify a set of influential users in a social network, who, when convinced to adopt a product, shall influence other users in the network and trigger a large cascade of adoptions. This problem has been studied extensively in the literature from both the modeling and the algorithmic aspects~\citep{RicDom02, KemKleTar03,  Leskovecetal07, CheWanYan09, CheWanWan2010, CheYuaYifZha2010, CheColCumKeetal2011, CheLuZha2012, LenBonCas10,GoyLuLak2011a,GoyLuLak2011b,GomSch12},
where it has been typically assumed that the host (\eg, the owner of an online social platform) faces a single product, endless user attention, unlimited budgets and unbounded time. However, in reality, the host often encounters a much more constrained scenario:
\begin{itemize}
  \item {\bf Multiple-Item Constraints:} multiple products can spread simultaneously among the same set of social entities. These products may have different characteristics, such as their revenues and speed of spread.
  \item {\bf Timing Constraints:} the advertisers expect the influence to occur within a certain time window, and different products may have different timing requirements.
  \item {\bf User Constraints:} users of the social network, each of which can be a potential source, would like to be exposed to only a small number of ads. Furthermore, users may be grouped by their geographical locations, and advertisers may have a target population they want to reach.
  \item {\bf Product Constraints:} seeking initial adopters entails a cost to the advertiser, who needs to pay to the host and often has a limited amount of money.
\end{itemize}
For example, Facebook (\ie, the host) needs to allocate ads for various products with different characteristics, \eg, clothes, books, or cosmetics. 
While some products, such as clothes, aim at influencing within a short time window, some others, such as books, may allow for longer periods. 
Moreover, Facebook limits the number of ads in each user'{}s side-bar (typically it shows less than five) and, as a consequence, it cannot assign all ads to a few highly influential users.
%
%
Finally, each advertiser has a limited budget to pay for ads on Facebook and thus each ad can only be displayed to some subset of users.
%
In our work, we incorporate these myriads of practical and important requirements into consideration in the influence maximization problem.

We account for the multi-product and timing constraints by applying product-specific continuous-time diffusion models. 
Here, we opt for continuous-time diffusion models instead of discrete-time models, which have been mostly used in previous work~\citep{KemKleTar03,CheWanYan09, CheWanWan2010, CheYuaYifZha2010,CheColCumKeetal2011, CheLuZha2012,BorBraChaLuc12}. This is because artificially discretizing the time axis into bins introduces additional errors. One can adjust the additional tuning parameters, like the bin size, to balance the tradeoff between the error and the computational cost, but the parameters are not easy to choose optimally. Extensive experimental comparisons on both synthetic and real-world data have shown that discrete-time models provide less accurate influence estimation than their continuous-time counterparts~\citep{GomBalSch11, GomSch12, GomLesSch2013a, DuSonZhaGom13, DuSonWooZha13}. 

%
%
However, maximizing influence based on continuous-time diffusion models also entails additional challenges. 
First, evaluating the objective function of the influence maximization problem (\ie, the influence estimation problem) in this setting is a difficult graphical model inference problem, \ie, computing the marginal density of continuous variables in loopy graphical models. The exact answer can be computed only for very special cases. For example,~\cite{GomSch12} have shown that the problem can be solved exactly when the transmission functions are exponential densities, by using continuous time Markov processes theory. However, the computational complexity of such approach, in general, scales exponentially with the size and density of the network. Moreover, extending the approach to deal with arbitrary transmission functions would require additional nontrivial approximations which would increase even more the computational complexity.
Second, it is unclear how to scale up influence estimation and maximization algorithms based on continuous-time diffusion models to millions of nodes. Especially in the maximization case, the influence estimation procedure needs to be called many times for different subsets of selected nodes. \emph{Thus, our first goal is to design a scalable algorithm which can perform influence 
estimation in the regime of networks with millions of nodes.}

We account for the user and product constraints by restricting the feasible domain over which the maximization is performed. 
We first show that the overall influence function of multiple products is a submodular function and then realize that the user and product constraints correspond to constraints over the ground set of this submodular function. 
To the best of our knowledge, previous work has not considered both user and product constraints simultaneously over general unknown different diffusion networks with non-uniform costs.  
In particular, \citep{SamAniNis2010} first tried to model both the product and user constraints only with uniform costs and infinite time window, which essentially reduces to a special case of our formulations. Similarly, ~\citep{LuBonAmiLak13} considered the allocation problem of multiple products which may have competitions within the infinite time window. Besides, they all assume that multiple products spread within the same network. In contrast, our formulations generally allow products to have different diffusion networks, which can be unknown in practice. 
~\cite{SomKakInaKaw14} studied the influence maximization problem for one product subject to one knapsack constraint over a known bipartite graph between marketing channels and potential customers;
~\cite{LenBonCas10} and~\cite{SunCheLiuWanetal11} considered user constraints but disregarded product constraints during the initial assignment;
and, ~\cite{NarNan12} studied the cross-sell phenomenon (the selling of the first product raises the chance of selling the second) and included monetary constraints for all the products. 
However, no user constraints were considered, and the cost of each user was still uniform for each product. 
\emph{Thus, our second goal is to design an efficient submodular maximization algorithm which can take into account both user and product constraints simultaneously.}

Overall, this article includes the following major contributions:
\begin{itemize}
\item Unlike prior work that considers an a priori described simplistic discrete-time diffusion model, we first {\bf learn} the diffusion networks from data by using continuous-time diffusion models. 
This allows us to address the timing constraints in a principled way.
\item We provide a novel formulation of the influence estimation problem in the continuous-time diffusion model from the perspective of probabilistic graphical models, which allows heterogeneous diffusion dynamics over the edges.
\item We propose an efficient randomized algorithm for continuous-time influence estimation, which can scale up to millions of nodes and estimate the influence of each node to an accuracy of $\epsilon$ using  $n=\Ocal(1/\epsilon^2)$ randomizations.
\item We formulate the influence maximization problem with the aforementioned constraints as a submodular maximization under the intersection of matroid constraints and knapsack constraints. The submodular function we use is based on the actual diffusion model learned from the data for the time window constraint. This novel formulation provides us a firm theoretical foundation for designing greedy algorithms with theoretical guarantees.
\item We develop an efficient adaptive-threshold greedy algorithm which is linear in the number of products and proportional to $\widetilde \Ocal(|\Node|+|\Edge^*|)$, where $|\Node|$ is the number of nodes (users) and $|\Edge^*|$ is the number of edges in the largest diffusion network. 
We then prove that this algorithm is guaranteed to find a solution with an overall influence of at least $\frac{k_a}{2+2k}$ of the optimal value, when $k_a$ out of the $k$ knapsack constraints are active. 
This improves over the best known approximation factor achieved by polynomial time algorithms in the combinatorial optimization literature. 
Moreover, whenever advertising each product to each user entails the same cost, the constraints reduce to an intersection of matroids, and we obtain an approximation factor of $1/3$, which is optimal for such optimization.

\item We evaluate our algorithms over large synthetic and real-world datasets and show that our proposed methods significantly improve over previous state-of-the-arts in terms of both the accuracy of the estimated influence and the quality of the selected nodes in maximizing the influence over independently hold-out real testing data.

\end{itemize}

In the remainder of the paper, we will first tackle the influence estimation problem in section~\ref{sec:influestimation}. We then formulate different realistic constraints for the influence maximization in section~\ref{sec:modelconstraints} and present the adaptive-thresholding greedy algorithm with its theoretical analysis in section \ref{sec:influmaximization}; we investigate the performance of the proposed algorithms in both synthetic and real-world datasets in section~\ref{sec:exp}; and finally we conclude in section~\ref{sec:cl}.

\section{Influence Estimation\label{sec:influestimation}}

We start by revisiting the continuous-time diffusion model by~\cite{GomBalSch11} and then explicitly formulate the influence estimation problem from the perspective of probabilistic graphical models.
Because the efficient inference of the influence value for each node is highly non-trivial, we further develop a scalable influence estimation algorithm which is able to handle networks of millions of nodes. 
The influence estimation procedure will be a key building block for our later influence maximization algorithm.

\subsection{Continuous-Time Diffusion Networks}
%

The continuous-time diffusion model associates each edge with a transmission function, that is, a density over the transmission time along the edge, in contrast to previous discrete-time models which associate each edge with a fixed infection probability~\citep{KemKleTar03}. Moreover, it also differs from discrete-time models in the sense that events in a cascade are not generated iteratively in rounds, but event timings are sampled directly from the transmission function in the continuous-time model.

{\bf Continuous-Time Independent Cascade Model.} Given a \emph{directed} contact network, $\Gcal = (\Vcal,\Ecal)$, we use the independent cascade model for modeling a diffusion process~\citep{KemKleTar03,GomBalSch11}. The process begins with a set of infected source nodes, $\Acal$, initially adopting certain contagion (idea, meme or product) at time zero. The contagion is transmitted from the sources along their out-going edges to their direct neighbors. Each transmission through an edge entails \emph{random} waiting times, $\tau$, drawn from different independent pairwise waiting time distributions(one per edge). Then, the infected neighbors transmit the contagion to their respective neighbors, and the process continues. We assume that an infected node remains infected for the entire diffusion process. Thus, if a node $i$ is infected by multiple neighbors, only the neighbor that first infects node $i$ will be the \emph{true} parent. As a result, although the contact network can be an arbitrary directed network, each diffusion process induces a Directed Acyclic Graph (DAG).

{\bf Heterogeneous Transmission Functions.}
Formally, the pairwise transmission function $f_{ji}(t_i|t_j)$ for a directed edge $j\rightarrow i$ is the conditional density of node $i$ getting infected at time $t_i$ given that node $j$ was infected at time $t_j$. We assume
it is shift invariant: $f_{ji}(t_i|t_j) = f_{ji}(\tau_{ji})$, where $\tau_{ji}:=t_i - t_j$, and causal: $f_{ji}(\tau_{ji}) = 0$ if $\tau_{ji} < 0$. Both parametric transmission functions, such as the exponential and Rayleigh
function~\citep{GomBalSch11}, and nonparametric functions~\citep{DuSonSmoYua12} can be used and estimated from cascade data.

{\bf Shortest-Path Property.} The independent cascade model has a useful property we will use later: given a sample of transmission times of all edges, the time $t_i$ taken to infect a node $i$ is the length of the shortest path in $\Gcal$ from the sources to node $i$, where the edge weights correspond to the associated transmission times.

\subsection{Probabilistic Graphical Model for Continuous-Time Diffusion Networks\label{sec:graphicalmodel}}
The continuous-time independent cascade model is essentially a directed graphical model for a set of \emph{dependent} random variables, that is, the infection times $t_i$ of the nodes, where the conditional independence structure is supported on the contact network $\Gcal$. Although the original contact graph $\Gcal$ can contain directed loops, each diffusion process (or a cascade) induces a directed acyclic graph (DAG). For those cascades consistent with a particular DAG, we can model the joint density of $t_i$ using a directed graphical model:
\begin{align}
	p\rbr{\{t_i\}_{i\in\Vcal}} = \prod\nolimits_{i \in \Vcal} p\rbr{t_i | \{t_j\}_{j \in \pi_i}},
	\label{eq:dag_factorization}
\end{align}
where each $\pi_i$ denotes the collection of parents of node $i$ in the induced DAG, and each term $p(t_i | \{t_j\}_{j \in \pi_i})$ corresponds to a conditional density of $t_i$ given the infection times of the parents of node $i$. This is true because given the infection times of node $i$'{}s parents, $t_i$ is independent of other infection times, satisfying the local Markov property of a directed graphical model. We note that the independent cascade model only specifies explicitly the pairwise transmission function of each directed edge, but does not directly define the conditional density $p(t_i | \{t_j\}_{j \in \pi_i})$.

However, these conditional densities can be derived from the pairwise transmission functions based on the Independent-Infection property:
\begin{align}
	p\rbr{t_i | \{t_j\}_{j \in \pi_i}} = \sum\nolimits_{j\in\pi_i}f_{ji}(t_i|t_j)\prod\nolimits_{l\in\pi_{i}, l\neq j}S(t_i|t_l),
	\label{eq:cpt_transmission}
\end{align}
which is the sum of the likelihoods that node $i$ is infected by each parent node $j$. More precisely, each term in the summation can be interpreted as the likelihood $f_{ji}(t_i|t_j)$ of node $i$ being infected at $t_i$ by node $j$ multiplied by the probability $S(t_i|t_l)$ that it has survived from the infection of each other parent node $l\neq j$ until time $t_i$.

Perhaps surprisingly, the factorization in Equation~\eq{eq:dag_factorization} is the same factorization that can be used for an arbitrary induced DAG consistent with the contact network $\Gcal$. In this case, we only need to replace the definition of $\pi_i$ (the parent of node $i$ in the DAG)  to the set of neighbors of node $i$ with an edge pointing to node $i$ in $\Gcal$. This is not immediately obvious from Equation~\eq{eq:dag_factorization}, since the contact network $\Gcal$ can contain directed loops which seems to be in conflict with the conditional independence semantics of directed graphical models. The reason why it is possible to do so is as follows: any fixed set of infection times, $t_1,\ldots,t_d$, induces an ordering of the infection times. If $t_i \leq t_j$ for an edge $j\rightarrow i$ in $\Gcal$, $h_{ji}(t_i|t_j)=0$, and the corresponding term in Equation~\eq{eq:cpt_transmission} is zeroed out, making the conditional density consistent with the semantics of directed graphical models.

Instead of directly modeling the infection times $t_i$, we can focus on the set of mutually \emph{independent} random transmission times $\tau_{ji} = t_i - t_j$. Interestingly, by switching from a node-centric view to an edge-centric view, we obtain a fully factorized joint density of the set of transmission times
\begin{align}
  p\rbr{\{\tau_{ji}\}_{(j,i)\in \Ecal}} = \prod\nolimits_{(j,i)\in\Ecal} f_{ji}(\tau_{ji}),
\end{align}
Based on the Shortest-Path property of the independent cascade model, each variable $t_i$ can be viewed as a transformation from the collection of variables $\{\tau_{ji}\}_{(j,i)\in \Ecal}$. More specifically, let $\Qcal_i$ be the collection of directed paths in $\Gcal$ from the source nodes to node
$i$, where each path $q\in \Qcal_i$ contains a sequence of directed edges $(j,l)$. Assuming all source nodes are infected at time zero, then we obtain variable $t_i$ via
\begin{align}
  t_i = g_i\rbr{\{\tau_{ji}\}_{(j,i)\in \Ecal}|\Acal} = \min_{q \in \Qcal_i} \sum\nolimits_{(j,l)\in q} \tau_{jl},
\end{align}
where the transformation $g_i(\cdot|\Acal)$ is the value of the shortest-path minimization. As a special case, we can now compute the probability of node $i$ infected before $T$ using a set of independent variables:
\begin{align}
  \Pr\cbr{t_i \leq T|\Acal} = \Pr\cbr{g_i\rbr{\{\tau_{ji}\}_{(j,i)\in \Ecal}|\Acal} \leq T}.
  \label{eq:equivalence}
\end{align}

The significance of the relation is that it allows us to transform a problem involving a sequence of dependent variables $\{t_i\}_{i\in\Vcal}$ to one with independent variables $\{\tau_{ji}\}_{(j,i)\in\Ecal}$.
Furthermore, the two perspectives are connected via the shortest path algorithm in weighted directed graph, a standard well-studied operation in graph analysis.

\subsection{Influence Estimation Problem in Continuous-Time Diffusion Networks\label{sec:influenceestimation}}
Intuitively, given a time window, the wider the spread of infection, the more influential the set of sources. We adopt the definition of influence as the expected number of infected nodes given a set of source nodes and a time window, as in previous work~\citep{GomSch12}. More formally, consider a set of source nodes $\Acal \subseteq \Vcal, |\Acal|\leq C$ which get infected at time zero. Then, given a time window $T$, a node $i$ is infected
within the time window if $t_i\leq T$. The expected number of infected nodes (or the influence) given the set of transmission functions $\cbr{f_{ji}}_{(j,i)\in\Ecal}$ can be computed as
\begin{align}
  \sigma(\Acal,T)
  = \EE\left[\sum\nolimits_{i\in \Vcal}\II\cbr{t_i\leq T|\Acal}\right]
  = \sum\nolimits_{i \in \Vcal} \Pr\cbr{t_i \leq T|\Acal},
  \label{eq:influence}
\end{align}
where $\II\cbr{\cdot}$ is the indicator function 
and the expectation is taken over the the set of \emph{dependent}
variables $\{t_i\}_{i\in\Vcal}$. By construction, $\sigma(\Acal, T)$ is a non-negative, monotonic nondecreasing submodular function in the set of source nodes shown by~\cite{GomSch12}.

Essentially, the influence estimation problem in Equation~\eq{eq:influence} is an inference problem for graphical models, where the probability of event $t_i \leq T$ given sources in $\Acal$ can be obtained by summing out the possible configuration of other variables $\{t_j\}_{j\neq i}$. That is
\begin{align}
\label{eq:influence_computation}
\Pr\{t_i \leq T|\Acal\} = \int_{0}^{\infty} \cdots \int_{t_i=0}^{T}\cdots \int_{0}^{\infty} \rbr{\prod\nolimits_{j\in \Vcal} p\rbr{t_j | \{t_l\}_{l\in\pi_j}} } \rbr{\prod\nolimits_{j \in \Vcal} dt_j},
\end{align}
which is, in general, a very challenging problem. First, the corresponding directed graphical models can contain nodes with high in-degree and high out-degree. For example, in Twitter, a user can follow dozens of other users, and another user can have hundreds of ``followers''. The tree-width corresponding to this directed graphical model can be very high, and we need to perform integration for functions involving many continuous variables. Second, the integral in general can not be evaluated analytically for heterogeneous transmission functions, which means that we need to resort to numerical integration by discretizing the domain $[0, \infty)$. If we use $N$ levels of discretization for each variable, we would need to enumerate $\Ocal(N^{|\pi_i|})$ entries, exponential in the number of parents.

Only in very special cases, can one derive the closed-form equation for computing $\Pr\{t_i \leq T|\Acal\}$. For instance,~\cite{GomSch12} proposed an approach for exponential transmission functions, where the special properties of exponential density are used to map the problem into a continuous time Markov process problem, and the computation can be carried out via a matrix exponential. However, without further heuristic approximation, the computational complexity of the algorithm is exponential in the size and density of the network. The intrinsic complexity of the problem entails the utilization of approximation algorithms, such as mean field algorithms or message passing algorithms. We will design an efficient randomized (or sampling) algorithm in the next section.

\subsection{Efficient Influence Estimation in Continuous-Time Diffusion Networks\label{sec:randomized}}
Our first key observation is that we can transform the influence estimation problem in Equation~\eq{eq:influence} into a problem with \emph{independent} variables. With the relation in Equation~\eq{eq:equivalence}, we can derive the influence function as
\begin{align}
 \label{eq:key1}
  \sigma(\Acal,T)
  = &\sum\nolimits_{i \in \Vcal} \Pr\cbr{g_i\rbr{\{\tau_{ji}\}_{(j,i)\in\Ecal}|\Acal}\leq T}\nonumber\\
  = &\EE \sbr{\sum\nolimits_{i\in \Vcal} \II\cbr{g_i\rbr{\{\tau_{ji}\}_{(j,i)\in\Ecal}|\Acal}\leq T}},
\end{align}
where the expectation is with respect to the set of independent variables $\{\tau_{ji}\}_{(j,i)\in\Ecal}$. This equivalent formulation suggests a naive sampling (NS) algorithm for approximating $\sigma(\Acal,T)$: draw
$n$ samples of $\{\tau_{ji}\}_{(j,i)\in\Ecal}$, run a shortest path algorithm for each sample, and finally average the results (see Appendix~\ref{app:ns} for more details). However, this naive sampling approach has a computational complexity of $O(n C |\Vcal||\Ecal| + n C |\Vcal|^2\log|\Vcal|)$ due to the repeated calling of the shortest path algorithm. This is quadratic to the network size, and hence not scalable to millions of nodes.

Our second key observation is that for each sample $\{\tau_{ji}\}_{(j,i)\in\Ecal}$, we are only interested in the neighborhood size of the source nodes, \ie, the summation $\sum_{i\in\Vcal} \II\cbr{\cdot}$ in
Equation~\eq{eq:key1}, rather than in the individual shortest paths. Fortunately, the neighborhood size estimation problem has been studied in the theoretical computer science literature. Here, we adapt a
very efficient randomized algorithm by~\cite{Cohen1997} to our influence estimation problem.
This randomized algorithm has a computational complexity of $O(|\Ecal|\log|\Vcal|+|\Vcal|\log^2|\Vcal|)$ and it estimates the neighborhood sizes for \emph{all} possible single source node locations. Since it needs to
run once for each sample of $\{\tau_{ji}\}_{(j,i)\in\Ecal}$, we obtain an overall influence estimation algorithm with $O(n |\Ecal|\log|\Vcal|+n |\Vcal|\log^2|\Vcal|)$ computation, nearly linear in network size.
Next we will revisit Cohen's algorithm for neighborhood estimation.

\subsubsection{Randomized Algorithm for Single-Source Neighborhood-Size Estimation}

Given a fixed set of edge transmission times $\{\tau_{ji}\}_{(j,i)\in\Ecal}$ and a source node $s$, infected at time zero, the neighborhood $\Ncal(s,T)$ of a source node $s$ given a time window $T$ is the set of nodes within
distance $T$ from $s$,~\ie,~
\begin{align}
  \Ncal(s,T) = \cbr{i~\big|~g_i \rbr{\{\tau_{ji}\}_{(j,i)\in\Ecal}} \leq T,~i\in \Vcal}.
  \label{eq:neighbor}
\end{align}
Instead of estimating $\Ncal(s,T)$ directly, the algorithm will assign an exponentially distributed random label $r_i$ to each network node $i$. Then, it makes use of the fact that the minimum of a set of exponential random variables $\{r_i\}_{i \in \Ncal(s,T)}$ is still an exponential random variable, but with its parameter being equal to the total number of variables, that is, 
if each $r_i \sim \exp(-r_i)$, then the smallest label within distance $T$ from source $s$, $r_\ast:=\min_{i\in \Ncal(s,T)} r_i$, will distribute as $r_\ast \sim \exp\cbr{-|\Ncal(s,T)| r_\ast}$. Suppose we randomize over the labeling $m$ times and obtain $m$ such least labels, $\{r_\ast^u\}_{u=1}^m$. Then the neighborhood size can be estimated as
\begin{align}
  |\Ncal(s,T)| \approx \frac{m-1}{\sum_{u=1}^m r_\ast^u}.
  \label{eq:estimate_size}
\end{align}
which is shown by~\cite{Cohen1997} to be an unbiased estimator of $|\Ncal(s,T)|$.
This is an elegant relation since it allows us to transform the counting problem in~\eq{eq:neighbor} to a problem of finding the minimum random label $r_\ast$. The key question is whether we can compute the least label $r_\ast$ efficiently, given random labels $\{r_i\}_{i \in \Vcal}$ and any source node $s$.

\cite{Cohen1997} designed a modified Dijkstra'{}s algorithm (Algorithm~\ref{a1}) to construct a data structure $r_\ast(s)$, called least label list, for each node $s$ to support such query. Essentially, the
algorithm starts with the node $i$ with the smallest label $r_i$, and then it traverses in breadth-first search fashion along the reverse direction of the graph edges to find all reachable nodes. For each reachable node
$s$, the distance $d_\ast$ between $i$ and $s$, and $r_i$ are added to the end of $r_\ast(s)$. Then the algorithm moves to the node $i'$ with the second smallest label $r_{i'}$, and similarly find all reachable nodes.
For each reachable node $s$, the algorithm will compare the current distance $d_\ast$ between $i'$ and $s$ with the last recorded distance in $r_\ast(s)$. If the current distance is smaller, then the current $d_\ast$
and $r_{i'}$ are added to the end of $r_\ast(s)$. Then the algorithm move to the node with the third smallest label and so on. The algorithm is summarized in Algorithm~\ref{a1} in Appendix~\ref{app:leastlabellist}.

Algorithm~\ref{a1} returns a list $r_\ast(s)$ per node $s\in \Vcal$, which contains information about distance to the smallest reachable labels from $s$. In particular, each list contains pairs of distance and random labels,
$(d,r)$, and these pairs are ordered as
\begin{align}
  \infty > \;&d_{(1)} > d_{(2)} > \ldots > d_{(|r_\ast(s)|)} = 0 \\
  &r_{(1)} < r_{(2)} < \ldots < r_{(|r_\ast(s)|)},
\end{align}
where $\{\cdot\}_{(l)}$ denotes the $l$-th element in the list. (see Appendix \ref{app:leastlabellist} for an example).

If we want to query the smallest reachable random label $r_\ast$ for a given source $s$ and a time $T$, we only need to perform a binary search on the list for node $s$:
\begin{align}
  r_\ast = r_{(l)},~\text{where}~d_{(l-1)} > T \geq d_{(l)}.
\end{align}
Finally, to estimate $\abr{\Ncal(s,T)}$, we generate $m$~\iid~collections of random labels, run Algorithm~\ref{a1} on each collection, and obtain $m$ values $\cbr{r_\ast^u}_{u=1}^m$, which we use in Equation~\eq{eq:estimate_size}
to estimate $|\Ncal(i,T)|$.

The computational complexity of Algorithm~\ref{a1} is $O(|\Ecal|\log|\Vcal| + |\Vcal|\log^2|\Vcal|)$, with expected size of each $r_\ast(s)$ being $O(\log|\Vcal|)$. Then the expected time for querying $r_\ast$
is $O(\log\log|\Vcal|)$ using binary search. Since we need to generate $m$ set of random labels and run Algorithm~\ref{a1} $m$ times, the overall computational complexity for estimating the single-source
neighborhood size for all $s\in\Vcal$ is $O(m|\Ecal|\log|\Vcal| + m|\Vcal|\log^2|\Vcal| + m |\Vcal|\log\log|\Vcal|)$. For large-scale network, and when $m\ll \min\{|\Vcal|,|\Ecal|\}$, this randomized algorithm can
be much more efficient than approaches based on directly calculating the shortest paths.

\subsubsection{Constructing Estimation for Multiple-Source Neighborhood Size}

When we have a set of sources, $\Acal$, its neighborhood is the union of the neighborhoods of its cons\-ti\-tuent sources
\begin{align}
  \Ncal(\Acal,T) = \bigcup\nolimits_{i \in \Acal} \Ncal(i,T).
\end{align}
This is true because each source independently infects its downstream nodes. Furthermore, to calculate the least label list $r_\ast$ corresponding to $\Ncal(\Acal,T)$, we can simply reuse the least label list $r_\ast(i)$
of each individual source $i \in \Acal$. More formally,
\begin{align}
  r_\ast = \min\nolimits_{i \in \Acal}~ \min\nolimits_{j \in \Ncal(i,T)} r_j,
\end{align}
where the inner minimization can be carried out by querying $r_\ast(i)$. Similarly, after we obtain $m$ samples of $r_\ast$, we can estimate $|\Ncal(\Acal,T)|$ using Equation~\eq{eq:estimate_size}.
Importantly, very little additional work is needed when we want to calculate $r_\ast$ for a set of sources $\Acal$, and we can reuse work done for a single source. This is very different from a naive sampling approach where the sampling
process needs to be done completely anew if we increase the source set. In contrast, using the randomized algorithm, only an additional constant-time minimization over $|\Acal|$ numbers is needed.

\subsubsection{Overall Algorithm}\label{sec:infEstAlg}
So far, we have achieved efficient neighborhood size estimation of $|\Ncal(\Acal, T)|$ with respect to a given set of transmission times $\{\tau_{ji}\}_{(j,i)\in\Ecal}$. Next, we will estimate the influence by averaging over multiple sets of samples for $\{\tau_{ji}\}_{(j,i)\in\Ecal}$. More specifically, the relation from~\eq{eq:key1}
\begin{align}
  \sigma(\Acal,T)
  = \EE_{\{\tau_{ji}\}_{(j,i)\in\Ecal}} \sbr{|\Ncal(\Acal,T)|}
  = \EE_{\{\tau_{ji}\}} \EE_{\{r^1,\ldots,r^m\}|\{\tau_{ji}\}}\sbr{\frac{m-1}{\sum_{u=1}^m r_\ast^u}},
\end{align}
suggests the following overall algorithm : \\[2mm]

\hspace{-8mm}
\fbox{
  \parbox{\columnwidth}{
    \vspace{1mm}
    Continuous-Time Influence Estimation (\continmax):
    \begin{itemize}
      \item[1.] Sample $n$ sets of random transmission times
      $
        \{\tau_{ij}^l\}_{(j,i)\in\Ecal}~\sim~\prod\nolimits_{(j,i) \in \Ecal} f_{ji}(\tau_{ji}).
      $
      \item[2.] Given a set of $\{\tau_{ij}^l\}_{(j,i)\in\Ecal}$, sample $m$ sets of random labels
      $
        \{r_i^u\}_{i\in \Vcal}~\sim~\prod\nolimits_{i \in \Vcal} \exp(-r_i).
      $
      \item[3.] Estimate $\sigma(\Acal,T)$ by sample averages
      $
        \sigma(\Acal,T)
        \approx \frac{1}{n} \sum_{l=1}^n \rbr{(m-1)/\sum_{u_l=1}^m r_\ast^{u_l}}.
      $
    \end{itemize}
  }
}
\\[2mm]

What is even more important is that the number of random labels, $m$, does not need to be very large. Since the estimator for $|\Ncal(A,T)|$ is unbiased~\citep{Cohen1997}, essentially the outer-loop of averaging over $n$ samples of random transmission times further reduces the variance of the estimator in a rate of $\Ocal(1/n)$. In practice, we can use a very small $m$ (\eg, $5$ or $10$) and still achieve good results, which is also confirmed by our later experiments. Compared to~\citep{CheWanYan09}, the novel application of  Cohen's algorithm arises for estimating influence for multiple sources, which drastically reduces the computation by cleverly using the least-label list from single source. Moreover, we have the following theoretical guarantee (see Appendix~\ref{app:proof1} for the proof).

\begin{theorem}
  \vspace{-2mm}
  Draw the following number of samples for the set of random transmission times
  \begin{align}
    n \geq \frac{C \Lambda}{\epsilon^2} \log\rbr{\frac{2 |\Vcal|}{\alpha}}
  \end{align}
  where $\Lambda:= \max_{\Acal:\abr{\Acal}\leq C} 2\sigma(\Acal,T)^2 / (m-2) + 2Var(\abr{\Ncal(\Acal,T)})(m-1)/(m-2) + 2 a \epsilon /3$ and $\abr{\Ncal(\Acal,T)}\leq |\Vcal|$, and for each set of random transmission times, draw $m$ sets of random labels. Then
  $
    \abr{\widehat{\sigma}(\Acal,T) - \sigma(\Acal,T)} \leq \epsilon
  $
  uniformly for all $\Acal$ with $\abr{\Acal}\leq C$, with probability at least $1 - \alpha$.
  \vspace{-2mm}
\end{theorem}

The theorem indicates that the minimum number of samples, $n$, needed to achieve certain accuracy is related to the actual size of the influence $\sigma(\Acal, T)$, and the variance of the neighborhood size $|\Ncal(\Acal,T)|$ over the random draw of samples. The number of random labels, $m$, drawn in the inner loop of the algorithm will monotonically decrease the dependency of $n$ on $\sigma(\Acal, T)$. It suffices to draw a small number of random labels, as long as the value of $\sigma(\Acal,T)^2 / (m-2)$ matches that of $Var(\abr{\Ncal(\Acal,T)})$. Another implication is that influence at larger time window $T$ is harder to estimate, since $\sigma(\Acal,T)$ will generally be larger and hence require more random samples.

\section{Constraints of Practical Importance\label{sec:modelconstraints}}
By treating our proposed influence estimation algorithm \continmax as a building block, we can now tackle the influence maximization problem under various constraints of practical importance. Here, since \continmax can estimate the influence value of any source set with respect to any given time window $T$, the {\bf Timing Constraints} can thus be naturally satisfied. Therefore, in the following sections, we mainly focus on modeling the {\bf Multiple-Item Constraints}, the  {\bf User Constraints} and the {\bf Product Constraints}.

\subsection{Multiple-Item Constraints}
Multiple products can spread simultaneously across the same set of social entities over different diffusion channels. Since these products may have different characteristics, such as the revenue and the speed of spread, and thus may follow different diffusion dynamics, we will use multiple diffusion networks for different types of products.

Suppose we have a set of products $\Lcal$ that propagate on the same set of nodes $\Vcal$. The diffusion network for product $i$ is denoted as $\Gcal_i = (\Vcal,\Ecal_i)$. For each product $i\in\Lcal$, we search for a set of source nodes $\Rcal_i\subseteq\Vcal$ to which we can assign the product $i$ to start its campaign. We can represent the selection of $\Rcal_i$'s using an assignment matrix $A \in \{0,1\}^{|\Lcal| \times |\Vcal|}$ as follows: $A_{ij} = 1$ if $j \in \Rcal_i$ and $A_{ij} = 0$ otherwise. Based on this representation, we define a new ground set $\Ground=\Item \times \Node$ of size $N=|\Item| \times |\Node|$. Each element of $\Zcal$ corresponds to the index $(i,j)$ of an entry in the assignment matrix $A$, and selecting element $z=(i,j)$ means assigning product $i$ to user $j$ (see Figure~\ref{fig:assignmentMatrix} for an illustration). We also denote $\Zcal_{*j} = \Lcal\times\cbr{j}$ and $\Zcal_{i*} = \cbr{i}\times\Vcal$ as the $j$-th column and $i$-th row of matrix $A$, respectively. Then, under the above mentioned additional requirements, we would like to find a set of assignments $\Scal\subseteq\Zcal$ so as to maximize the following {\bf overall influence}
\begin{align} \label{eq:fs}
  f(S) = \sum_{i \in \Lcal} a_i \sigma_i(\Rcal_i, T_i),
\end{align}
where $\sigma_i(\Rcal_i, T_i)$ denote the influence of product $i$ for a given time $T_i$, $\cbr{a_i>0}$ is a set of weights reflecting the different benefits of the products 
and $\Rcal_{i} = \{j \in \Node: (i,j) \in S\}$.
\begin{figure}[t]
  \centering
  \includegraphics[width=.4\columnwidth]{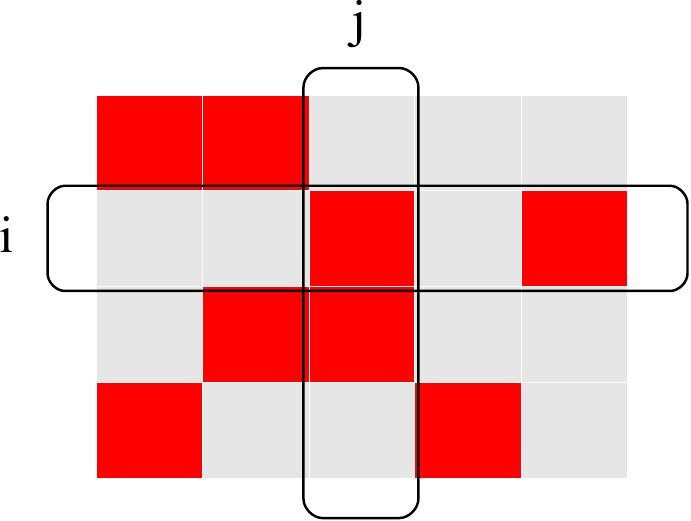}
  \caption{Illustration of the assignment matrix $A$ associated with partition matroid $\Mcal_1$ and group knapsack constraints. If product $i$ is assigned to user $j$, then $A_{ij} =1$ (colored in red).
  The ground set $\Ground$ is the set of indices of the entries in $A$, and selecting an element $(i,j) \in \Ground$ means assigning product $i$ to user $j$.
  The user constraint means that there are at most $u_j$ elements selected in the $j$-th column; the product constraint means that the total cost of the elements selected in the $i$-th row is at most $B_i$.
  } \label{fig:assignmentMatrix}
\end{figure}
We now show that the overall influence function $f(S)$ in Equation~\eq{eq:fs} is submodular over the ground set $\Ground$.

\begin{lemma}\label{lem:sub}
Under the continuous-time independent cascade model,
the overall influence $f(S)$ is a normalized monotone submodular function of $S$.
\end{lemma}
\begin{proof}
By definition, $f(\emptyset) = 0$ and $f(S)$ is monotone.
By Theorem 4 in~\cite{GomSch12}, the component influence function $\sigma_i(\Rcal_{i},T_i)$ for product $i$ is submodular in $\Rcal_{i} \subseteq \Node$.
Since non-negative linear combinations of submodular functions are still submodular, $f_i(S) := a_i\sigma_i(\Rcal_{i},T_i)$ is also submodular in $S \subseteq \Ground=\Item\times\Node$, and  $f(S) = \sum_{i\in \Item} f_i(S)$ is submodular.
\end{proof}

\subsection{User Constraints}
Each social network user can be a potential source and would like to be exposed only to a small number of ads. 
Furthermore, users may be grouped according to their geographical locations, and advertisers may have a target population they want to reach. 
Here, we will incorporate these constraints using the matroids which are combinatorial structures that generalize the notion of linear independence in matrices~\citep{Schrijver03,Fujishige05}.
Formulating our constrained influence maximization task with matroids allows us to design a greedy algorithm with provable guarantees. 

Formally, suppose that each user $j$ can be assigned to at most $\attention_j$ products. A matroid can be defined as follows:
\begin{definition}
A matroid is a pair, $\Mcal=(\Ground, \Ind)$, defined over a finite set (the ground set) $\Ground$ and a family of sets (the independent sets) $\Ind$, 
that satisfies three axioms:
\begin{enumerate}
  \item {Non-emptiness:} The empty set $\emptyset \in \Ind$.
  \item {Heredity:} If $Y \in \Ind$ and $X\subseteq Y$, then $X \in \Ind$.
  \item {Exchange:} If $X \in \Ind, Y \in \Ind$ and $|Y| > |X|$, then there exists $z \in Y\setminus X$ such that $X \cup \{z\} \in \Ind$.
\end{enumerate}
\end{definition}

An important type of matroid is the partition matroid where the ground set $\Zcal$ is partitioned into disjoint subsets $\Ground_1, \Ground_2,\dots,\Ground_t$
for some $t$ and
$$\Ind=\{S~|~S\subseteq \Ground~\text{and}~|S \cap \Ground_i|\leq u_i, \forall i=1,\dots,t \}$$
for some given parameters $u_1,\dots, u_t$.
The user constraints can then be formulated as
\begin{quote}
{Partition matroid $\Mcal_1$:} partition the ground set $\Zcal$ into $\Ground_{* j}=\Item \times\cbr{j}$ each of which corresponds to a column of  $A$. Then $\Mcal_1=\cbr{\Ground, \Ical_1}$ is
$$
  \Ical_1 = \cbr{S| S\subseteq \Ground~\text{and}~|S\cap \Ground_{* j}|\leq \attention_j,\forall j}.
$$
\end{quote}
%
%

\subsection{Product Constraints}
Seeking initial adopters entails a cost to the advertiser, which needs to be paid to the host, while the advertisers of each product have a limited amount of money. Here, we will incorporate these 
requirements using knapsack constraints which we describe below.

Formally, suppose that each product $i$ has a budget $B_i$, and assigning item $i$ to user $j$ costs $c_{ij} > 0$. 
%
%
Next, we introduce the following notation to describe product constraints over the ground set $\Ground$. 
For an element $z = (i,j) \in \Ground$, define its cost as $c(z) := c_{ij}$. Abusing the notation slightly, we denote the cost of a subset $S \subseteq \Ground$ as $c(S) := \sum_{z\in S} c(z)$. 
Then, in a feasible solution $S \subseteq \Ground$, the cost of assigning product $i$, $c(S \cap \Ground_{i*})$, should not be larger than its budget $B_i$.

Now, without loss of generality, we can assume $B_i=1$ (by normalizing $c_{ij}$ with $B_i$), and also $c_{ij} \in (0, 1]$ (by throwing away any element $(i,j)$ with $c_{ij} > 1$), and define
\begin{quote}
  {Group-knapsack:} partition the ground set into $\Ground_{i *}=\cbr{i}\times \Node$ each of which corresponds to one row of $A$. Then a feasible solution $S \subseteq \Ground$ satisfies
  $$
    c(S \cap \Ground_{i*}) \leq 1, \forall i.
  $$
\end{quote}
Importantly, these knapsack constraints have very specific structure: they are on different groups of a partition $\cbr{\Ground_{i *}}$ of the ground set and the submodular function $f(S) = \sum_i a_i\sigma_i(\Rcal_{i},T_i)$ is defined over the partition. In consequence, such structures allow us to design an efficient algorithm with improved guarantees over the known results.

\subsection{Overall Problem Formulation}
\label{sec:inf}
Based on the above discussion of various constraints in viral marketing and our design choices for tackling them, we can think of the influence maximization problem as a special case of the 
following constrained submodular maximization problem with $P=1$ matroid constraints and $k=|\Item|$ knapsack constraints,
\begin{eqnarray}
& \text{max}_{S\subseteq \Ground} & f(S)  \label{pro:infMax}\\
& \text{subject to} &  c(S \cap \Ground_{i*}) \leq 1, \quad 1 \leq i \leq k, \nonumber\\
&  &  S\in \bigcap_{p=1}^P \Ical_p ,\nonumber
\end{eqnarray}
where, for simplicity, we will denote all the feasible solutions $S\subseteq \Ground$ as $\Fcal$. This formulation in general includes the following cases of practical importance : 

\paragraph{Uniform User-Cost.} An important case of influence maximization, which we denote as the {Uniform Cost}, is that for each product $i$, all users have the same cost $c_{i*}$, \ie, $c_{ij}=c_{i*}$. 
Equivalently, each product $i$ can be assigned to at most $b_i:= \lfloor B_i / c_{i*} \rfloor$ users. Then the product constraints are simplified to
\begin{quote}
\vspace{-8mm}
 \item {Partition matroid $\Mcal_2$:} for the product constraints with uniform cost, define a matroid $\Mcal_2=\cbr{\Ground, \Ical_2}$ where 
 $$\Ical_2 = \cbr{S | S\subseteq \Ground~\text{and}~|S\cap \Ground_{i*} |\leq \budget_i, \forall i}.$$
\end{quote}
In this case, the influence maximization problem defined by Equation~\eq{pro:infMax} becomes the problem with $P=2$ matroid constraints and no knapsack constraints ($k=0$). 
%
In addition, if we assume only one product needs campaign, the formulation of Equation~\eq{pro:infMax} further reduces to the classic influence maximization problem with the simple cardinality constraint.

\paragraph{User Group Constraint.}Our formulation in Equation~\eq{pro:infMax} essentially allows for general matroids which can model more sophisticated real-world constraints, and the proposed formulation, algorithms, and analysis can still hold. For instance, suppose there is a hierarchical community structure on the users, \ie, a tree $\Tcal$ where leaves are the users and the internal nodes are communities consisting of all users underneath, such as customers in different countries around the world. In consequence of marketing strategies, on each community $C \in \Tcal$, there are at most $u_C$ slots for assigning the products. Such constraints are readily modeled by the Laminar Matroid, which generalizes the partition matroid by allowing the set $\cbr{\Ground_i}$ to be a laminar family (\ie, for any $\Ground_i \neq \Ground_j$, either $\Ground_i \subseteq \Ground_j$, or $\Ground_j \subseteq \Ground_i$, or $\Ground_i \cap \Ground_j = \emptyset$). It can be shown that the community constraints can be captured by the matroid $\Mcal=(\Ground, \Ind)$ where $\Ical = \cbr{S\subseteq \Ground: |S\cap C|\leq u_C,\forall C \in \Tcal}$. In the next section, we first present our algorithm, then provide the analysis for the uniform cost case and finally leverage such analysis for the general case.

\section{Influence Maximization\label{sec:influmaximization}}
In this section, we first develop a simple, practical and intuitive adaptive-thresholding greedy algorithm to solve the continuous-time influence maximization problem with the aforementioned constraints. Then, we provide a detailed theoretical analysis of its performance.

\subsection{Overall Algorithm}
There exist algorithms for submodular maximization under multiple knapsack constraints achieving a $1-\frac{1}{e}$ approximation factor by~\citep{Sviridenko04}. 
Thus, one may be tempted to convert the matroid constraint in the problem defined by Equation~\eq{pro:infMax} to $|\Node|$ knapsack constraints, 
so that the problem becomes a submodular maximization problem under $|\Item|+|\Node|$ knapsack constraints. 
However, this naive approach is not practical for large-scale scenarios because the running time of such algorithms is exponential in the number of knapsack 
constraints. 
%
%
Instead, if we opt for algorithms for submodular maximization under $k$ knapsack constraints and $P$ matroids constraints,
the best approximation factor achieved by polynomial time algorithms is $\frac{1}{P+2 k + 1}$~\citep{BadVon14}.
However, this is not good enough yet, since in our problem $k=|\Item|$ can be large, though $P=1$ is small.

Here, we will design an algorithm that achieves a better approximation factor by exploiting the following key observation about the structure of the problem defined by Equation~\eq{pro:infMax}: 
the knapsack constraints are over different groups $\Ground_{i*}$ of the whole ground set, and the objective function is a sum of submodular functions over these different groups.

%
%
%

The details of the algorithm, called \budgetmax, are described in Algorithm~\ref{alg:densityEnu}.
\budgetmax enumerates different values of a so-called density threshold $\rho$, runs a subroutine to find a solution for each $\rho$, which 
quantifies the cost-effectiveness of assigning a particular product to a specific user, and finally outputs the solution with the maximum objective 
value.
Intuitively, the algorithm restricts the search space to be the set of most cost-effective allocations.
The details of the subroutine to find a solution for a fixed density threshold $\rho$ are described in Algorithm~\ref{alg:greedyFixedDensity}. 
Inspired by the lazy evaluation heuristic~\citep{Leskovecetal07}, the algorithm maintains a working set $G$ and a marginal gain threshold $w_t$, 
which geometrically decreases by a factor of $1+\delta$ until it is sufficiently small to be set to zero. 
At each $w_t$, the subroutine selects each new element $z$ that satisfies the following properties:
\begin{enumerate}
\item It is feasible and the density ratio (the ratio between the marginal gain and the cost) is above the current density threshold;
\item Its marginal gain
$$
  f(z|\Greedy) := f(\Greedy \cup \{z\}) - f(\Greedy)
$$
is above the current marginal gain threshold.
\end{enumerate}
The term ``density'' comes from the knapsack problem, where the marginal gain is the mass and the cost is the volume. A large density means gaining a lot without paying much. 
In short, the algorithm considers only high-quality assignments and repeatedly selects feasible ones with marginal gain ranging from large to small.

\begin{algorithm}[t]
\caption{Density Threshold Enumeration}
\label{alg:densityEnu}
  \SetAlgoVlined
  \KwIn{parameter $\delta$; objective $f$ or its approximation $\widehat f$; assignment cost $c(z),z\in\Zcal$}

  Set $d = \max \cbr{ f(\{z\}): z \in \Ground}$\;

  \For{ $\rho \in \cbr{\frac{2d}{P+2k+1}, (1+\delta)\frac{2d}{P+2k+1}, \dots, \frac{2|\Ground|d}{P+2k+1} }$ }{Call Algorithm~\ref{alg:greedyFixedDensity} to get $S_\rho$\;}

  \KwOut{$\argmax_{S_\rho} f(S_\rho)$}
\end{algorithm}
\begin{algorithm}[h]
\caption{Adaptive Threshold Greedy for Fixed Density}
\label{alg:greedyFixedDensity}
  \SetAlgoVlined
  \KwIn{parameters $\rho$, $\delta$; objective $f$ or its approximation $\widehat f$; assignment cost $c(z),z\in\Zcal$;set of feasible solutions $\Fcal$; and $d$ from Algorithm~\ref{alg:densityEnu}.}

  Set $d_\rho = \max \cbr{f(\cbr{z}): z \in \Ground, f(\cbr{z}) \geq c(z) \rho}$\;
  Set $w_t = \frac{d_\rho}{(1+\delta)^t}$ for $t = 0,\dots, L= \argmin_i \bigl[w_i \leq \frac{\delta d}{\nGround}\bigr]$
and $ w_{L+1} = 0$\;
  Set $\Greedy= \emptyset$\;

  \For{$t=0,1,\dots,L,L+1$}{
      \For{$z \not\in \Greedy$ with $\Greedy \cup \{z\} \in \Fcal$ and $f(z|\Greedy) \geq c(z) \rho$}{
            \If{$f(z|\Greedy) \geq w_t$}{
              Set $\Greedy \leftarrow \Greedy \cup \{z\}$\;
            }
        }
  }

  \KwOut{$S_\rho=G$}
\end{algorithm}
\paragraph{Remark 1.} The traditional lazy evaluation heuristic also keeps a threshold, however, it only uses the threshold to speed up selecting the element with maximum marginal 
gain.
Instead, Algorithm~\ref{alg:greedyFixedDensity} can add multiple elements $z$ from the ground set at each threshold,
and thus reduces the number of rounds from the size of the solution to the number of thresholds $\Ocal(\frac{1}{\delta}\log \frac{N}{\delta})$.
This allows us to trade off between the runtime and the approximation ratio (refer to our theoretical guarantees in section~\ref{sec:thm}).

\paragraph{Remark 2.} Evaluating the influence of the assigned products $f$ is expensive. Therefore, we will use the randomized algorithm in Section~\ref{sec:infEstAlg} to compute 
an estimation $\widehat f(\cdot)$ of the quantity $f(\cdot)$.

\subsection{Theoretical Guarantees\label{sec:thm}}
Although our algorithm is quite intuitive, it is highly non-trivial to obtain the theoretical guarantees. For clarity, we first analyze the simpler case with uniform cost, which then provides the base 
for analyzing the general case.
\subsubsection{Uniform Cost}\label{sec:uni}
As shown at the end of Section~\ref{sec:inf}, the influence maximization, in this case, corresponds to the problem defined by Equation~\eq{pro:infMax} with $P=2$ and no knapsack constraints. 
Thus, we can simply run Algorithm~\ref{alg:greedyFixedDensity} with $\rho = 0$ to obtain a solution $G$, which is roughly a $\frac{1}{P+1}$-approximation.

\paragraph{Intuition.} The algorithm greedily selects feasible elements with sufficiently large mar\-gi\-nal gain. 
However, it is unclear whether our algorithm will find \emph{good} solutions and whether it will be \emph{robust} to noise. 
Regarding the former, one might wonder whether the algorithm will select just a few elements while many elements in the optimal solution $O$ will become infeasible and will not be selected, 
in which case the greedy solution $G$ is a poor approximation.
Regarding the latter, we only use the estimation $\widehat f$ of the influence $f$ (\ie, $|\widehat f(S) - f(S)| \leq \epsilon$ for any $S \subseteq \Ground$),
which introduces additional error to the function value.
A crucial question, which has not been addressed before~\citep{BadVon14}, is whether the adaptive threshold greedy algorithm is robust to such perturbations.

Fortunately, it turns out that the algorithm will provably select sufficiently many elements of high quality.
First, the elements selected in the optimal solution $O$ but not selected in $G$ can be partitioned into $|G|$ groups, each of which is associated with an element in $G$.
Thus, the number of elements in the groups associated with the first $t$ elements in $G$, by the property of the intersection of matroids, are bounded by $Pt$.
See Figure~\ref{fig:greedyPartition} for an illustration.
Second, the marginal gain of each element in $G$ is at least as large as that of any element in the group associated with it (up to some small error).
This means that even if the submodular function evaluation is inexact, the quality of the elements in the greedy solution is still good.
The two claims together show that the marginal gain of $O \setminus G$ is not much larger than the gain of $G$, and thus $G$ is a good approximation for the problem.

Formally, suppose we use an inexact evaluation of the influence $f$ such that $|\widehat f(S) - f(S)| \leq \epsilon$ for any $S \subseteq \Ground$, let product $i \in \Item$ spread according to 
a diffusion network $\Graph_i = (\Node, \Edge_i)$, and $i^*=\argmax_{i\in\Item}|\Edge_i|$. Then, we have:

\begin{theorem}\label{thm:infMax_uni}
Suppose $\widehat f$ is evaluated up to error $\epsilon = \delta/16$ with \continmax. 
For influence maximization with uniform cost,  Algorithm~\ref{alg:greedyFixedDensity} (with $\rho=0$) outputs a solution $G$ with
$
  f(\Greedy) \geq \frac{1-2\delta}{3} f(\Optimal)
$
in expected time $\widetilde\Ocal\left(\frac{|\Edge_{i^*}|+|\Node|}{\delta^2}  + \frac{|\Item||\Node|}{\delta^3} \right).$
\end{theorem}

The parameter $\delta$ introduces a tradeoff between the approximation guarantee and the runtime:
larger $\delta$ decreases the approximation ratio but results in fewer influence evaluations.
Moreover, the runtime has a linear dependence on the network size and the number of products to propagate (ignoring some small logarithmic terms) and,
as a consequence, the algorithm is scalable to large networks.

\paragraph{Analysis.}
Suppose $G=\{g_1,\dots, g_{|G|}\}$ in the order of selection, and let $G^t =\{\g_1, \dots, \g_t\}$.
Let $C_t$ denote all those elements in $O \setminus G$ that satisfy the following: they are still feasible before selecting the $t$-th element $g_t$ but are infeasible after selecting $g_t$.
Equivalently, $C_t$ are all those elements $j\in O \setminus G$ such that (1) $j \cup G^{t-1}$ does not violate the matroid constraints but (2) $j \cup G^{t}$ violates the matroid constraints.
In other words, we can think of $C_t$ as the optimal elements ``blocked'' by $g_t$.
Then, we proceed as follows.

\begin{figure}[!t]
    \centering

\begin{tikzpicture}
\usetikzlibrary{patterns,snakes}
\newcommand*{\BlockWidth}{0.1}%
\pgfmathsetmacro{\Radius}{\BlockWidth}%
\pgfmathsetmacro{\Xoff}{4*\BlockWidth}%
\pgfmathsetmacro{\Yoff}{20*\BlockWidth}%

\node at (-6, 0) {$G$};
\node at (-6, -\Yoff) {$O \setminus G$};

\foreach \x in {-4.5, 0, 1.5, 4.5}
\foreach \y in {0}
{
  \path [fill=red] (\x-\BlockWidth,\y-\BlockWidth) rectangle (\x+\BlockWidth, \y+\BlockWidth);
  \path [fill=blue] (\x - \Xoff, \y - \Yoff) circle (\Radius);
  \path [fill=blue] (\x + \Xoff, \y - \Yoff) circle (\Radius);
  \draw [->, thick] (\x, \y-\BlockWidth) to (\x - \Xoff, \y - \Yoff + \Radius);
  \draw [->, thick] (\x, \y-\BlockWidth) to (\x + \Xoff, \y - \Yoff + \Radius);
 }
 \node (g1) at (-4.5, 0.4) {$g_1$};   \node (C1) at (-4.5, -\Yoff-0.5) {$C_1$};
 \node (gt1) at (0, 0.4) {$g_{t-1}$};   \node (Ct1) at (0, -\Yoff-0.5) {$C_{t-1}$};
 \node (gt) at (1.5, 0.4) {$g_t$};      \node (Ct) at (1.5, -\Yoff-0.5) {$C_t$};
 \node at (4.5, 0.4) {$g_{|G|}$}; \node at (4.5, -\Yoff-0.5) {$C_{|G|}$};

  \node at (-1.5, -\Yoff/2) {$\cdots\cdots$};
  \node at (3, -\Yoff/2) {$\cdots\cdots$};

 \foreach \x in {-3}
\foreach \y in {0}
{
  \path [fill=red] (\x-\BlockWidth,\y-\BlockWidth) rectangle (\x+\BlockWidth, \y+\BlockWidth);
  \path [fill=blue] (\x, \y - \Yoff) circle (\Radius);
  \draw [->, thick] (\x, \y-\BlockWidth) to (\x, \y - \Yoff + \Radius);
 }
  \node at (-3, 0.4) {$g_2$};   \node at (-3, -\Yoff-0.5) {$C_2$};

 \draw [
    thick,
    decoration={
        brace,
        mirror,
        raise=0.6cm
    },
    decorate
] (C1.north) -- (Ct.north)
node [pos=0.5,anchor=north,yshift=-0.8cm] {$\bigcup_{i=1}^t C_i$};

 \draw [
    thick,
    decoration={
        brace,
        raise=0.1cm
    },
    decorate
] (g1.north) -- (gt.north)
node [pos=0.5,anchor=north,yshift=0.8cm] {$G^t$};

\end{tikzpicture}

    \caption{Notation for analyzing Algorithm~\ref{alg:greedyFixedDensity}.
    The elements in the greedy solution $\Greedy$ are arranged according to the order in which Algorithm~\ref{alg:greedyFixedDensity} selects them in Step 3.
    The elements in the optimal solution $O$ but not in the greedy solution $G$ are partitioned into groups $\{ C_t \}_{1 \leq t \leq |G|}$, where $C_t$ are those elements in $O\setminus G$ that are still feasible before selecting $g_t$ but are infeasible after selecting $g_t$.
    } \label{fig:greedyPartition}
\end{figure}
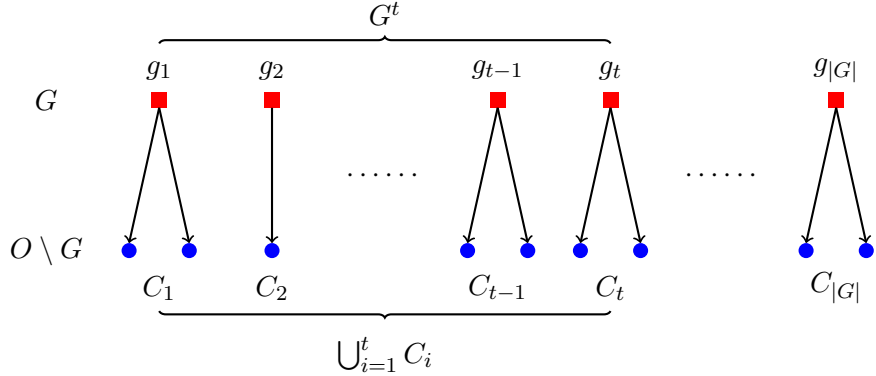

By the property of the intersection of matroids, the size of the prefix $\bigcup_{i=1}^t C_t$ is bounded by $Pt$.
As a consequence of this property, for any $Q \subseteq \Ground$, the sizes of any two maximal independent subsets $T_1$ and $T_2$ of $Q$
can only differ by a multiplicative factor at most $P$.
This can be realized with the following argument. First, note that any element $z \in T_1 \setminus T_2$, $\cbr{z} \cup T_2$ violates at least one of 
the matroid constraints since $T_2$ is maximal.
Then, let $\{ V_i \}_{1\leq i \leq P}$ denote all elements in $T_1 \setminus T_2$ that violate the $i$-th matroid, and partition $T_1 \cap T_2$ arbitrarily among 
these $V_i$'s so that they cover $T_1$.
In this construction, the size of each $V_i$ must be at most $|T_2|$, since otherwise by the Exchange axiom, there would exist $z \in V_i \setminus T_2$ 
that can be added to $T_2$, without violating the $i$-th matroid, leading to a contradiction.
Therefore, $|T_1|$ is at most $P$ times $|T_2|$. 

Next, we apply the above property as follows. Let $Q$ be the union of $G^{t}$ and $\bigcup_{i=1}^t C_t$.
On one hand, $G^{t}$ is a maximal independent subset of $Q$, since no element in $\bigcup_{i=1}^t C_t$ can be added to $G^t$ without violating the matroid constraints.
On the other hand, $\bigcup_{i=1}^t C_t$ is an independent subset of $Q$, since it is part of the optimal solution.
Therefore, $\bigcup_{i=1}^t C_t$ has size at most $P$ times $|G^t|$, which is $Pt$. Note that the properties of matroids are crucial for this analysis,
which justifies our formulation using matroids.
In summary, we have

\begin{claim}\label{cla:size}
$\sum_{i=1}^t |C_i| \leq P t$, for $t =1, \dots, |G|$.
\end{claim}

Now, we consider the marginal gain of each element in $C_t$ associated with $g_t$.
First, suppose $g_t$ is selected at the threshold $\tau_t>0$. Then, any $j \in C_t$ has marginal gain bounded by $(1+\delta)\tau_t + 2\epsilon$,
since otherwise $j$ would have been selected at a larger threshold before $\tau_t$ by the greedy criterion.
Second, suppose $g_t$ is selected at the threshold $w_{L+1}=0$. Then, any $j \in C_t$ has marginal gain approximately bounded by $\frac{\delta}{\nGround} d$.
Since the greedy algorithm must pick $g_1$ with $\widehat f(g_1) = d$ and $d \leq f(g_1)+\epsilon$, any $j \in C_t$ has marginal gain bounded by 
$\frac{\delta}{\nGround} f(G) + \Ocal(\epsilon)$.
Putting everything together we have:
\begin{claim}\label{cla:gain}
Suppose $g_t$ is selected at the threshold $\tau_t$. Then $f(j|G^{t-1}) \leq (1+\delta) \tau_t + 4\epsilon + \frac{\delta}{\nGround} f(G)$ for any $j \in C_t$.
\end{claim}
Since the evaluation of the marginal gain of $g_t$ should be at least $\tau_t$, this claims essentially indicates that the marginal gain of $j$ is approximately bounded 
by that of $g_t$.

Since there are not many elements in $C_t$ (Claim~\ref{cla:size}) and the marginal gain of each of its elements is not much larger than that of 
$g_t$ (Claim~\ref{cla:gain}), we can conclude that the marginal gain of $O \setminus G = \bigcup_{i=1}^{|G|} C_t$ is not much larger than that of $G$, which is just $f(G)$. 
\begin{claim}\label{cla:com}
The marginal gain of $O\setminus G$ satisfies
$$\sum_{j \in O\setminus G} f(j|\Greedy) \leq [(1+\delta) P + \delta] f(G)  + (6+2\delta) \epsilon P |G|. $$
\end{claim}

Finally, since by submodularity, $f(O) \leq f(O\cup G) \leq f(G) + \sum_{j \in O\setminus G} f(j|\Greedy)$, Claim~\ref{cla:com} shows that $f(G)$ is close to $f(O)$ up to a multiplicative 
factor roughly $(1+P)$ and additive factor $\Ocal(\epsilon P |G|)$.
Given that $f(G) > |G|$, it leads to roughly a $1/3$-approximation for our influence maximization problem by setting $\epsilon = \delta/16$ when evaluating $\widehat f$ with \continmax. 
Combining the above analysis and the runtime of the influence estimation algorithm, we have our final guarantee in Theorem~\ref{thm:infMax_uni}. Appendix~\ref{sec:influmaximization:proof:uniform} presents the complete proofs. 

\subsubsection{General Case}\label{sec:nonuni}
In this section, we consider the general case, in which users may have different associated costs. 
Recall that this case corresponds to the problem defined by Equation~\eq{pro:infMax} with $P=1$ matroid constraints and $k=|\Item|$ group-knapsack constraints. 
Here, we will show that there is a step in Algorithm~\ref{alg:densityEnu} which outputs a solution $S_\rho$ that is a good approximation.

\paragraph{Intuition.} The key idea behind Algorithm~\ref{alg:densityEnu} and Algorithm~\ref{alg:greedyFixedDensity} is simple:
spend the budgets \emph{efficiently} and spend them \emph{as much as possible}.
By spending them efficiently, we mean to only select those elements whose density ratio between the marginal gain and the cost is above the threshold $\rho$.
That is, we assign product $i$ to user $j$ only if the assignment leads to large marginal gain without paying too much.
By spending the budgets as much as possible, we mean to stop assigning product $i$ only if its budget is almost exhausted or
no more assignments are possible without violating the matroid constraints.
Here we make use of the special structure of the knapsack constraints on the budgets:
each constraint is only related to the assignment of the corresponding product and its budget,
so that when the budget of one product is exhausted, it does not affect the assignment of the other products.
In the language of submodular optimization, the knapsack constraints are on a partition $\Ground_{i*}$ of the ground set
and the objective function is a sum of submodular functions over the partition.

However, there seems to be a hidden contradiction between spending the budgets effi\-cient\-ly and spending them as much as possible.
On one hand, efficiency means the density ratio should be large, so the threshold $\rho$ should be large;
on the other hand, if $\rho$ is large, there are just a few elements that can be considered, and thus the budget 
might not be exhausted.
After all, if we set $\rho$ to be even larger than the maximum possible value, then no element is considered and no gain is achieved.
In the other extreme, if we set $\rho=0$ and consider all the elements, then a few elements with large costs may be selected, exhausting all the budgets and leading to a poor solution.

Fortunately, there exists a suitable threshold $\rho$ that achieves a good tradeoff between the two and leads to a good approximation.
On one hand, the threshold is sufficiently small, so that the optimal elements we abandon (\ie, those with low-density ratio) have a total gain at most a fraction of the optimum;
on the other hand, it is also sufficiently large, so that the elements selected are of high quality (\ie, of high-density ratio), and we achieve sufficient gain even if the budgets of some items 
are exhausted. 
\begin{theorem}\label{thm:infMax}
Suppose $\widehat f$ is evaluated up to error $\epsilon = \delta/16$ with \continmax. 
In Algorithm~\ref{alg:densityEnu}, there exists a $\rho$ such that
$$
    f(S_\rho) \geq  \frac{\max\cbr{k_a, 1} }{(2|\Item|+2) (1+3\delta)} f(O)
$$
where $k_a$ is the number of active knapsack constraints:
$$
 k_a = \left| \{i: S_\rho  \cup \{z\} \not\in \Fcal,  \forall z \in \Ground_{i*}\} \right|.
$$
The expected running time is $\widetilde\Ocal\left(\frac{|\Edge_{i^*}|+|\Node|}{\delta^2}  + \frac{|\Item||\Node|}{\delta^4} \right).$
\end{theorem}

Importantly, the approximation factor improves over the best known guarantee $\frac{1}{P+2k + 1} = \frac{1}{2|\Item| + 2}$ for efficiently maximizing submodular functions over $P$ matroids and 
$k$ general knapsack constraints. Moreover, since the runtime has a linear dependence on the network size, the algorithm easily scales to large networks.
As in the uniform cost case, the parameter $\delta$ introduces a tradeoff between the approximation and the runtime.

\paragraph{Analysis.}
The analysis follows the intuition. 
Pick $\rho = \frac{2 f(O)}{P+2k+1}$, where $O$ is the optimal solution, and define
\begin{align*}
    O_- & := \cbr{z \in O\setminus S_\rho: f(z| S_\rho) < c(z) \rho + 2\epsilon }, \\
    O_+ & := \cbr{z \in O\setminus S_\rho: z\not\in O_-}.
\end{align*}
Note that, by submodularity, $O_-$ is a superset of the elements in the optimal solution that we abandon due to the density threshold and, 
by construction, its marginal gain is small:
$$
    f(O_- | S_\rho) \leq \rho c(O_-) + \Ocal(\epsilon |S_\rho|) \leq k\rho + \Ocal(\epsilon |S_\rho|),
$$
where the small additive term $\Ocal(\epsilon |S_\rho|)$ is due to inexact function evaluations. Next, we proceed as follows.

First, if no knapsack constraints are active, then the algorithm runs as if there were no knapsack constraints (but only on elements with density ratio above 
$\rho$).
Therefore, we can apply the same argument as in the case of uniform cost (refer to the analysis up to Claim~\ref{cla:com} in Section~\ref{sec:uni});
the only caveat is that we apply the argument to $O_+$ instead of $O\setminus S_\rho$.
Formally, similar to Claim~\ref{cla:com}, the marginal gain of $O_+$ satisfies
%
$$
    f(O_+| S_\rho) \leq [(1+\delta)P+\delta ] f(S_\rho) + \Ocal(\epsilon P |S_\rho|),
$$
where the small additive term $\Ocal(\epsilon P |S_\rho|)$ is due to inexact function evaluations. Using that $f(O) \leq f(S_\rho) + f(O_-|S_\rho) + f(O_+ |S_\rho)$,
we can conclude that $S_\rho$ is roughly a $\frac{1}{P+ 2k + 1}$-approximation.

Second, suppose $k_a >0$ knapsack constraints are active and the algorithm discovers that the budget of product $i$ is exhausted
when trying to add element $z$ to the set $G_i = \Greedy \cap \Ground_{i*}$ of selected elements at that time. 
Since $c(G_i \cup \cbr{z}) > 1$ and each of these elements has density above $\rho$, the gain of $G_i \cup \cbr{z}$ is above $\rho$.
However, only $G_i$ is included in our final solution, so we need to show that the marginal gain of $z$ is not large compared to that of $G_i$.
To do so, we first realize that the algorithm greedily selects elements with marginal gain above a decreasing threshold $w_t$.
Then, since $z$ is the last element selected and $G_i$ is nonempty (otherwise adding $z$ will not exhaust the budget), the marginal 
gain of $z$ must be bounded by roughly that of $G_i$, which is at least roughly $\frac{1}{2}\rho$.
%
%
Since this holds for all active knapsack constraints, then the solution has value at least $\frac{k_a }{2}\rho$, which is an $\frac{k_a }{P+2k+1}$-approximation.

Finally, combining both cases, and setting $k=|\Item|$ and $P=1$ as in our problem, we have our final guarantee in Theorem~\ref{thm:infMax}. Appendix~\ref{sec:influmaximization:proof:general} presents the complete proofs. 

\section{Experiments on Synthetic and Real Data\label{sec:exp}}

In this section, we first evaluate the accuracy of the estimated influence given by \continmax and then investigate the performance of influence maximization on synthetic and real networks by incorporating \continmax into the framework of \budgetmax. We show that our approach significantly outperforms the state-of-the-art methods in terms of both speed and solution quality.

\begin{figure}
\centering
\renewcommand{\tabcolsep}{0pt}
\begin{tabular}{cccc}
\rotatebox{90}{\bf\small~~~~~~~{Core-perphery}}~
&
\begin{subfigure}[t]{0.31\textwidth}
  	\includegraphics[width=\textwidth]{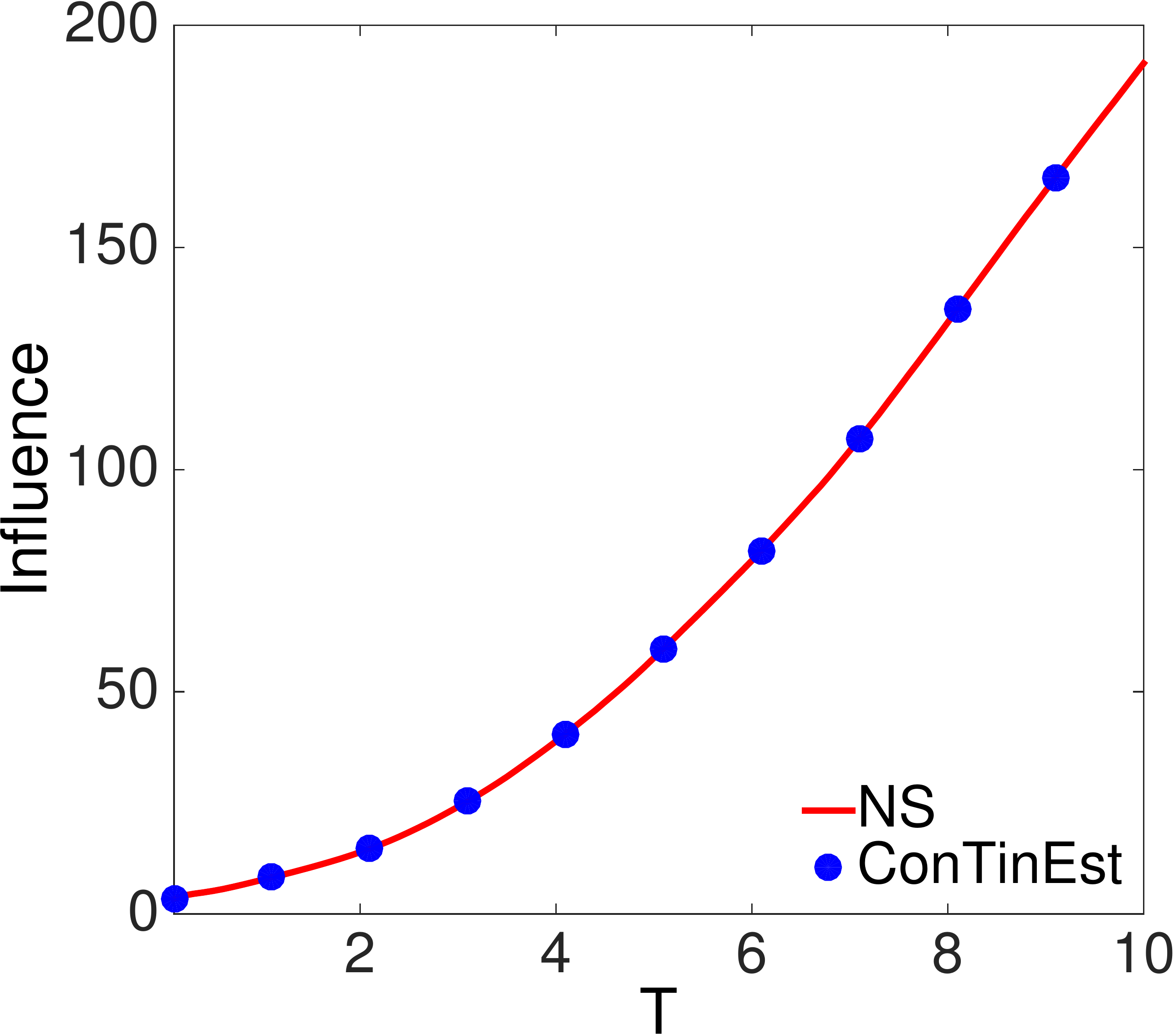} 
\end{subfigure} 
&
\begin{subfigure}[t]{0.31\textwidth}
  	\includegraphics[width=\textwidth]{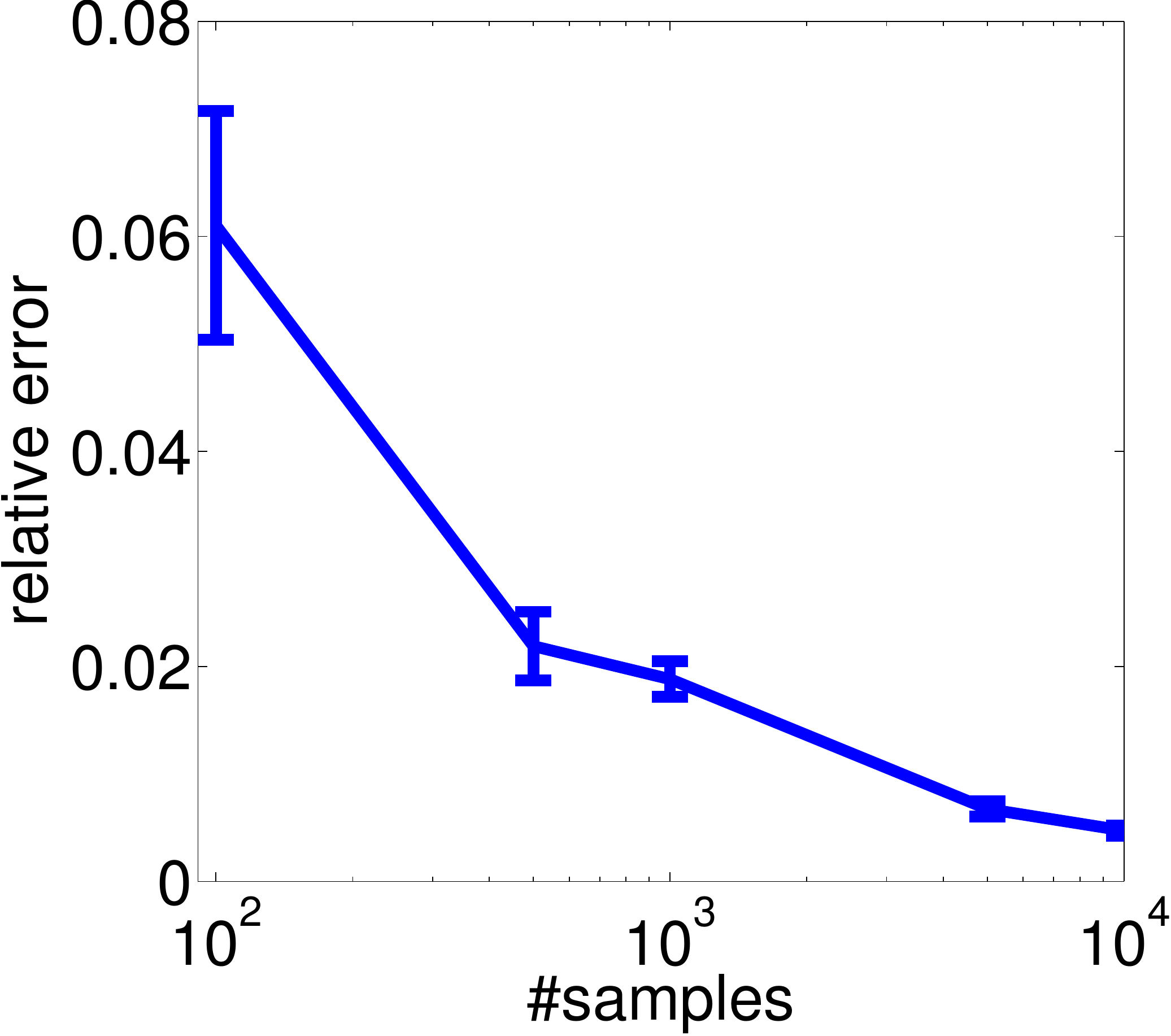} 
\end{subfigure} 
& 
\begin{subfigure}[t]{0.31\textwidth}
  	\includegraphics[width=\textwidth]{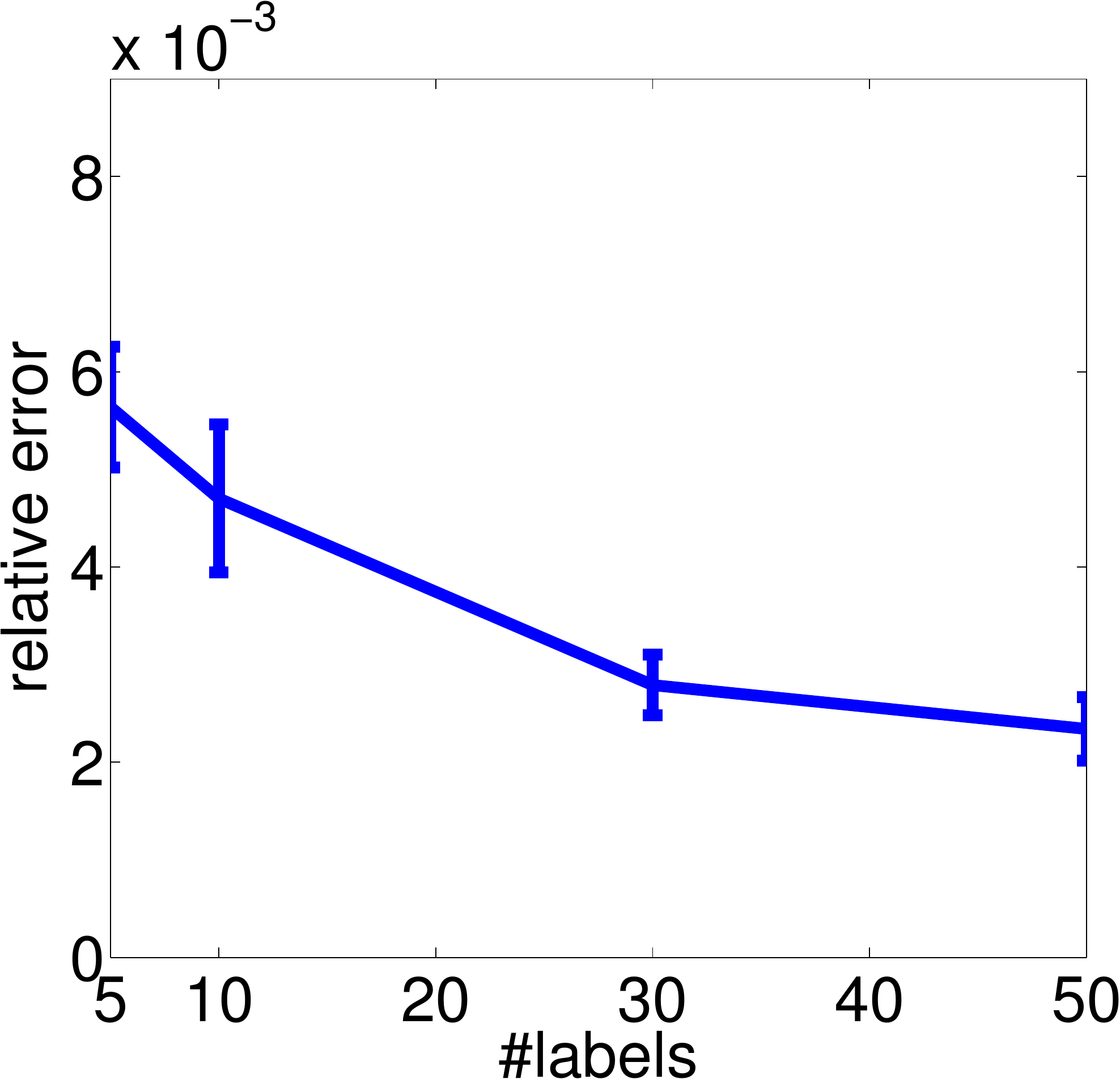} 
\end{subfigure} 
\\
\rotatebox{90}{\bf\small~~~~~~~~~~~~{Random}}~
&
\begin{subfigure}[t]{0.31\textwidth}
  	\includegraphics[width=\textwidth]{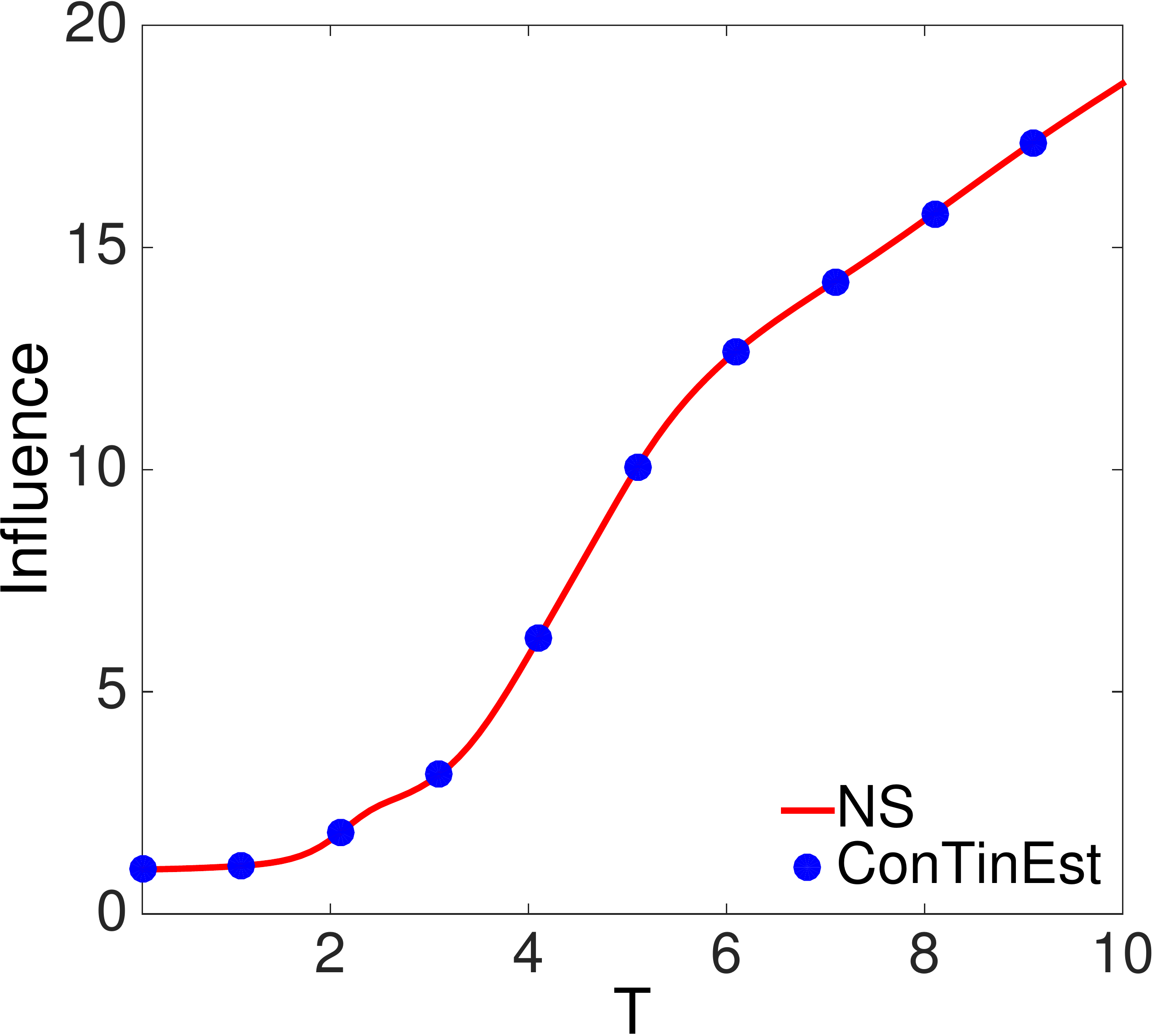} 
\end{subfigure} 
&
\begin{subfigure}[t]{0.31\textwidth}
  	\includegraphics[width=\textwidth]{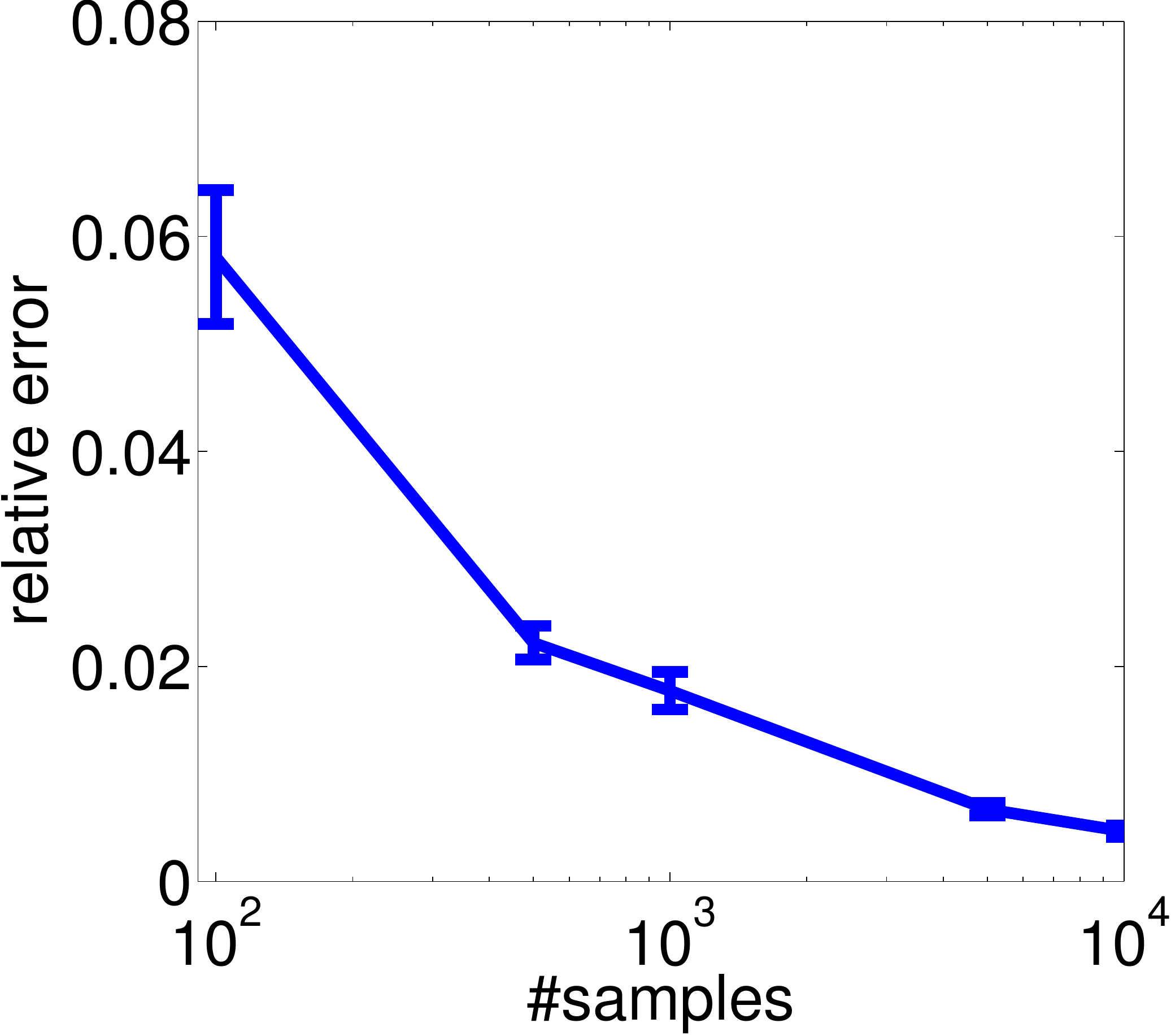} 
\end{subfigure} 	
&
\begin{subfigure}[t]{0.31\textwidth}
  	\includegraphics[width=\textwidth]{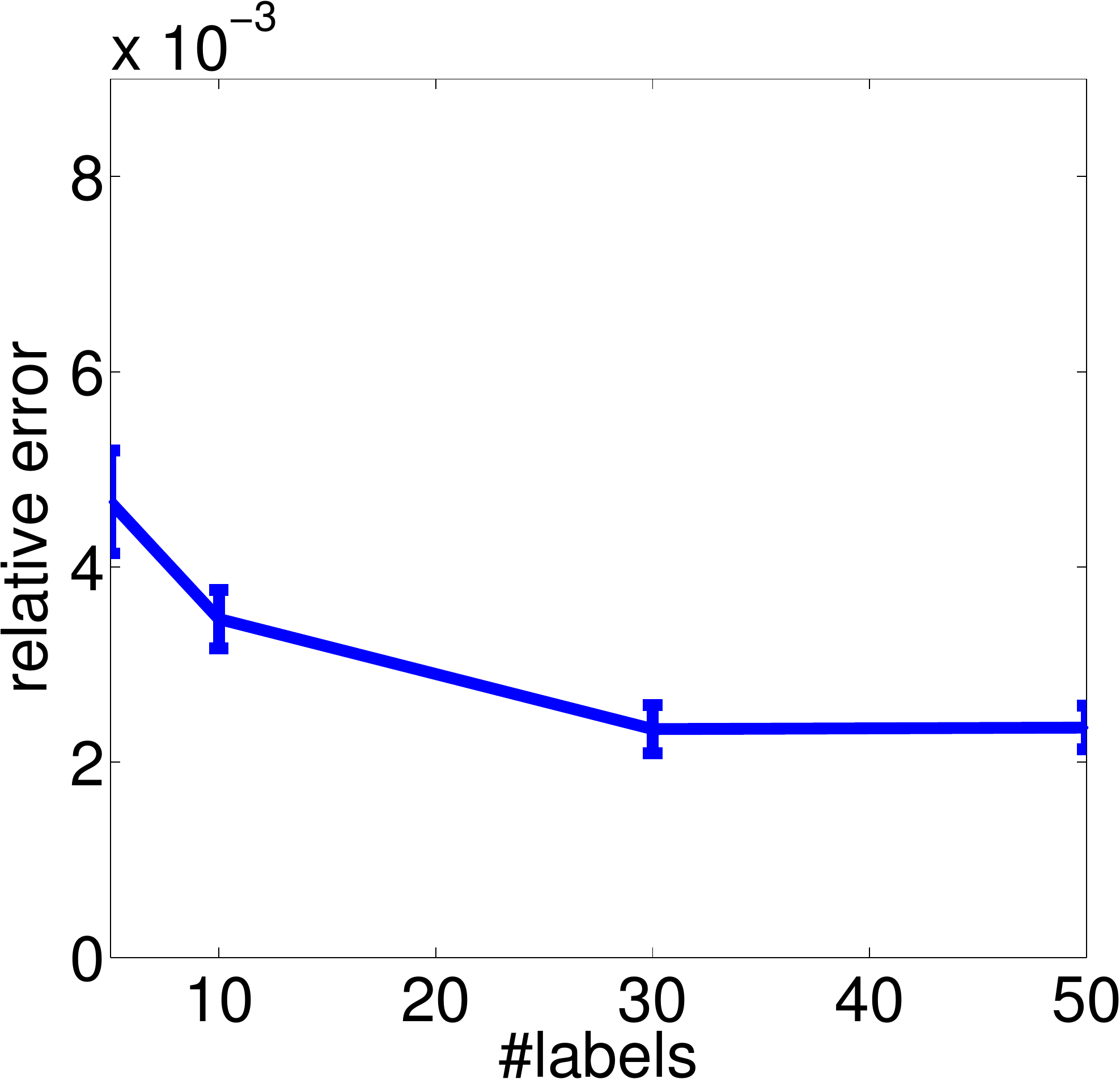} 
\end{subfigure} 
\\
\rotatebox{90}{\bf\small~~~~~~~~~~~{Hierarchy}}~
&
\begin{subfigure}[t]{0.31\textwidth}
  	\includegraphics[width=\textwidth]{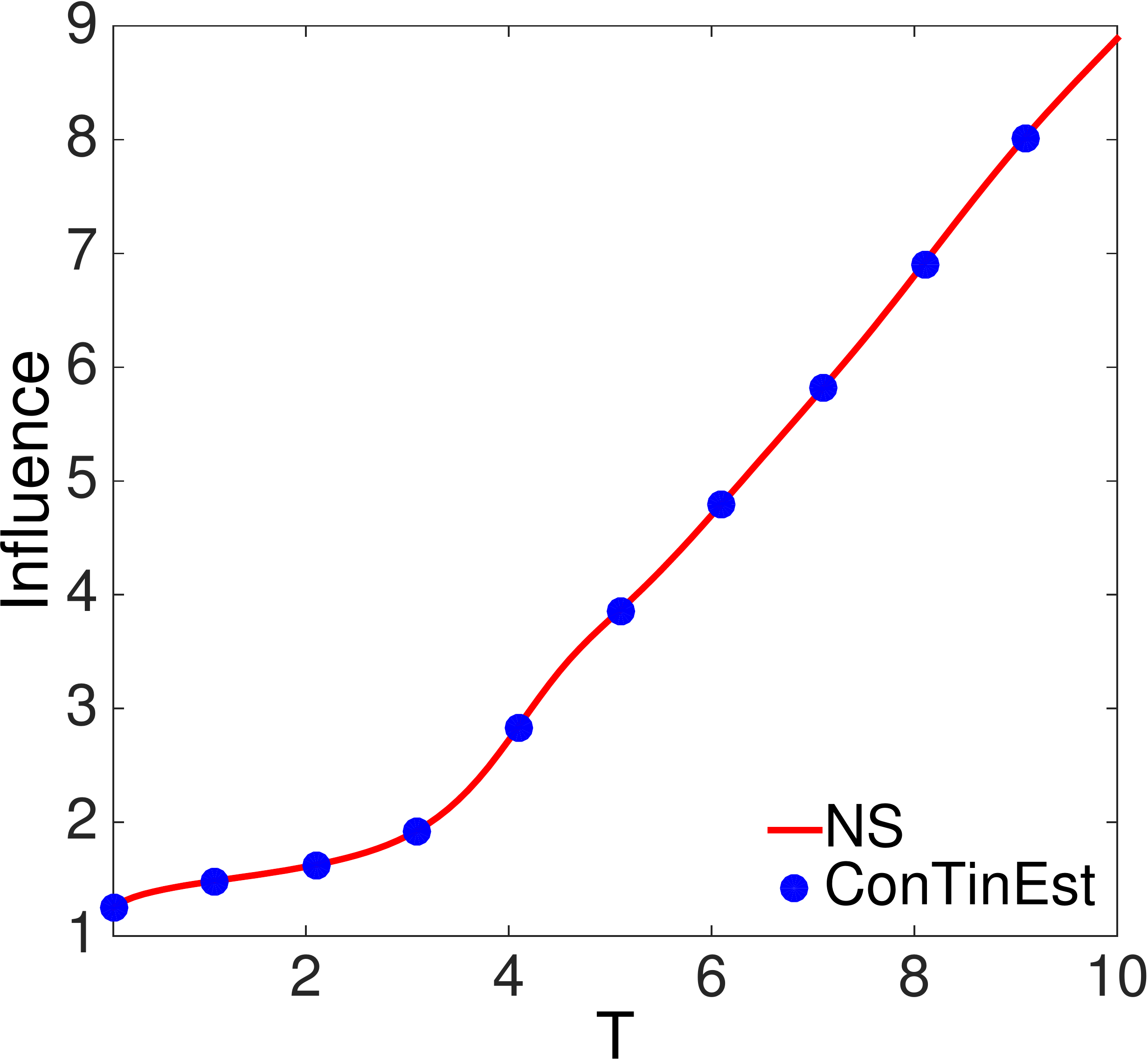} 
\end{subfigure} 
&
\begin{subfigure}[t]{0.31\textwidth}
  	\includegraphics[width=\textwidth]{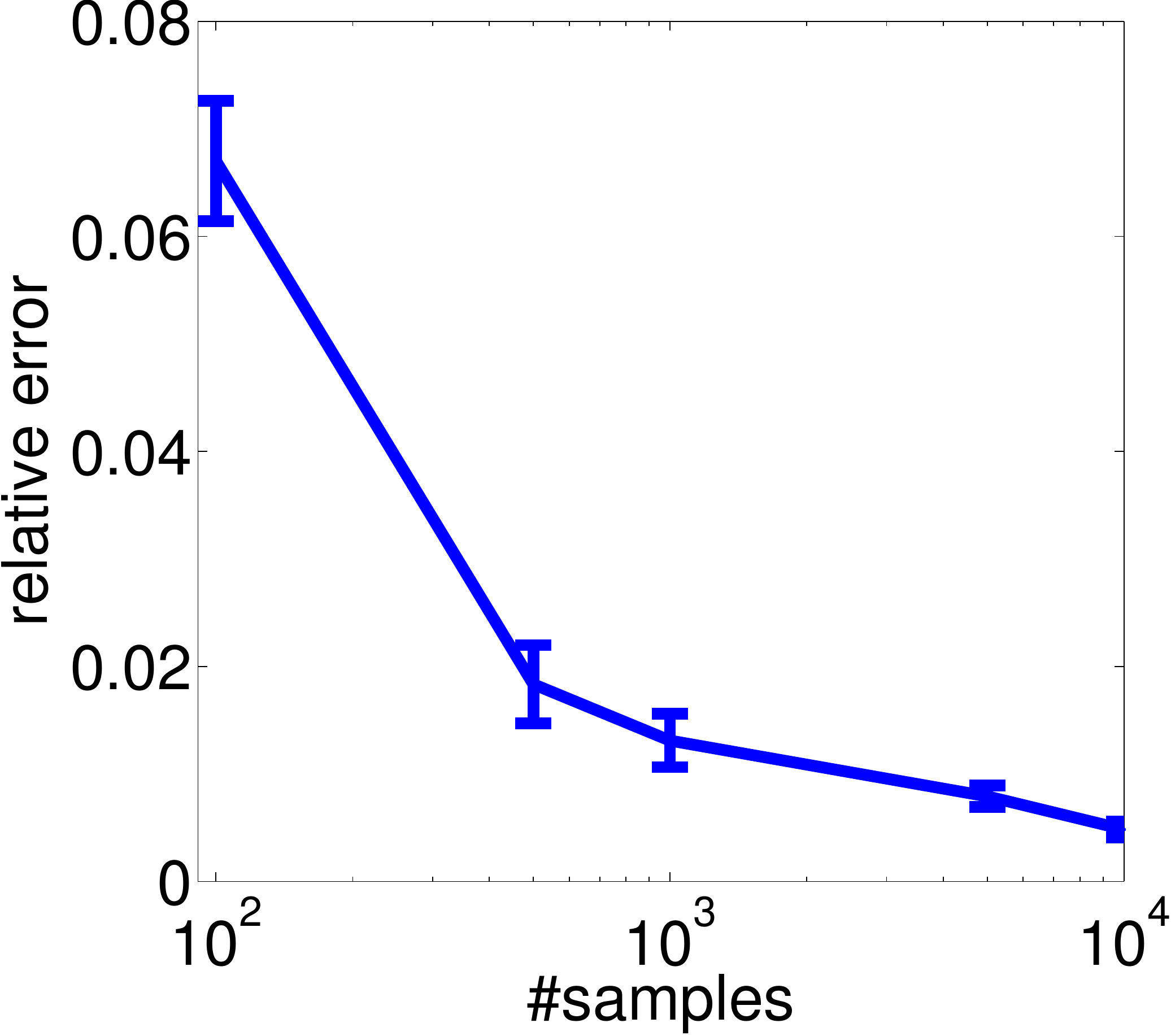} 
\end{subfigure}	
&
\begin{subfigure}[t]{0.31\textwidth}
  	\includegraphics[width=\textwidth]{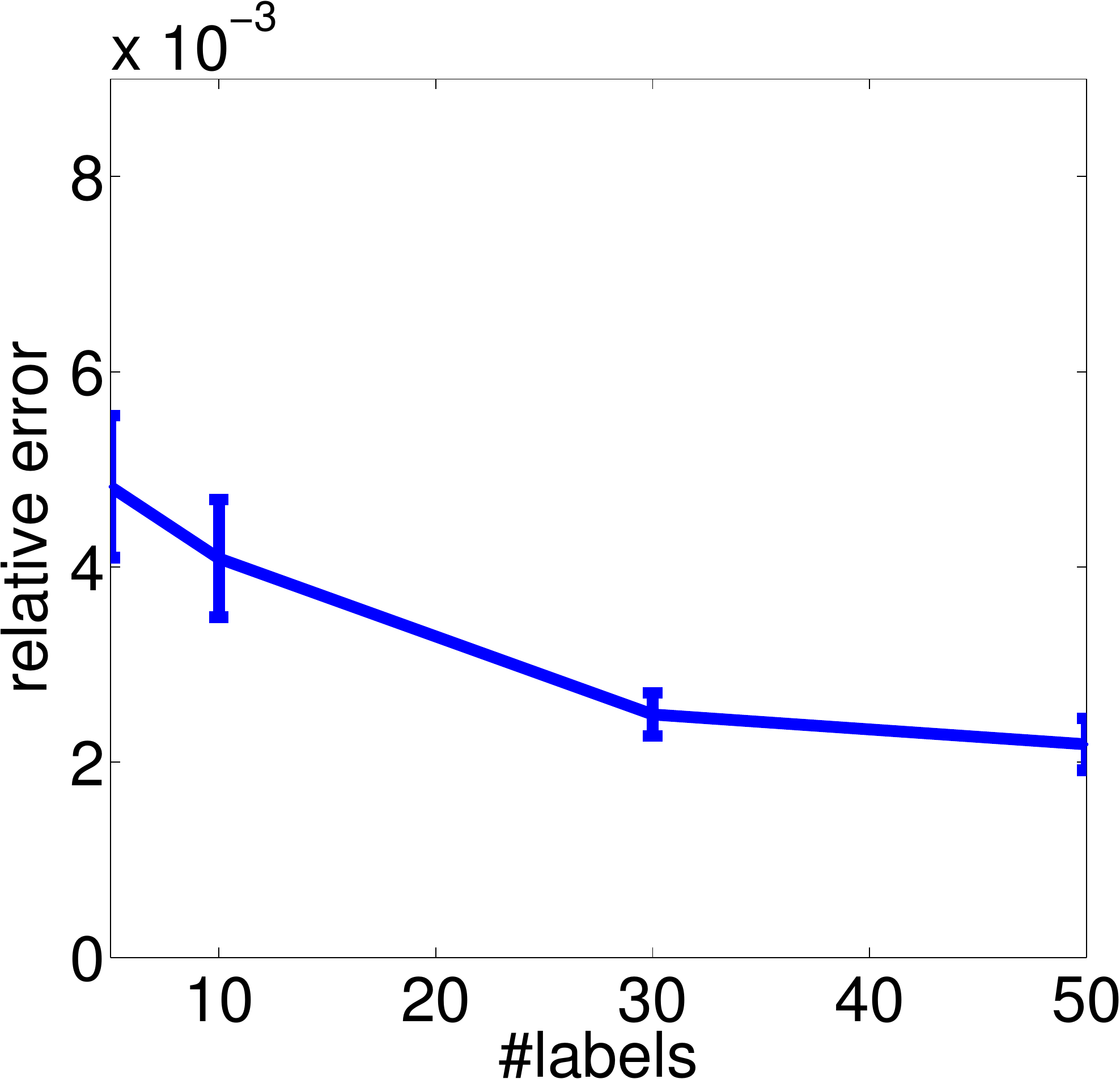} 
\end{subfigure} 
\\
& (a) Influence vs. time & (b) Influence vs. \#samples & (c) Error vs. \#labels
\end{tabular}

 \caption{Influence estimation for core-periphery, random, and hierarchical networks with 1,024 nodes and 2,048 edges.
Column (a) shows estimated influence by NS (near ground truth), and \continmax for increasing time window $T$; Column (b) shows \continmax'{}s relative error against
 number of samples with 5 random labels and $T = 10$; Column (c) reports \continmax'{}s relative error against the number of random labels with 10,000
 random samples and $T = 10$.\label{accuracy_wbl}} 
\end{figure}

\subsection{Experiments on Synthetic Data}
We generate three types of Kronecker
networks~\citep{LesChaKleFaletal10} which are synthetic networks generated by a recursive Kronecker product of a base 2-by-2 parameter matrix with itself to generate self-similar graphs. By tuning the base parameter matrix, we are able to generate the Kronecker networks which can mimic different structural properties of many real networks. In the following, we consider networks of three different types of structures: (\emph{i}) core-periphery networks (parameter matrix: [0.9 0.5; 0.5 0.3]), which mimic the information diffusion traces in real-world networks~\citep{GomBalSch11},
(\emph{ii}) random networks ([0.5 0.5; 0.5 0.5]), typically used in physics and graph theory~\citep{EasKle10}
and (\emph{iii}) hierarchical networks ([0.9 0.1; 0.1 0.9])~\citep{ClaMooNew08}. Next, we assign a pairwise transmission function for every directed edge in each type of network and set its parameters at random. 
In our experiments, we use the Weibull distribution from~\citep{AalBorGje08},
\begin{align}
f(t;\alpha, \beta)=\frac{\beta}{\alpha}\rbr{\frac{t}{\alpha}}^{\beta - 1}e^{-(t/\alpha)^{\beta}}, t\geq 0,
\end{align}
where $\alpha>0$ is a scale parameter and $\beta>0$ is a shape parameter.  The Weibull distribution (Wbl) has often been used to model lifetime events in survival analysis, providing more flexibility than an
exponential distribution.
We choose $\alpha$ and $\beta$ from 0 to 10 uniformly at random for each edge in order to have heterogeneous temporal dynamics. Finally, for each type of Kronecker network, we  generate 10 sample networks, each of which has different $\alpha$ and $\beta$ chosen for every edge.

\subsubsection{Influence Estimation}
To the best of our knowledge, there is no analytical solution to the influence estimation given Weibull transmission function. Therefore, we compare \continmax with the Naive Sampling (NS) approach by considering the highest degree node in a network as the source, and draw 1,000,000 samples for NS to obtain near ground truth. In Figure~\ref{accuracy_wbl}, Column (a) compares \continmax with the ground truth provided by NS at different time window $T$, from $0.1$ to $10$ in networks of different structures. For \continmax, we generate up to 10,000 random samples (or sets of random waiting times), and 5 random labels in the inner loop. In all three networks, estimation provided by \continmax fits the ground truth accurately, and the relative error decreases quickly as we increase the number of samples and labels (Column (b) and~Column (c)). For 10,000 random samples with 5 random labels, the relative error is smaller than 0.01.
\begin{figure}[t]
\renewcommand{\tabcolsep}{1pt}
\begin{tabular}{cccc}
\rotatebox{90}{\bf\small~~~~~{Influence by Size}}~
&
\begin{subfigure}[t]{0.31\textwidth}
  	\includegraphics[width=\textwidth]{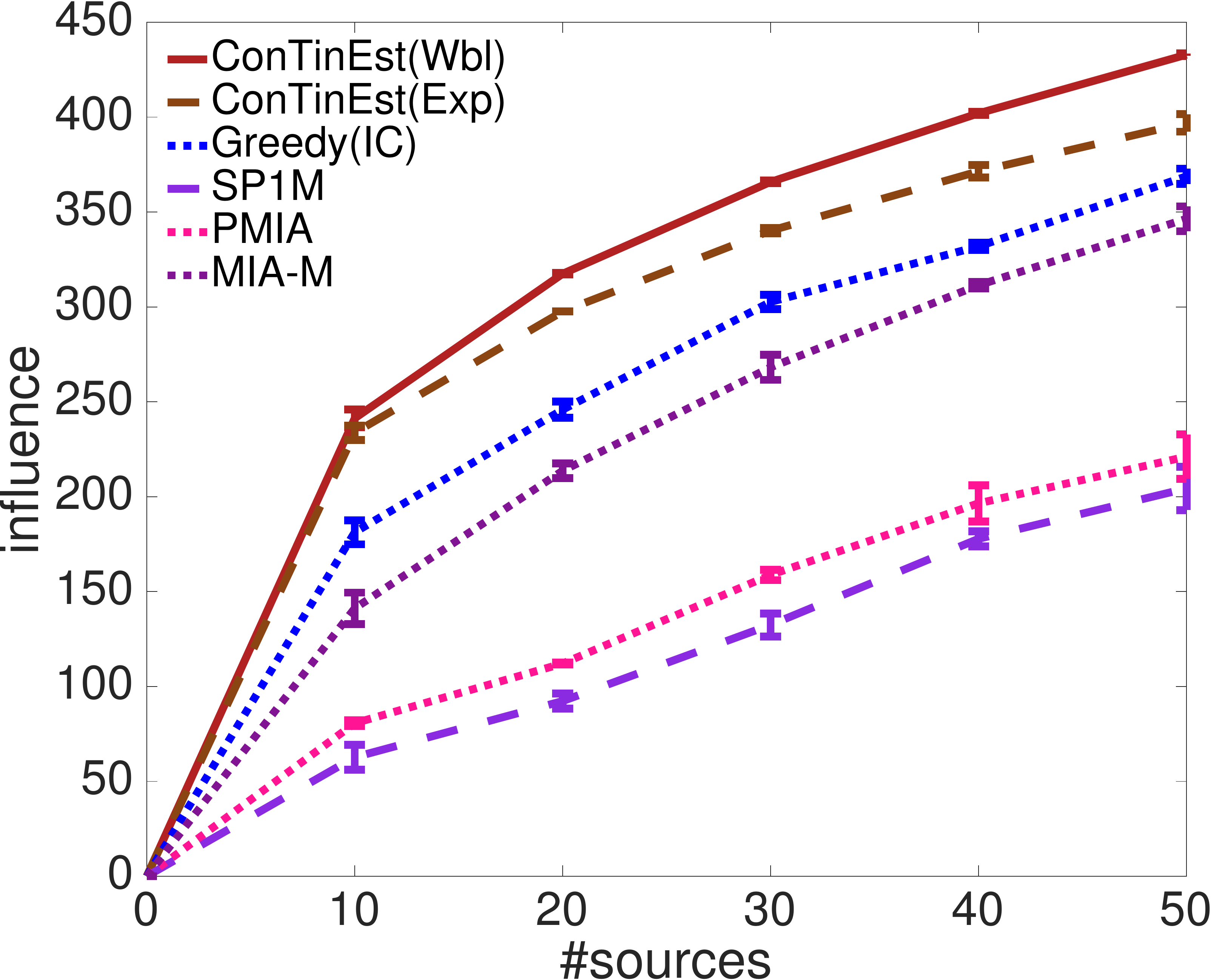} 
\end{subfigure} 
&
\begin{subfigure}[t]{0.31\textwidth}
  	\includegraphics[width=\textwidth]{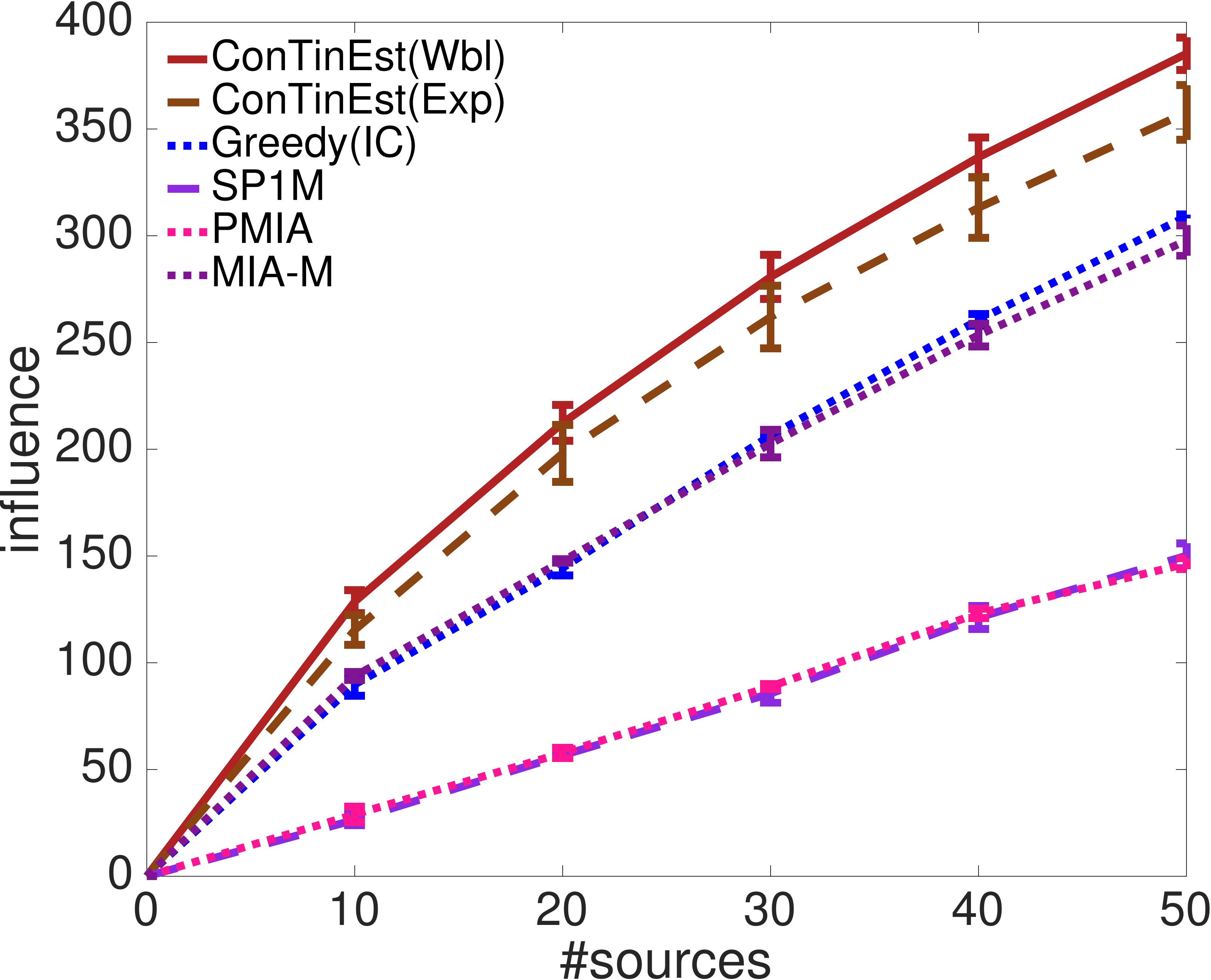} 
\end{subfigure} 	
& 
\begin{subfigure}[t]{0.31\textwidth}
  	\includegraphics[width=\textwidth]{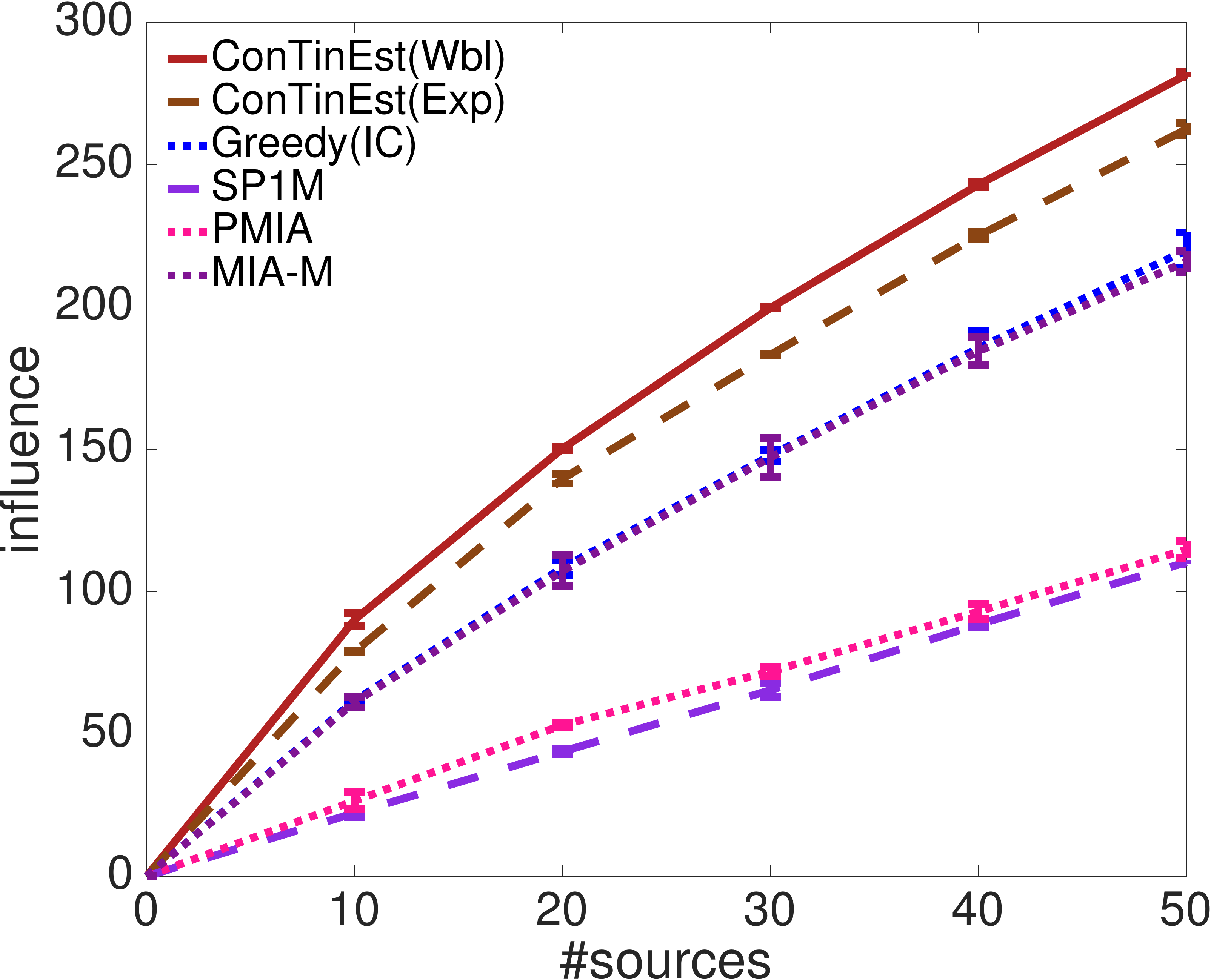} 
\end{subfigure} 
\\
\rotatebox{90}{\bf\small~~~{Influence by Time}}~
&
\begin{subfigure}[t]{0.31\textwidth}
  	\includegraphics[width=\textwidth]{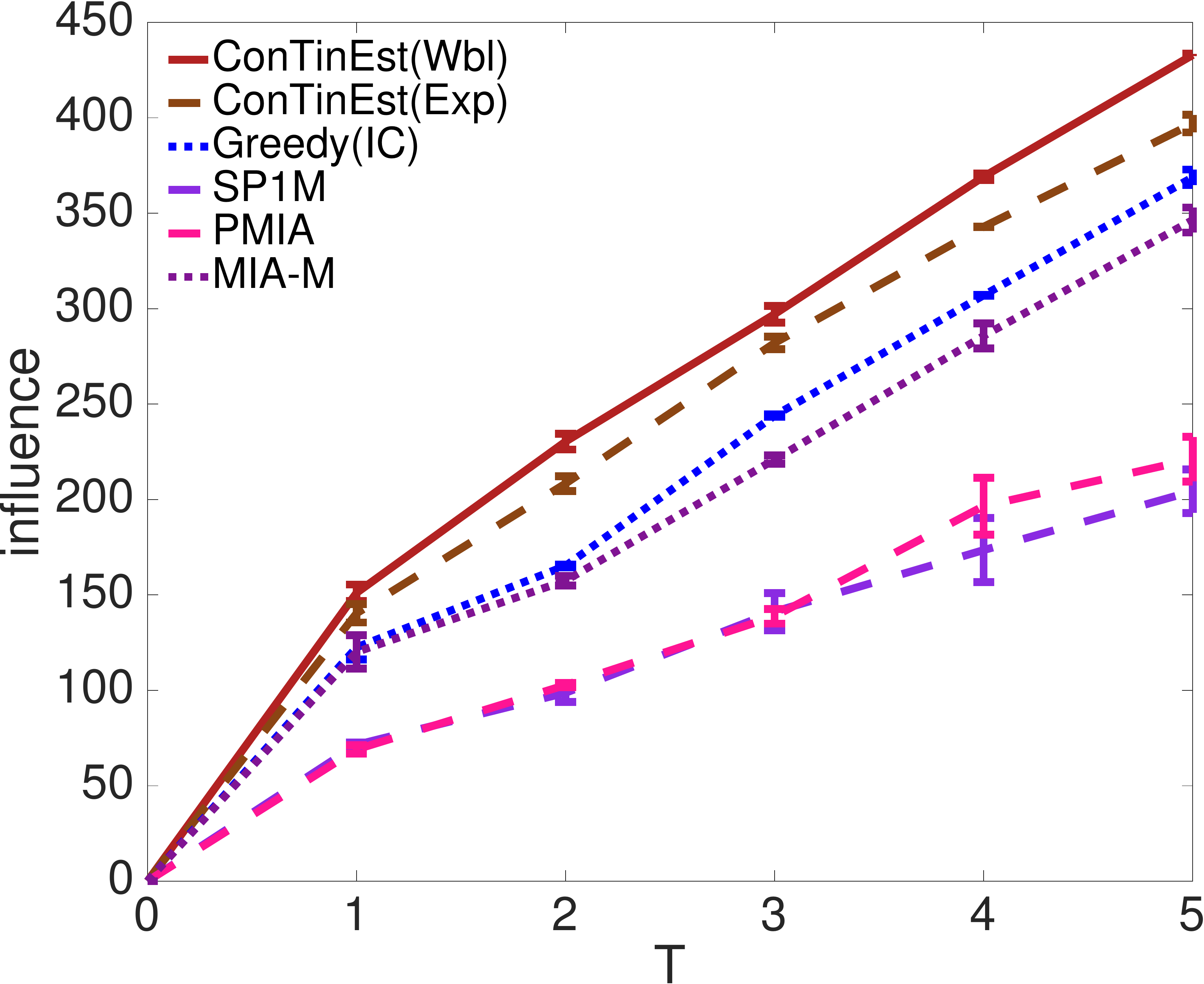} 
\end{subfigure} 
&
\begin{subfigure}[t]{0.31\textwidth}
  	\includegraphics[width=\textwidth]{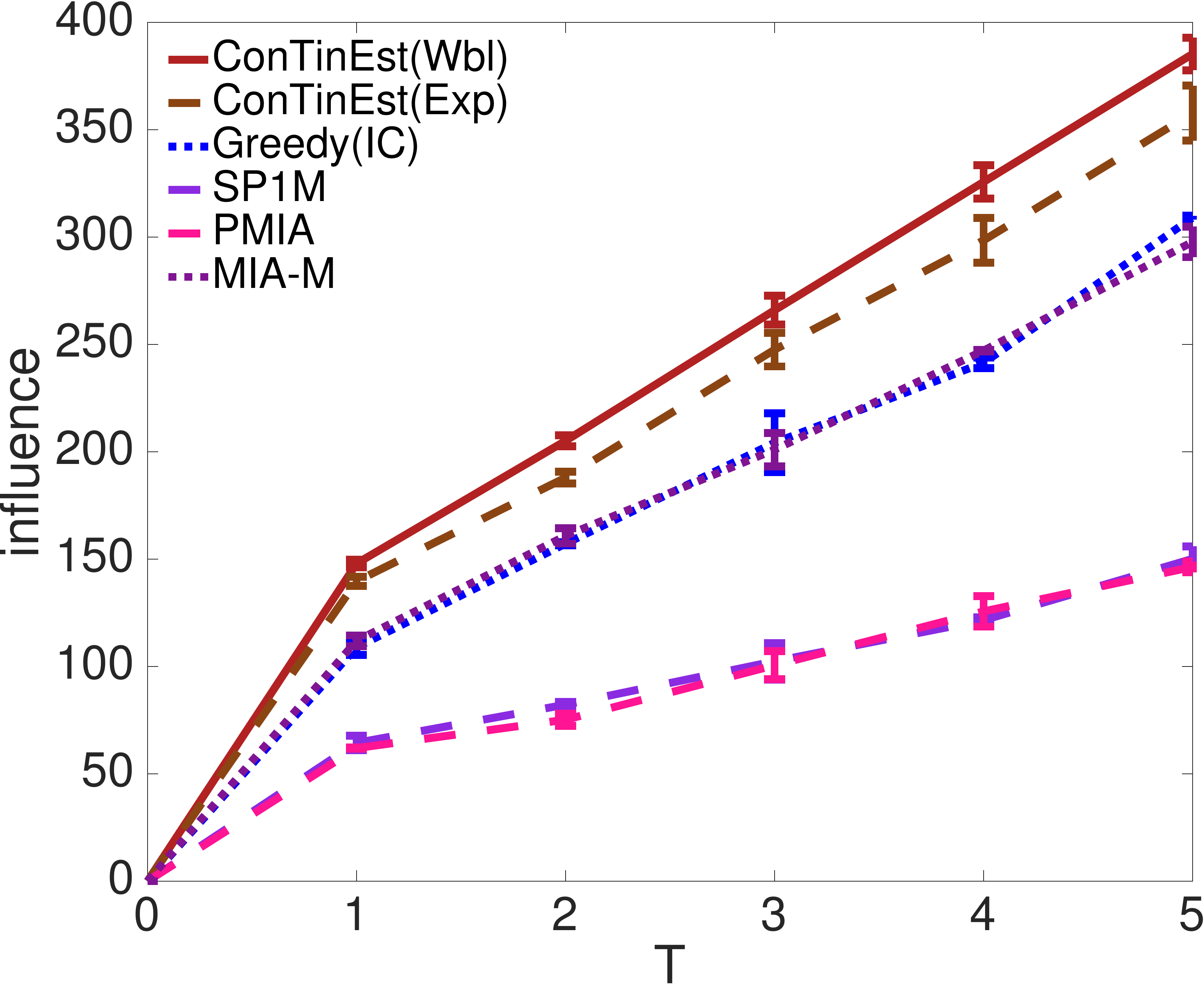} 
\end{subfigure} 
& 
\begin{subfigure}[t]{0.31\textwidth}
  	\includegraphics[width=\textwidth]{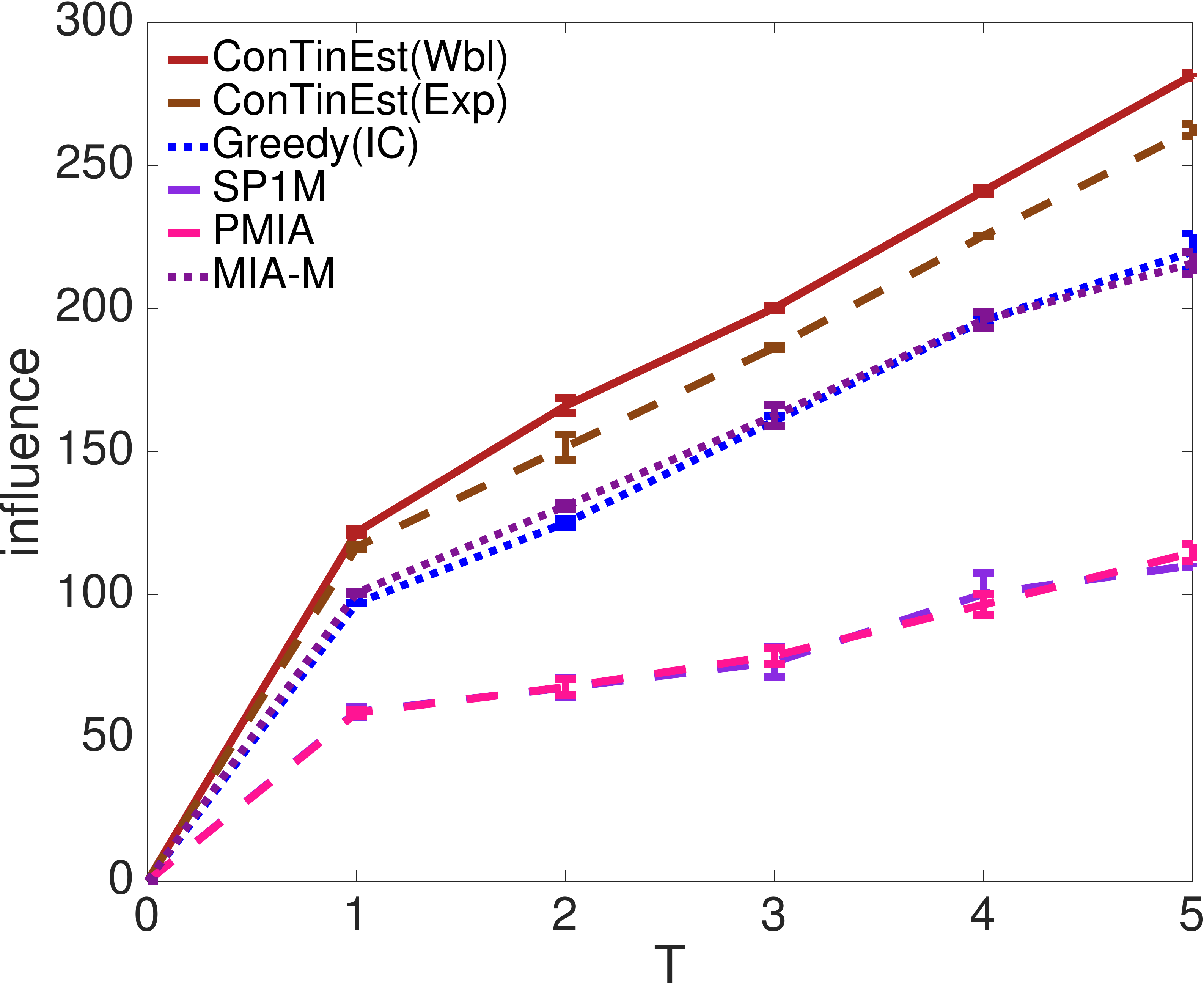} 
\end{subfigure} 
\\
& (a) Core-periphery & (b) Random & (c) Hierarchy
\end{tabular}
 \caption{Influence $\sigma(\Acal, T)$ achieved by varying number of sources $|\Acal|$ and observation window $T$ on the networks of different structures with 1,024 nodes, 2,048 edges and heterogeneous Weibull transmission functions. Top row: influence against \#sources by $T = 5$; Bottom row: influence against the time window $T$ using 50 sources. \label{performance-synthetic}} 
\end{figure}

\subsubsection{Influence Maximization with Uniform Cost}

\begin{figure}
 \centering
 \renewcommand{\tabcolsep}{1pt}
 \begin{tabular}{ccc}
\includegraphics[width=0.33\textwidth]{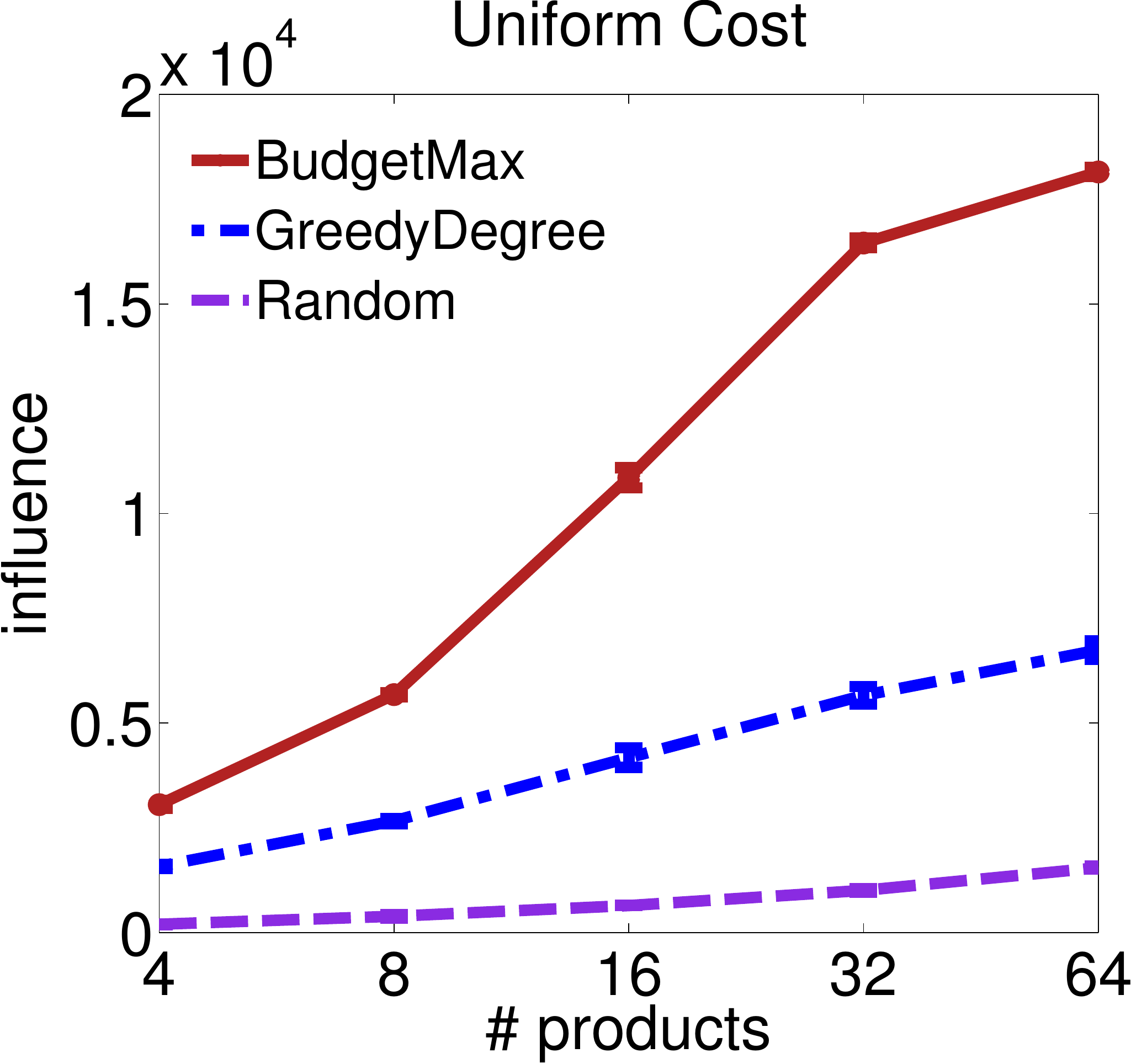} &
\includegraphics[width=0.33\textwidth]{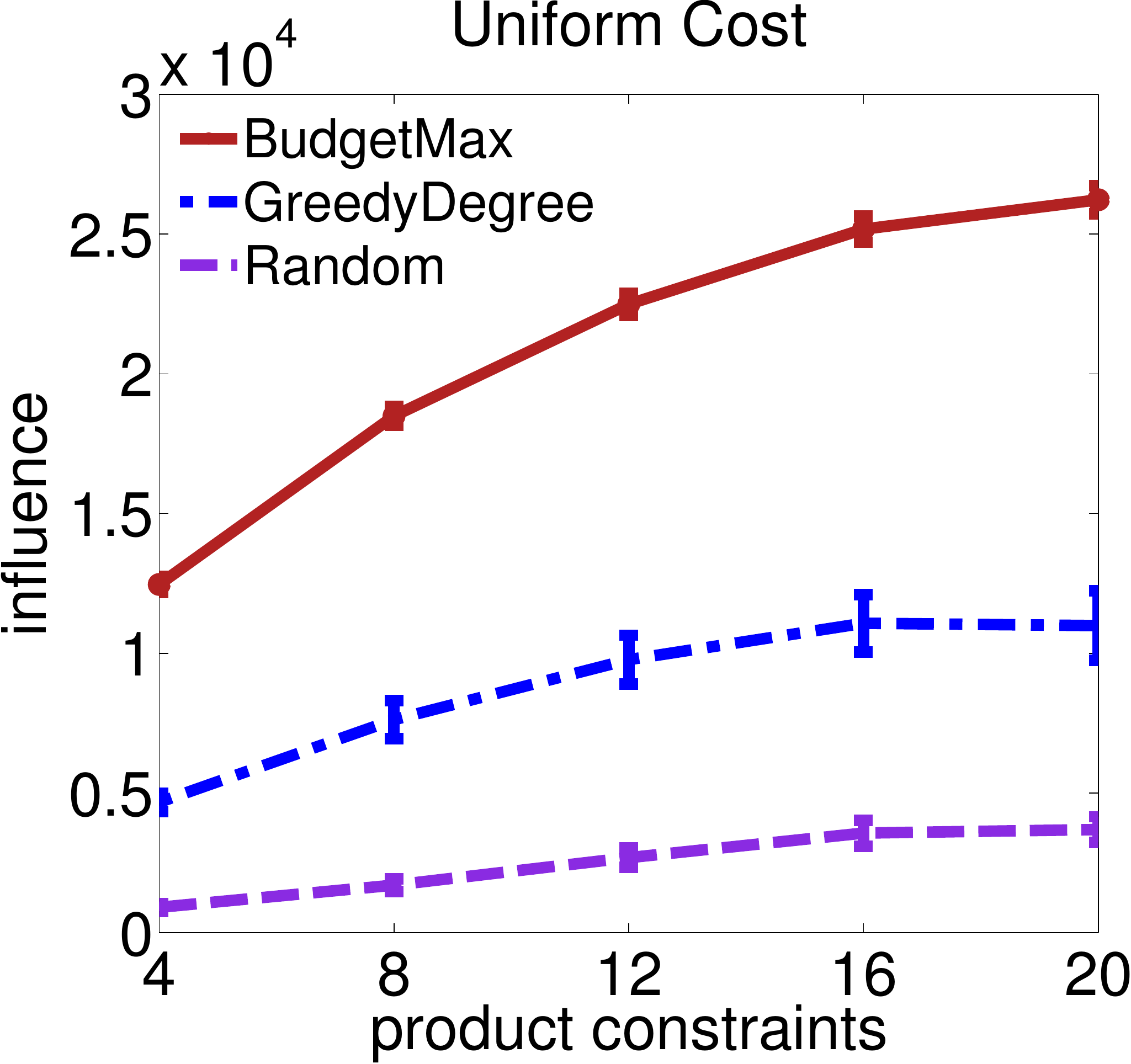} &
\includegraphics[width=0.33\textwidth]{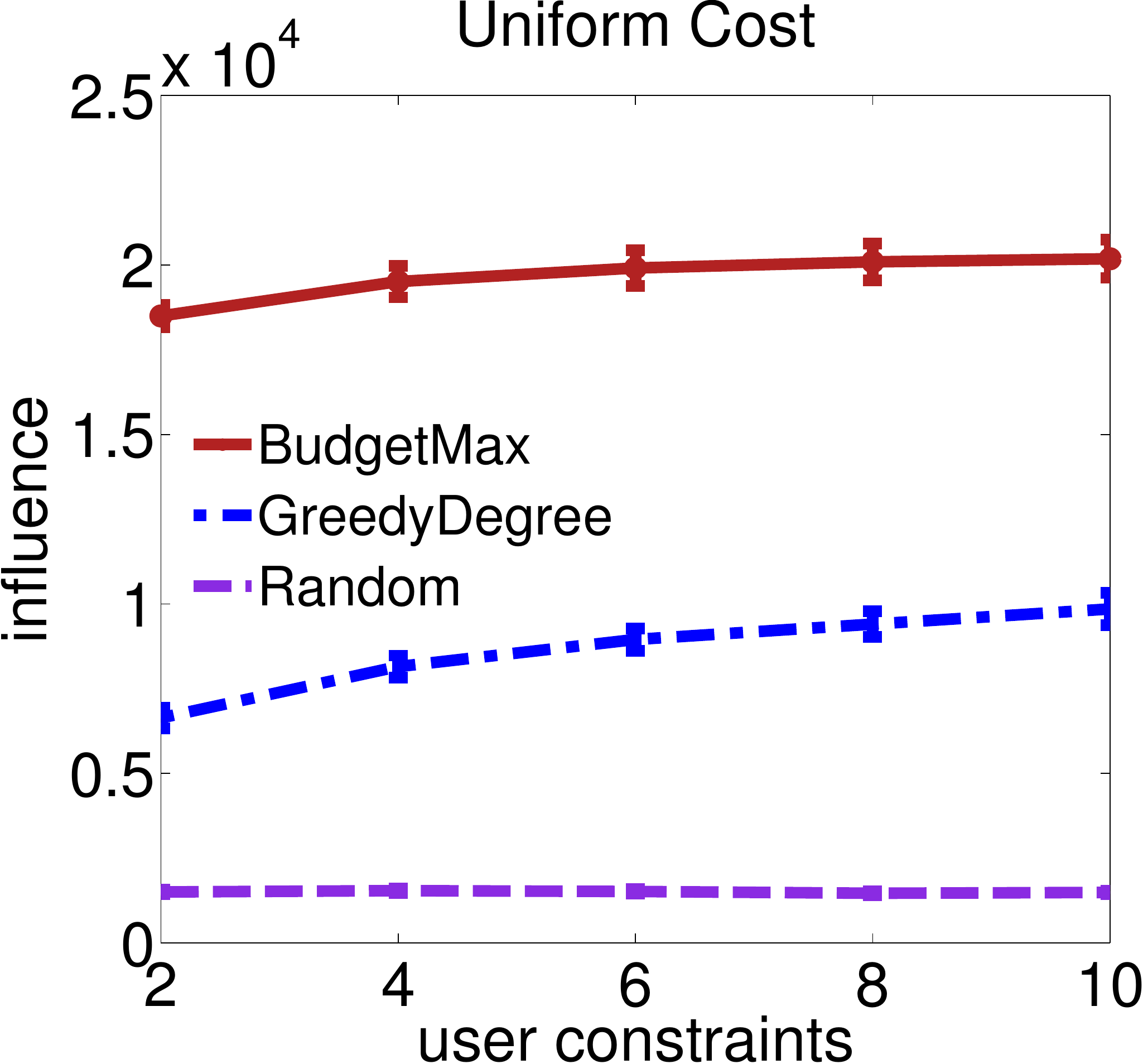} \\
(a) By products & (b) By product constraints& (c) By user constraints
\end{tabular}
 \caption{\label{inf-syn} Over the 64 product-specific diffusion networks, each of which has 1,048,576 nodes,  the estimated influence (a) for increasing the number of products by fixing the product-constraint at 8 and user-constraint at 2; (b) for increasing product-constraint by user-constraint at 2; and (c) for increasing user-constraint by fixing product-constraint at 8. For all experiments, we have $T=5$ time window.}
\end{figure}

In this section, we first evaluate the effectiveness of \continmax to the classic influence maximization problem where we have only one product to assign with the simple cardinality constraint on the users. We compare to other influence maximization methods developed based on discrete-time diffusion models: traditional greedy by~\citep{KemKleTar03}, with discrete-time Linear Threshold Model (LT) and Independent Cascade Model (IC) diffusion models, and the heuristic methods SP1M, PMIA and MIA-M by~\citep{CheWanYan09, CheWanWan2010, CheLuZha2012}.
For \influmax,  since it only supports exponential pairwise transmission functions, we fit an exponential distribution per edge by \netrate~\citep{GomBalSch11}. Furthermore, \influmax is not scalable. When the average network density (defined as the average degree per node) of the synthetic networks is $\sim2.0$, the run time for \influmax is more than $24$ hours. In consequence, we present the results of \continmax using fitted exponential distributions (Exp).
For the discrete-time IC model, we learn the infection probability within time window $T$ using Netrapalli'{}s method~\citep{NetPraSanSuj12}. The learned pairwise infection probabilities are also served for \spm and \pmia, which approximately calculate the influence based on the IC model.
For the discrete-time LT model, we set the weight of each incoming edge to a node $u$ to the inverse of its in-degree, as in previous work~\citep{KemKleTar03}, and choose each node's threshold uniformly at random. 
The top row of Figure~\ref{performance-synthetic} compares the expected number of infected nodes against the source set size for different methods. \continmax outperforms the rest, and the competitive advantage becomes more dramatic the larger the source set grows.
The bottom row of Figure~\ref{performance-synthetic} shows the expected number of infected nodes against the time window for 50 selected sources. Again, \continmax~performs the best for all three types of networks.

Next, using \continmax as a subroutine for influence estimation, we evaluate the performance of \budgetmax with the uniform-cost constraints on the users.
%
%
In our experiments we consider up to 64 products, each of which diffuses over one of the above three different types of Kronecker networks with $\sim$ one million nodes. 
Further, we randomly select a subset of 512 nodes $\Vcal_S\subseteq\Vcal$ as our candidate target users, who will receive the given products, and evaluate the potential influence of an 
allocation over the underlying one-million-node networks.
For \budgetmax, we set the adaptive threshold $\delta$ to 0.01 and the cost per user and product to 1.
For \continmax, we use 2,048 samples with 5 random labels on each of the product-specific diffusion networks.
We repeat our experiments 10 times and report the average performance.
%

We compare \budgetmax with a nodes'{} degree-based heuristic, which we refer to as \gdegree, where the degree is treated as a natural measure of influence, and a baseline method, which assigns the products to the target nodes randomly.
We opt for the nodes'{} degree-based heuristic since, in practice, large-degree nodes, such as users with millions of followers in Twitter, are often the targeted users who will receive a considerable payment if he (she) agrees to post the adoption of some products (or ads) from merchants. 
\gdegree proceeds as follows. 
It first sorts the list of all pairs of products $i$ and nodes $j\in\Vcal_S$ in descending order of node-$j$'s degree in the diffusion network associated to product $i$. 
Then, starting from the beginning of the list, it considers each pair one by one: if the addition of the current pair to the existing solution does not violate the predefined matroid constraints, it is added to the solution, and otherwise, it is skipped. This process continues until the end of the list is reached. 
In other words, we greedily assign products to nodes with the largest degree.  
Due to the large size of the underlying diffusion networks, we do not apply other more expensive node centrality measures such as the clustering coefficient and betweenness.

Figure~\ref{inf-syn} summarizes the results. Panel (a) shows the achieved influence against number of products, fixing the budget per product to 8 and the budget
per user to 2. 
As the number of products increases, on the one hand, more and more nodes become assigned, so the total influence will increase. Yet, on the other hand, the 
competition among products for a few existing \emph{influential} nodes also increases. 
\gdegree achieves a modest performance, since high degree nodes may have many overlapping children. In contrast, \budgetmax, by taking the submodularity of 
the problem, the network structure and the diffusion dynamics of the edges into consideration, achieves a superior performance, especially as the number of product
(\ie, the competition) increases.
Panel (b) shows the achieved influence against the budget per product, considering 64 products and fixing the budget per user to 2. 
We find that, as the budget per product increases, the performance of \gdegree tends to flatten and the competitive advantage of \budgetmax becomes more dramatic.
Finally, Panel (c) shows the achieved influence against the budget per user, considering 64 products and fixing the budget per product to 8. 
We find that, as the budget per user increases, the influence only increases slowly. This is due to the fixed budget per product, which prevents additional new nodes to be assigned. 
This meets our intuition: by making a fixed number of people watching more ads per day, we can hardly boost the popularity of the product. 
Additionally, even though the same node can be assigned to more products, there is hardly ever a node that is the \emph{perfect} source from which all products can efficiently spread.

\subsubsection{Influence Maximization with Non-Uniform Cost}
In this section, we evaluate the performance of \budgetmax under non-uniform cost constraints, using again \continmax as a subroutine for influence estimation.
Our designing of user-cost aim to mimic a real scenario, where advertisers pay much more money to celebrities with millions of social network followers than to 
normal citizens. 
To do so, we let $c_i \propto d_i^{1/n}$ where $c_i$ is the cost, $d_i$ is the degree, and $n\geq 1$ controls the increasing speed of cost with respect to the degree. 
In our experiments, we use $n=3$ and normalize $c_i$ to be within $(0,1]$. 
Moreover, we set the product-budget to a base value from 1 to 10 and add a random adjustment drawn from a uniform distribution $U(0,1)$. 

We compare our method to two modified versions of the above mentioned nodes'{} degree-based heuristic \gdegree and to the same baseline method. 
In the first modified version of the heuristic, which we still refer to as \gdegree, takes both the degree and the corresponding cost into consideration. In particular, it sorts the list 
of all pairs of products $i$ and nodes $j\in\Vcal_S$ in descending order of degree-cost ratio $d_j / c_j$ in the diffusion network associated to product $i$, instead of simply 
the node-$j$'s degree, and then proceeds similarly as before. 
In the second modified version of the heuristic, which we refer as GreedyLocalDegree, we use the same degree-cost ratio but allow the target users to be partitioned into distinct groups (or communities) and pick the most cost-effective pairs within each group locally instead.
%
\begin{figure}
 \centering
 \renewcommand{\tabcolsep}{1pt}
 \begin{tabular}{ccc}
\includegraphics[width=0.34\textwidth]{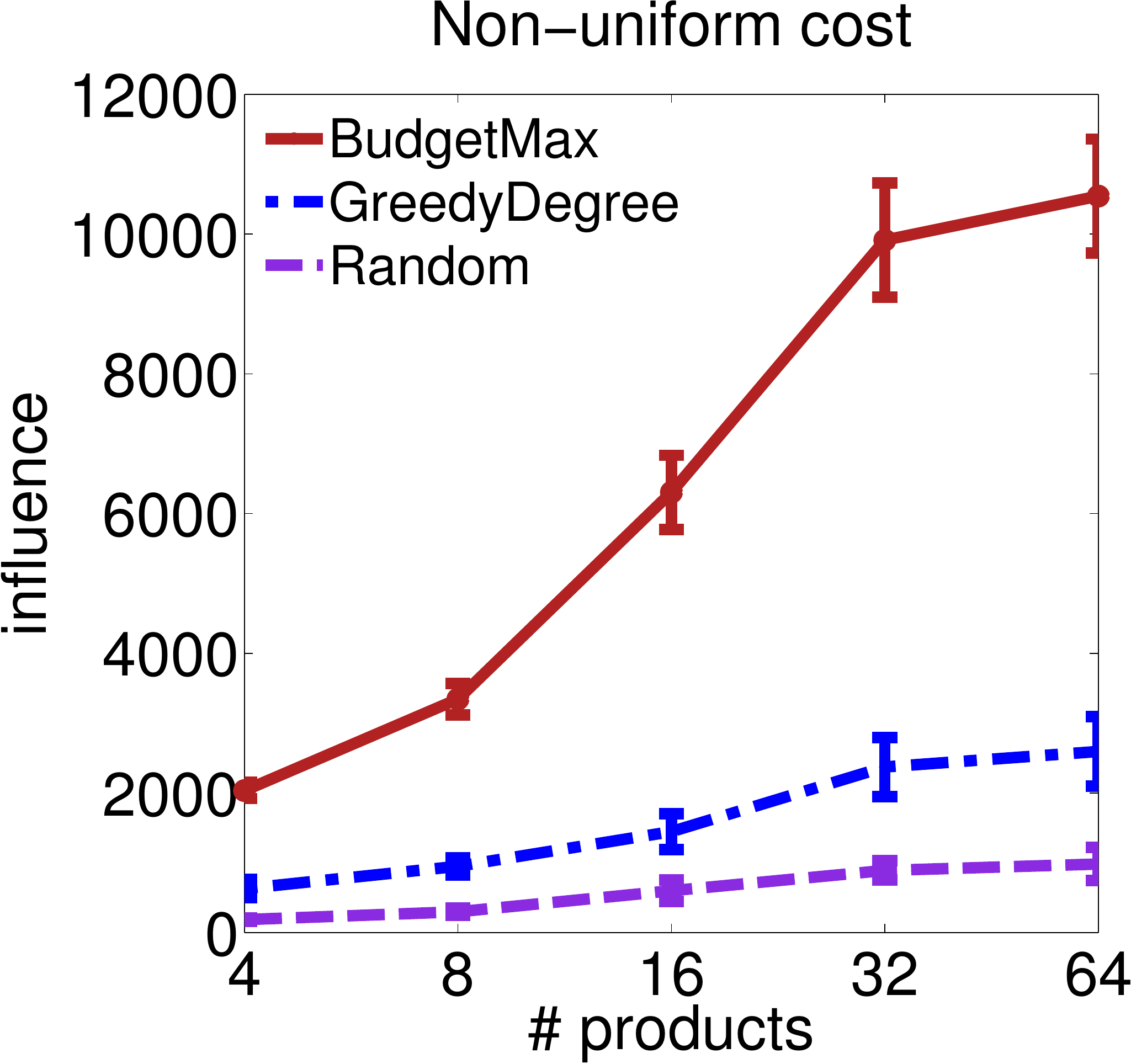} &
\includegraphics[width=0.33\textwidth]{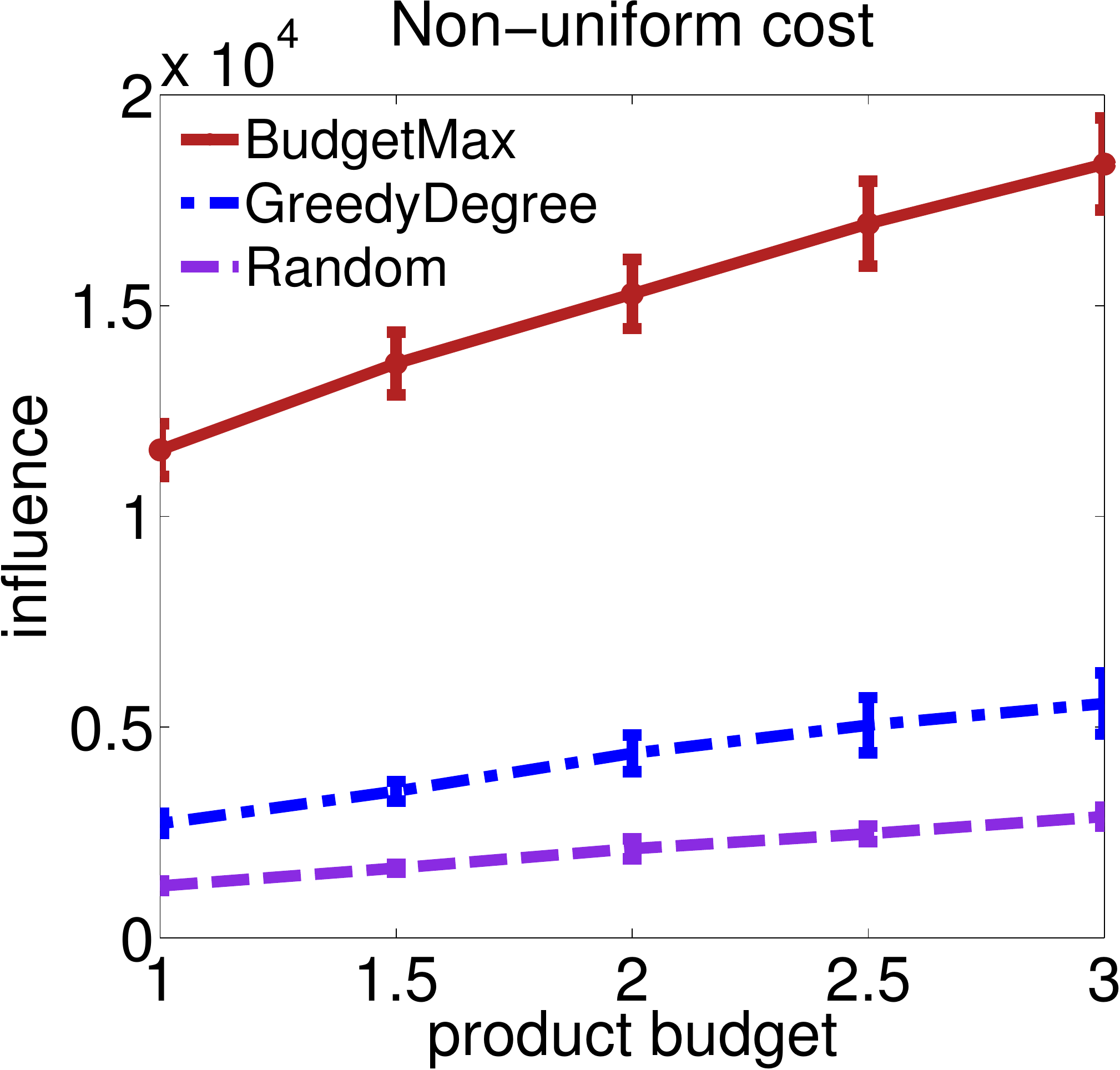} &
\includegraphics[width=0.33\textwidth]{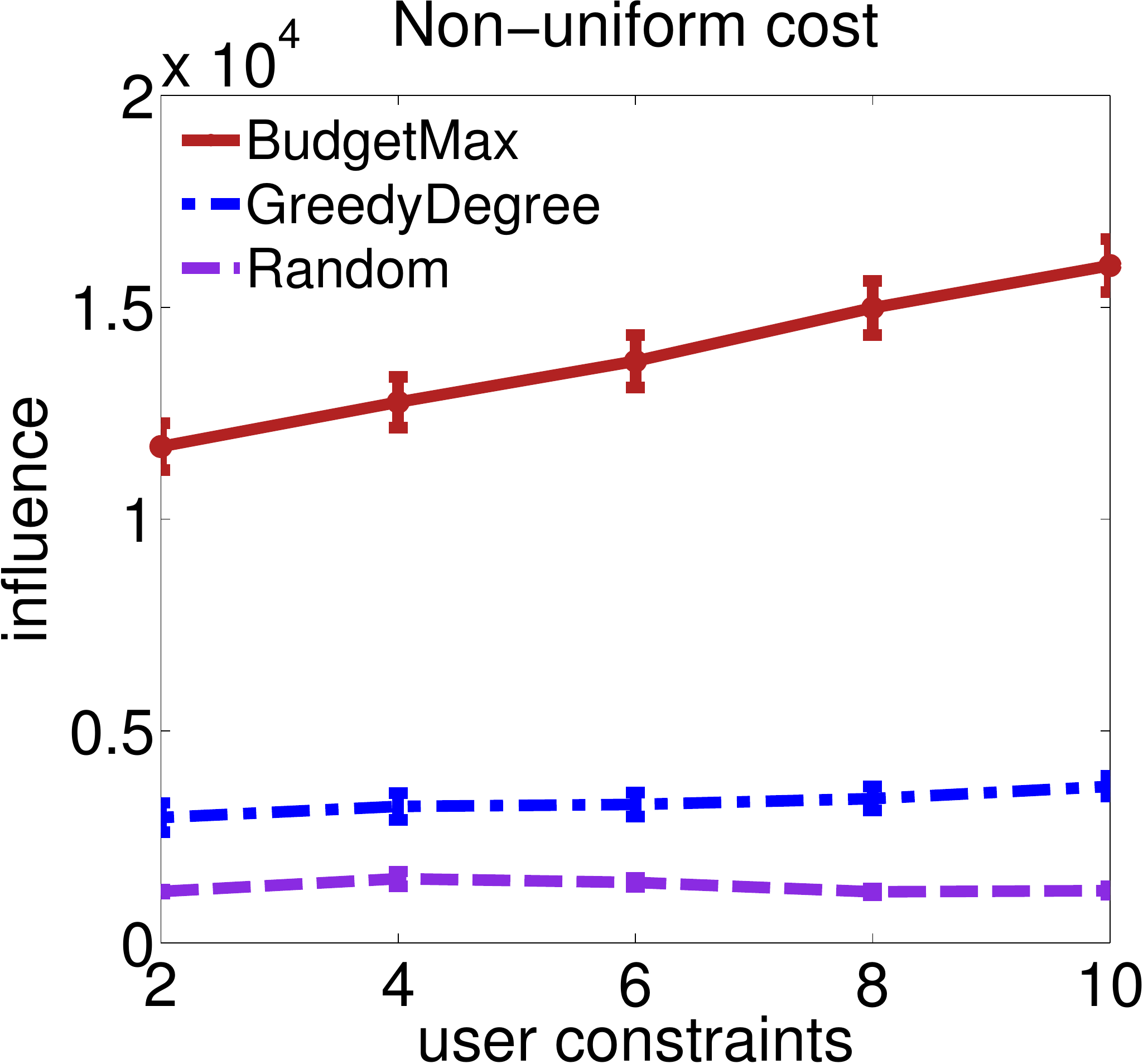} \\
(a) By products & (b) By product budgets& (c) By user constraints
\end{tabular}\\[8mm]
 \renewcommand{\tabcolsep}{4pt}
\begin{tabular}{cc}
\includegraphics[width=0.33\textwidth]{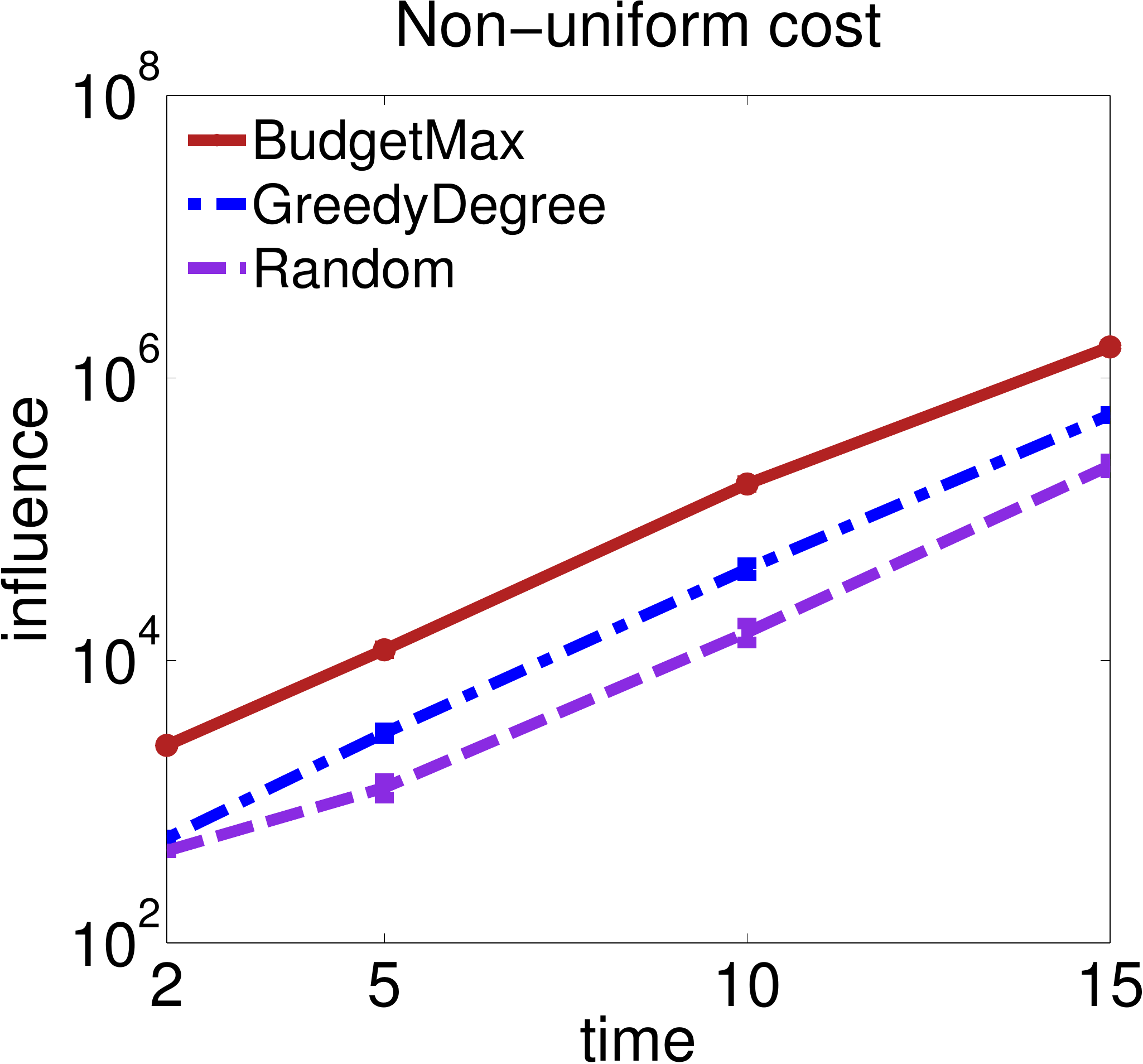} &
\includegraphics[width=0.33\textwidth]{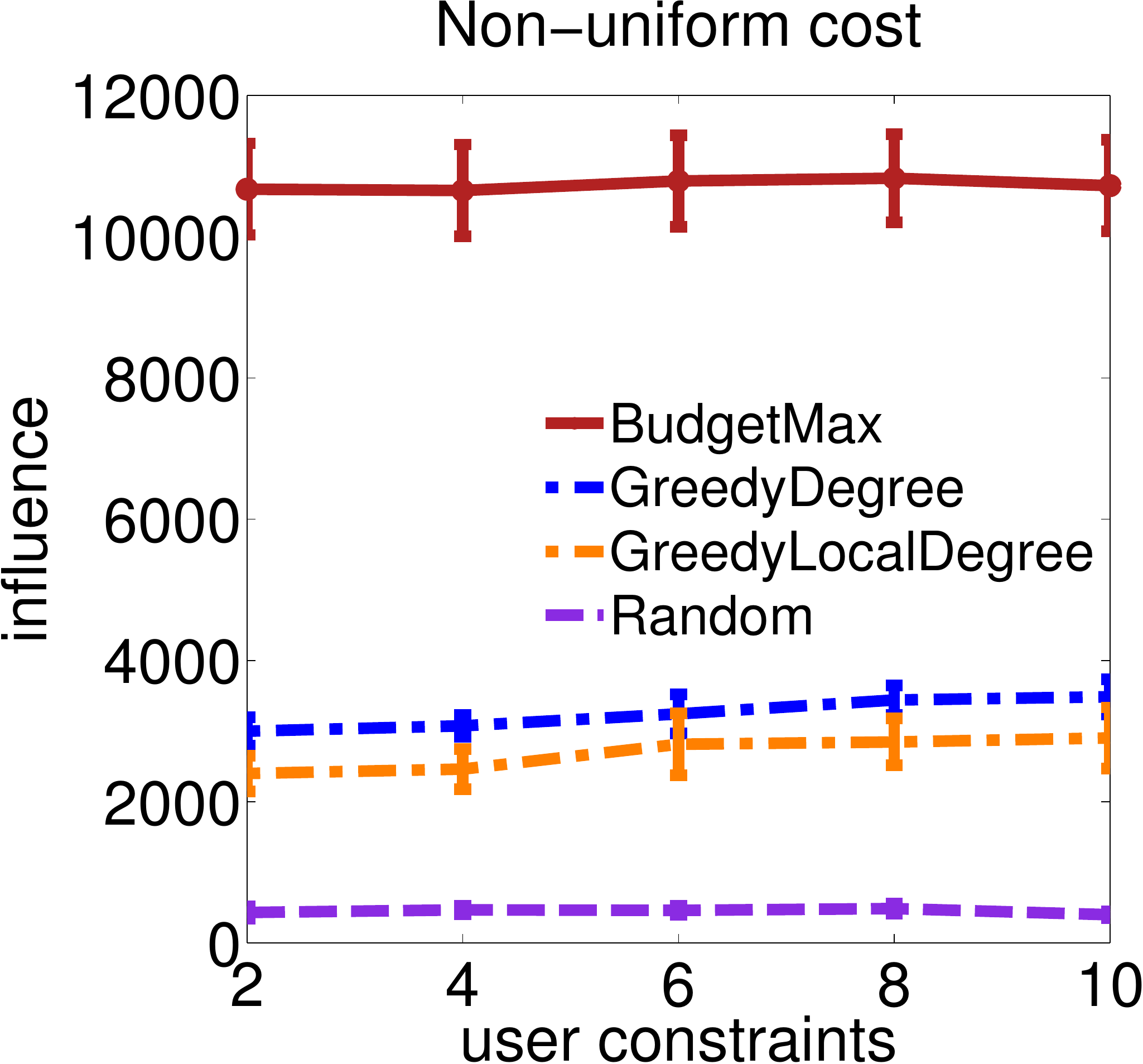} \\
(d) By time & (e) By group limits
\end{tabular}
 \caption{\label{inf-budget-syn} Over the 64 product-specific diffusion networks, each of which has a total 1,048,576 nodes, the estimated influence (a) for increasing the number of products by fixing the product-budget at 1.0 and user-constraint at 2; (b) for increasing product-budget by fixing user-constraint at 2; (c) for increasing user-constraint by fixing product-budget at 1.0; (d) for different time window T; and (e) for increasing user-constraint with group-limit 16 by fixing product-budget at 1.0.}
\end{figure}
Figure~\ref{inf-budget-syn} compares the performance of our method with the competing methods against four factors: (a) the number of products, (b) the budget per product, 
(c) the budget per user and (d) the time window $T$, while fixing the other factors.
In all cases, \budgetmax significantly outperforms the other methods, and the achieved influence increases monotonically with respect to the factor value, as one may have expected. 
In addition, in Figure~\ref{inf-budget-syn}(e), we study the effect of the Laminar matroid combined with group knapsack constraints, which is the most general type of constraint we handle in this paper (refer to Section~\ref{sec:inf}). The selected target users are further partitioned into $K$ groups randomly, each of which has, $Q_i,i = 1\dotso K$, limit which constrains the maximum allocations allowed in each group. In practical scenarios, each group might correspond to a geographical community or organization. In our experiment, we divide the users into 8 equal-size groups and set $Q_i = 16,i = 1\dotso K$ to indicate that we want a balanced allocation in each group. 
Figure~\ref{inf-budget-syn}(e) shows the achieved influence against the budget per user for $K = 8$ equally-sized groups and $Q_i = 16, i=1 \dotso K$.
In contrast to Figure~\ref{inf-budget-syn}(b), the total estimated influence does not increase significantly with respect to the budget (\ie, number of slots) per user. 
This is due to the fixed budget per group, which prevents additional new nodes to be assigned, even though the number of available slots per user increases.

\subsubsection{Effects of Adaptive Thresholding}
\begin{figure}[t]
 \centering
 \renewcommand{\tabcolsep}{5pt}
 \begin{tabular}{cc}
\includegraphics[width=0.45\columnwidth]{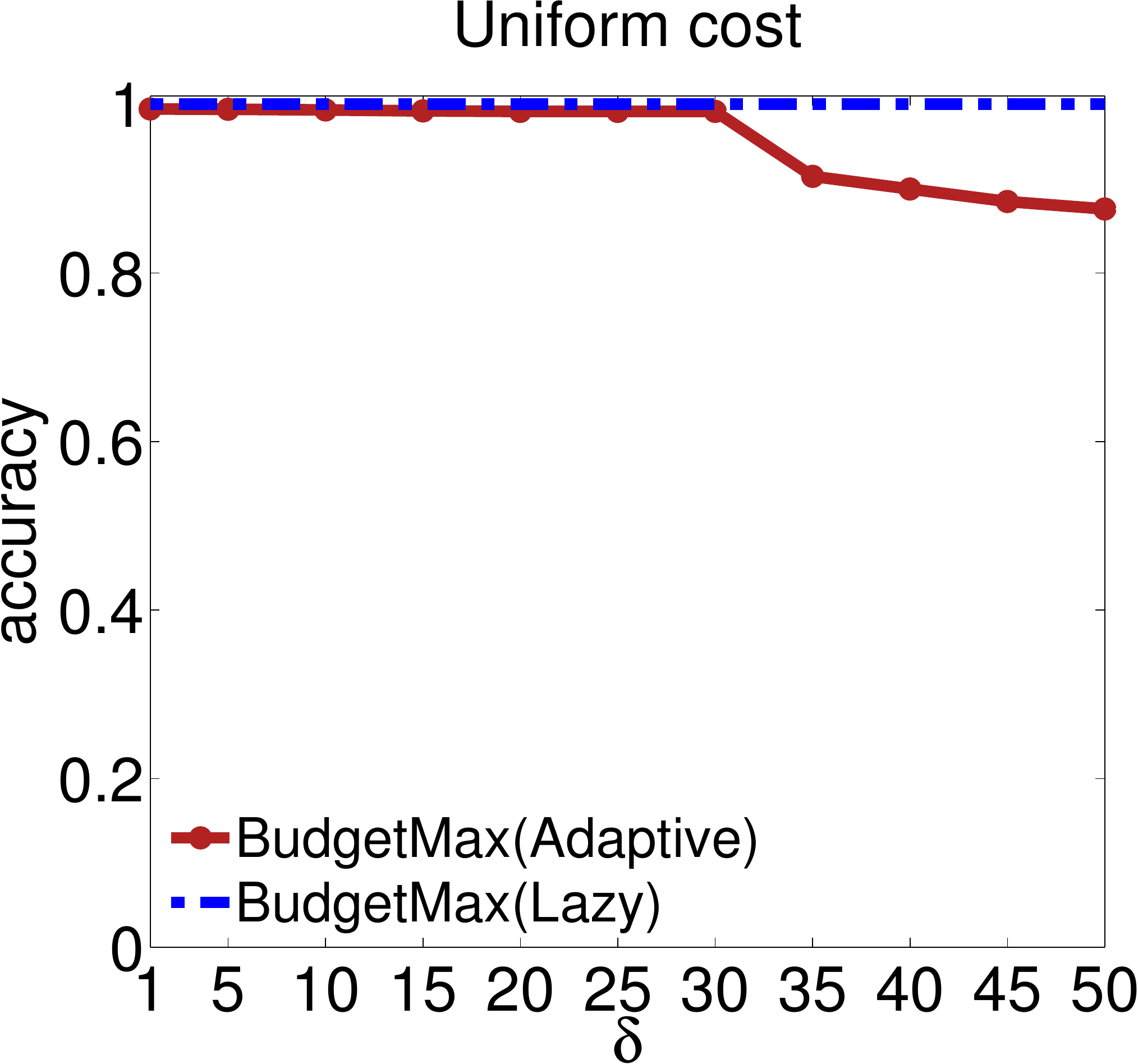}
& \includegraphics[width=0.45\columnwidth]{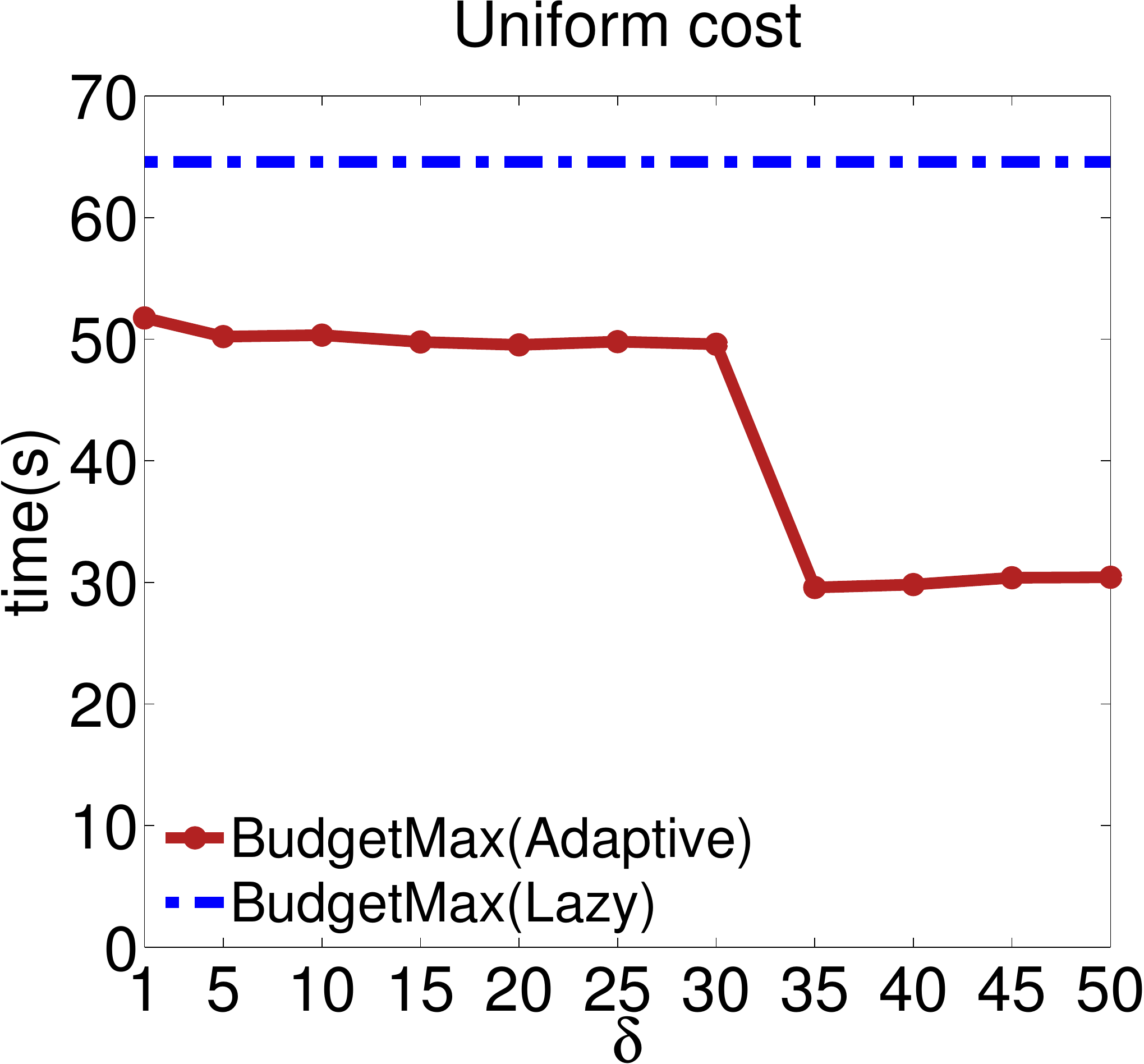} \\
(a) $\delta$ vs. accuracy & (b) $\delta$ vs. time
  \vspace{-1mm}
\end{tabular}
 \caption{\label{speed-syn}
The relative accuracy and the run-time for different threshold parameter $\delta$.}
\end{figure}
In Figure~\ref{speed-syn}, we investigate the impact that the threshold value $\delta$ has on the accuracy and runtime of our adaptive thresholding algorithm and compare it with the lazy evaluation method. Note that the performance and runtime of lazy evaluation do not change with respect to $\delta$ because it does not depend on it.
Panel (a) shows the achieved influence against the threshold $\delta$. As expected, the larger the $\delta$ value, the lower the accuracy. However, our method is relatively robust to the 
particular choice of $\delta$ since its performance is always over a 90-percent relative accuracy even for large $\delta$. 
Panel (b) shows the runtime against the threshold $\delta$. In this case, the larger the $\delta$ value, the lower the runtime. 
In other words, Figure~\ref{speed-syn} verifies the intuition that $\delta$ is able to trade off the solution quality of the allocation with the runtime time. 

\subsubsection{Scalability}

\begin{figure}[t]
  \centering
  \renewcommand{\tabcolsep}{1pt}
  {\small
  \begin{tabular}{ccc}
    \includegraphics[width=0.33\textwidth]{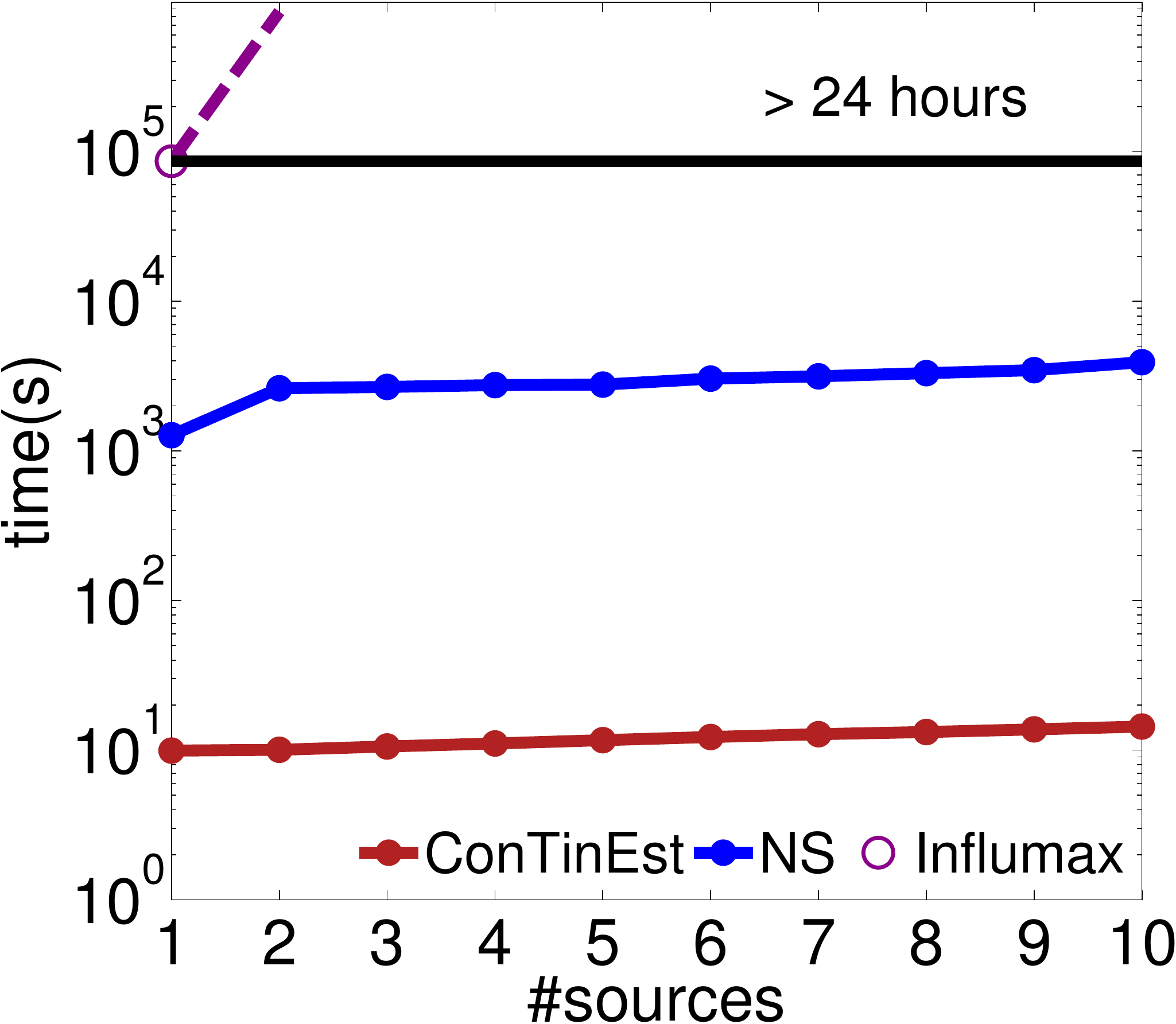} &
    \includegraphics[width=0.33\textwidth]{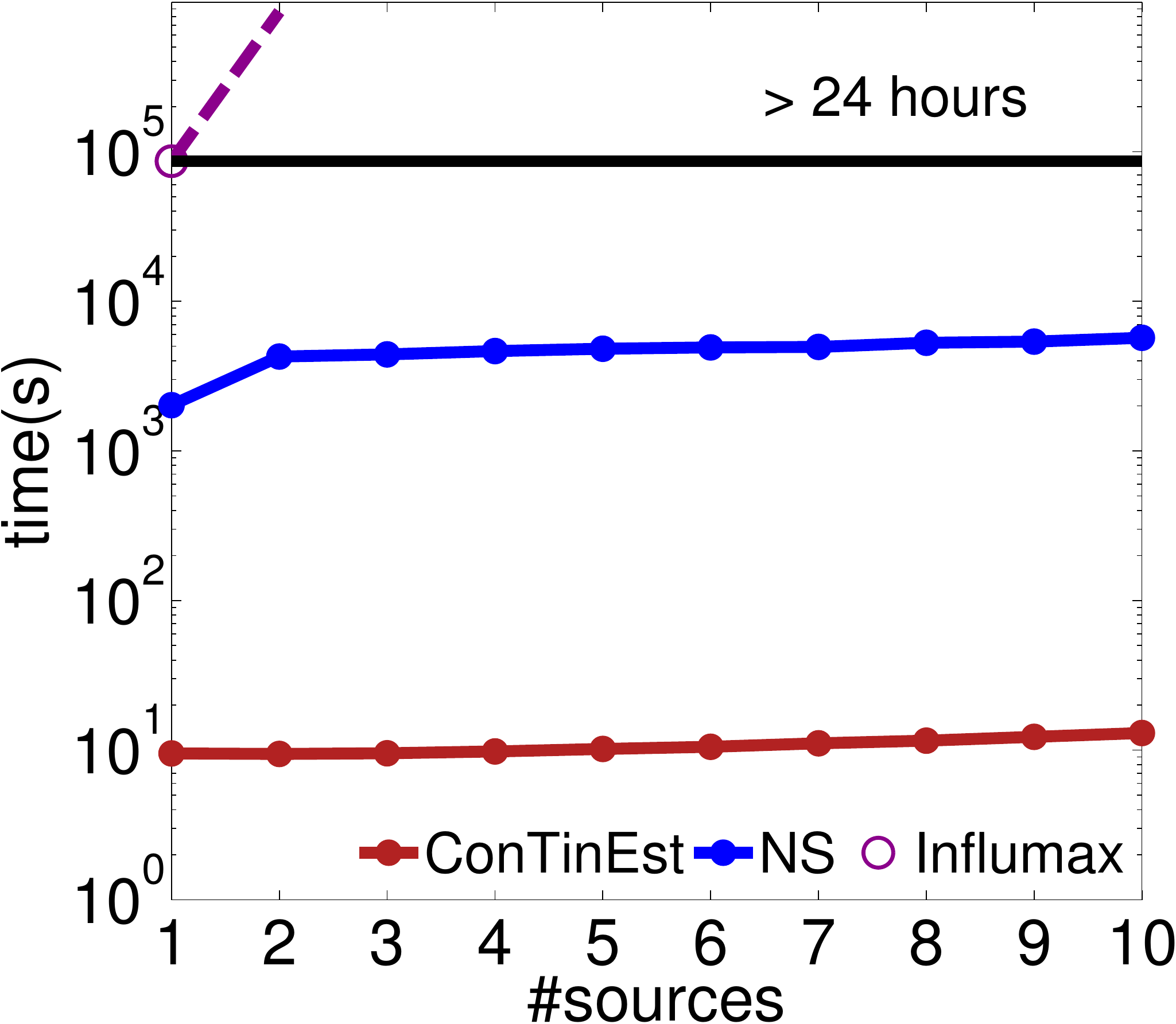} &
    \includegraphics[width=0.33\textwidth]{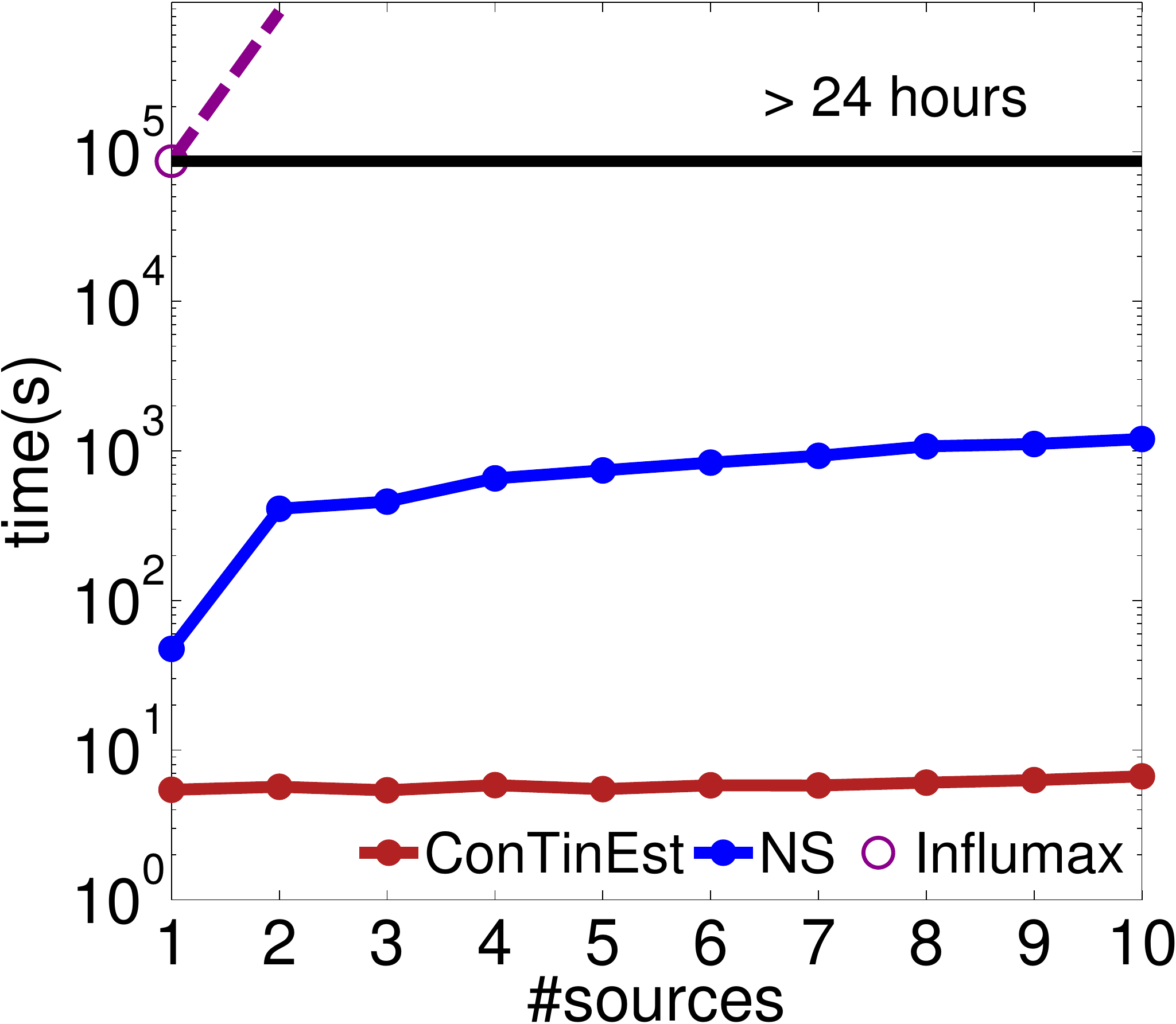} \\
    (a) Core-Periphery & (b) Random &  (c)  Hierarchy
  \end{tabular}
  }
  \caption{\label{source_speed} Runtime of selecting increasing number of sources on Kronecker networks of 128 nodes and 320 edges with $T = 10$.}
\end{figure}

\begin{figure}
  \centering
  \renewcommand{\tabcolsep}{1pt}
  {\small
  \begin{tabular}{ccc}
    \includegraphics[width=0.33\textwidth]{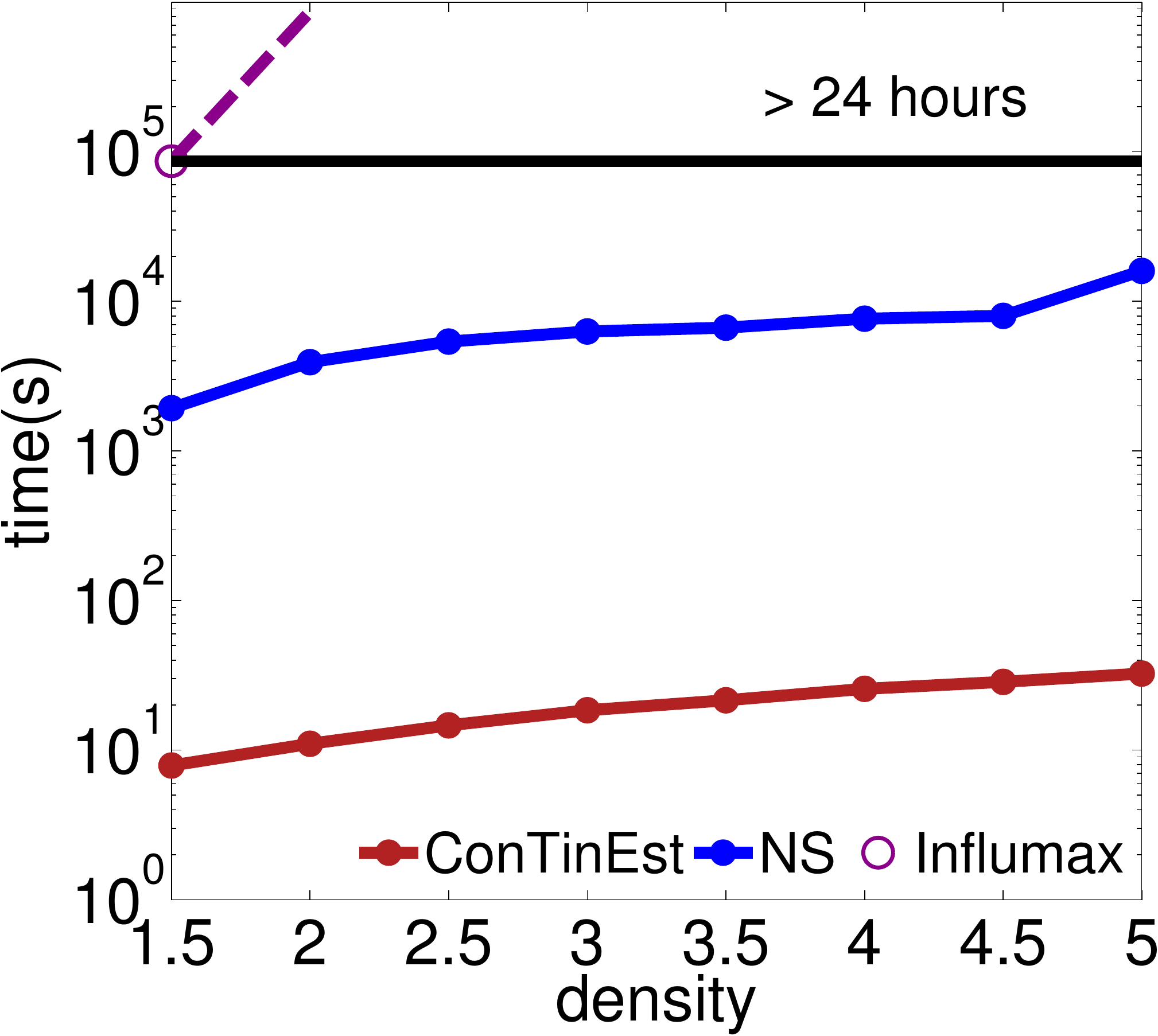} &
    \includegraphics[width=0.33\textwidth]{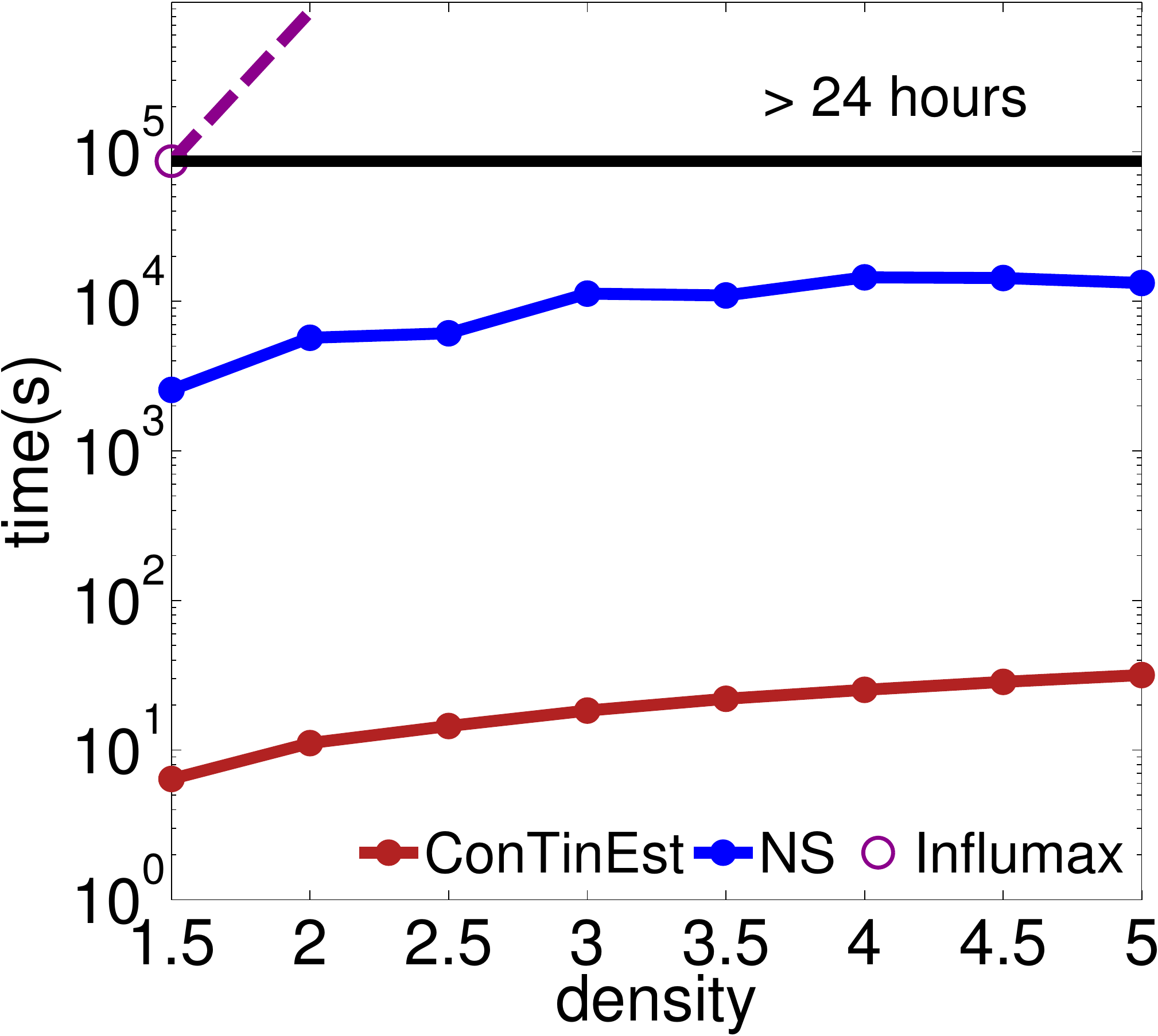} &
    \includegraphics[width=0.33\textwidth]{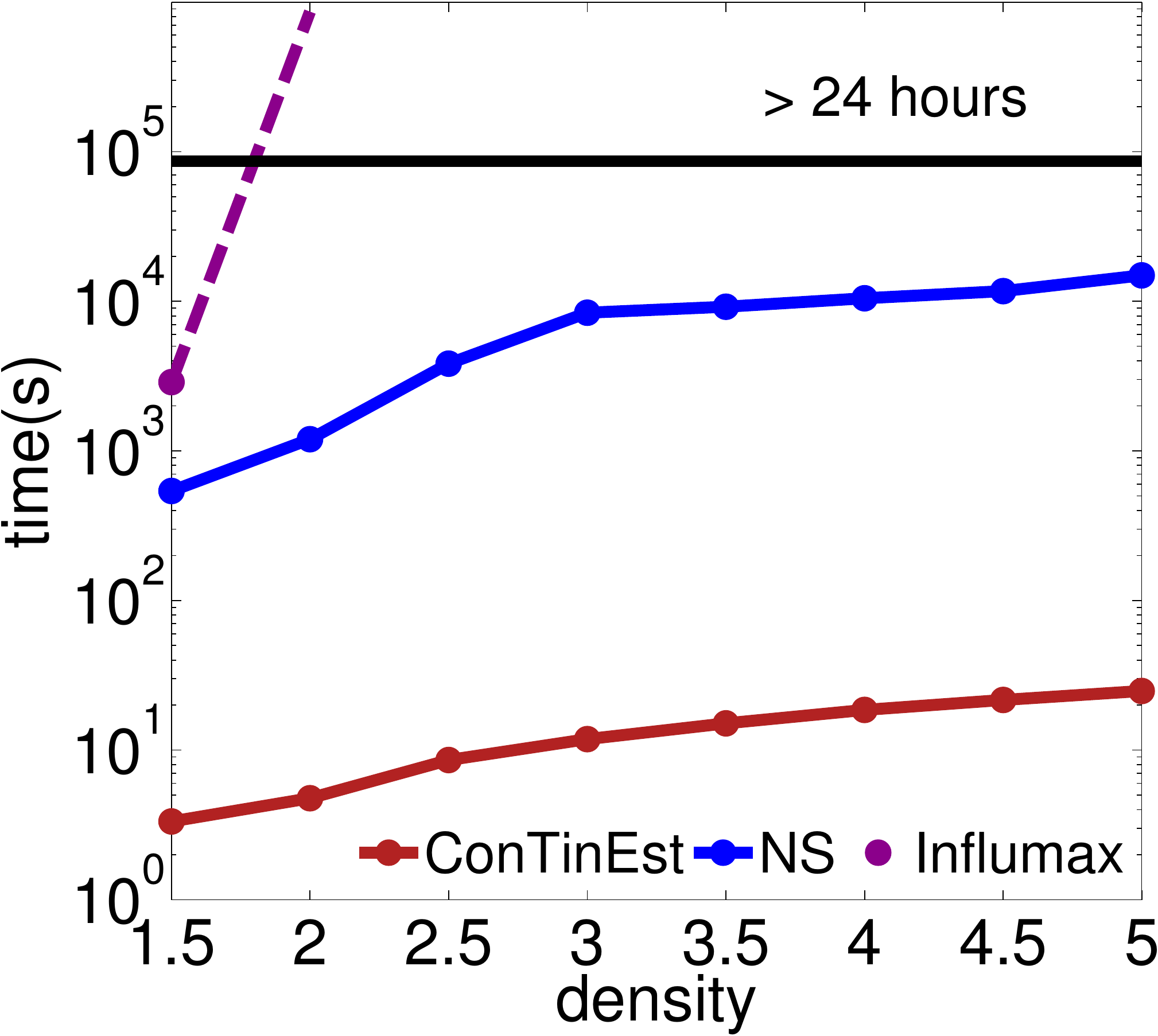} \\
    (a) Core-Periphery & (b) Random &  (c)  Hierarchy
  \end{tabular}
  }
  \caption{\label{density_speed} Runtime of selecting 10 sources in networks of 128 nodes with increasing density by $T = 10$.}
\end{figure}

\begin{figure}
  \centering
  \renewcommand{\tabcolsep}{1pt}
  {\small
  \begin{tabular}{c}
    \includegraphics[width=0.5\textwidth]{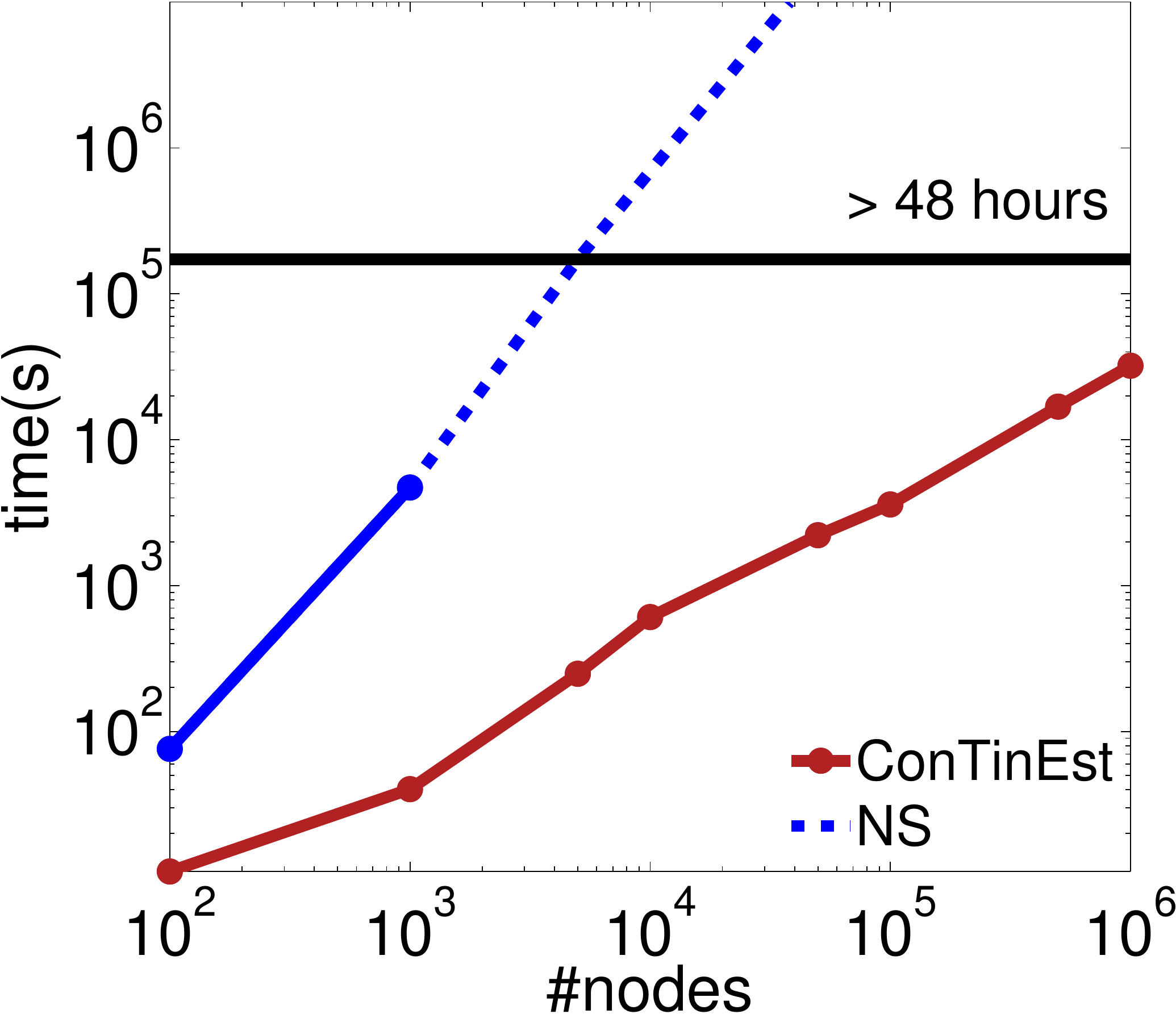}
  \end{tabular}
  }
  \caption{\label{continest_speed} For core-periphery networks by $T = 10$, runtime of selecting 10 sources with increasing network size from 100 to 1,000,000 by fixing 1.5 network density.}
\end{figure}

\begin{figure}
 \centering
 \renewcommand{\tabcolsep}{5pt}
 \begin{tabular}{cc}
\includegraphics[width=0.45\textwidth]{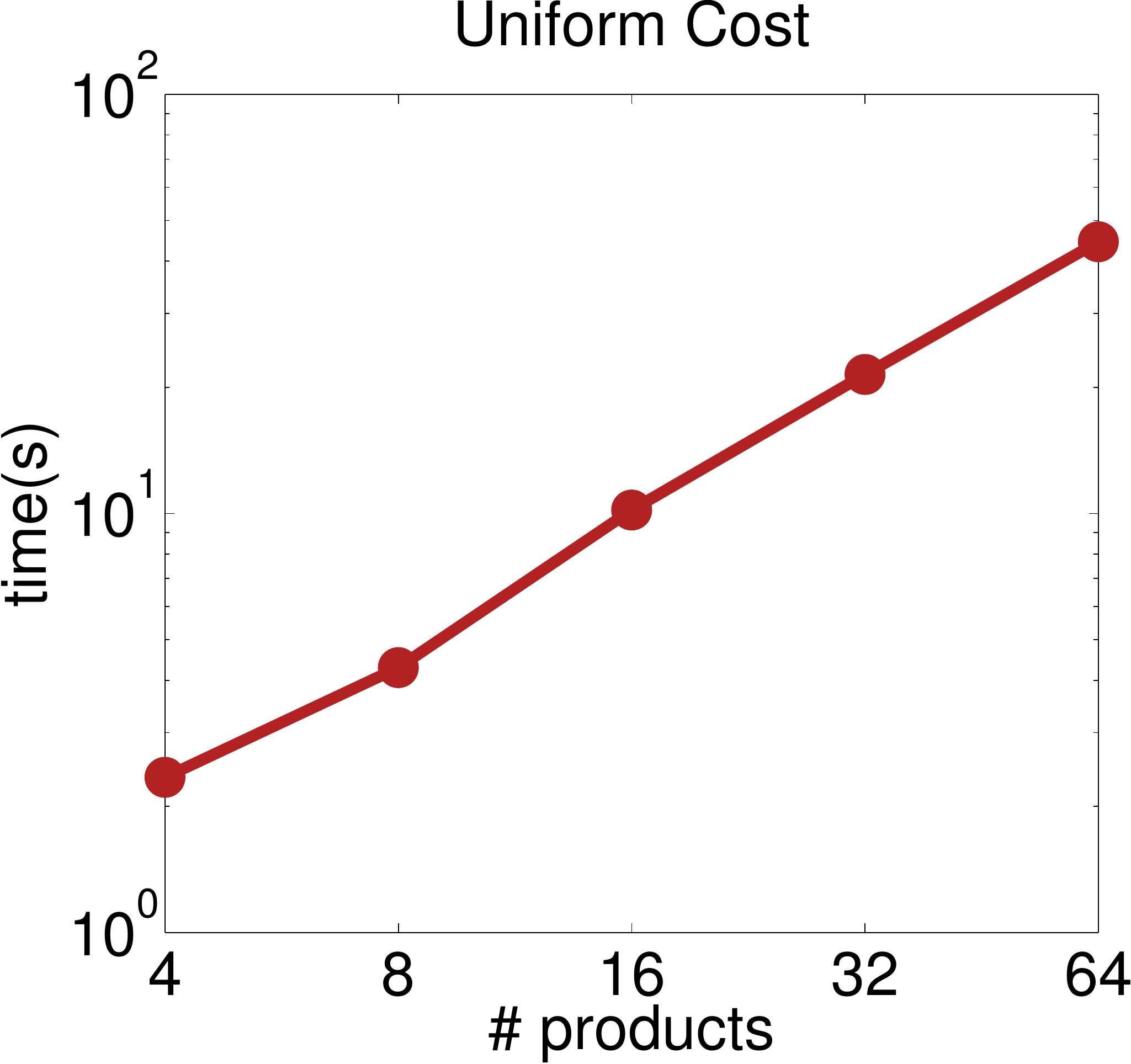} &
\includegraphics[width=0.45\textwidth]{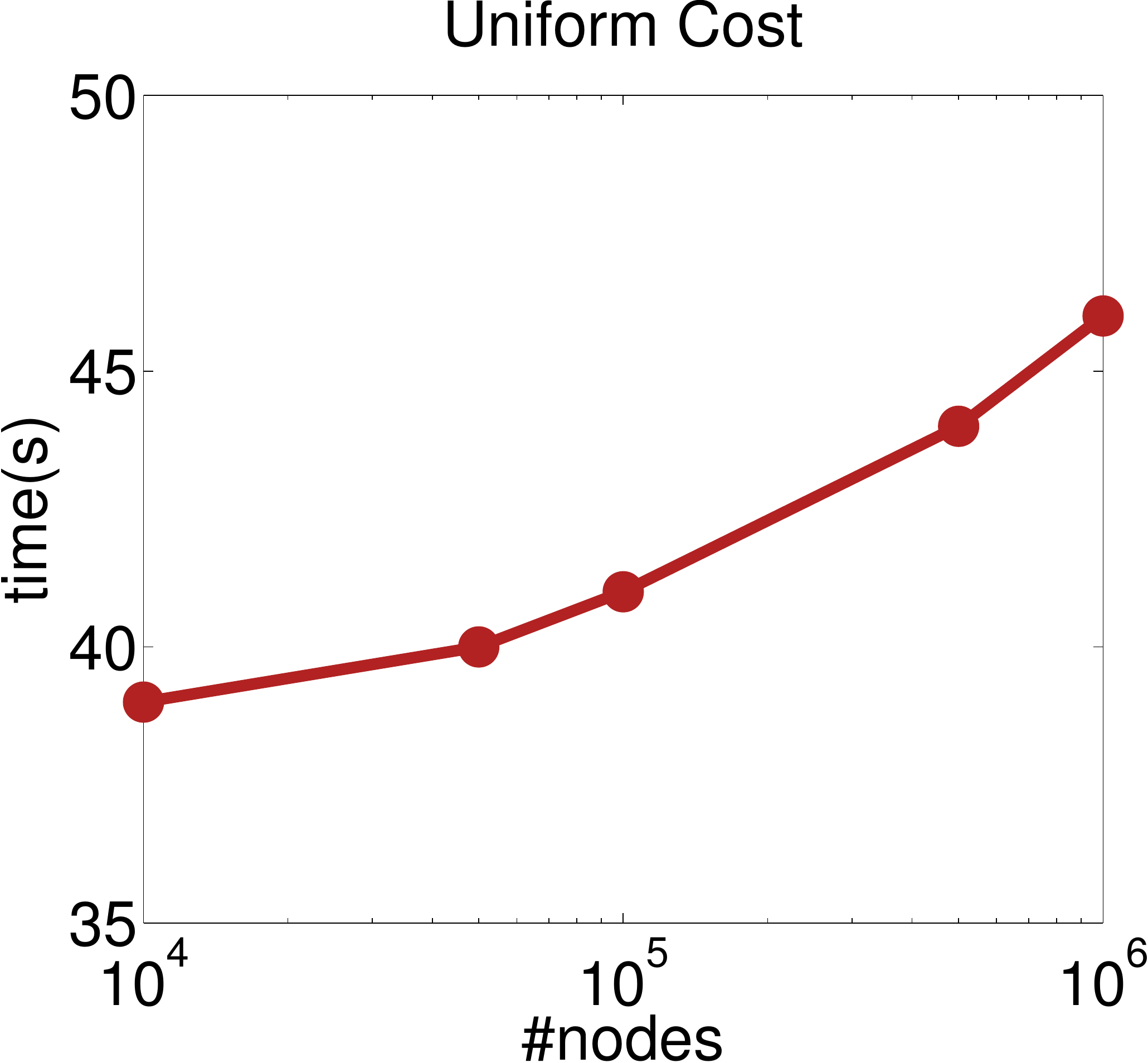} \\
(a) Speed by products & (b) Speed by nodes
\end{tabular}
 \caption{\label{allocation-speed} Over the 64 product-specific diffusion networks, each of which has 1,048,576 nodes, the runtime (a) of allocating increasing number of products and (b) of allocating 64 products to 512 users on networks of varying size. or all experiments, we have $T=5$ time window and fix product-constraint at 8 and user-constraint at 2.}
\end{figure}

In this section, we start with evaluating the scalability of the proposed algorithms on the classic influence maximization problem where we only have one product with the cardinality constraint on the users. 

We compare it to the state-of-the-art method \influmax~\citep{GomSch12} and the Naive Sampling (NS) method in terms of runtime for the continuous-time influence estimation and maximization. For \continmax, we draw 10,000 samples in the outer loop, each having 5 random labels in the inner loop. We plug \continmax as a subroutine into the classic greedy algorithm by~\citep{NemWolFis78}. For NS, we also draw 10,000 samples. The first two experiments are carried out in a single 2.4GHz processor.

Figure~\ref{source_speed} compares the performance of increasingly selecting sources (from 1 to 10) on small Kronecker networks. When the number of selected sources is 1, different algorithms essentially spend time estimating the influence for each node. \continmax outperforms other methods by order of magnitude and for the number of sources larger than 1, it can efficiently reuse computations for estimating influence for individual nodes. Dashed lines mean that a method did not finish in 24 hours, and the estimated run time is plotted.

Next, we compare the run time for selecting 10 sources with increasing densities (or the number of edges) in Figure~\ref{density_speed}. Again, \influmax and NS are order of magnitude slower due to their respective exponential and quadratic computational complexity in network density. In contrast, the run time of \continmax only increases slightly with the increasing density since its computational complexity is linear in the number of edges. We evaluate the speed on large core-periphery networks, ranging from 100 to 1,000,000 nodes with density 1.5 in Figure~\ref{continest_speed}. We report the parallel run time only for \continmax and NS (both are implemented by MPI running on 192 cores of 2.4Ghz) since \influmax is not scalable. In contrast to NS, the performance of \continmax increases linearly with the network size and can easily scale up to one million nodes.

Finally, we investigate the performance of \budgetmax in terms of runtime when using~\continmax as subroutine to estimate the influence. We can precompute the data structures and store the samples needed to estimate the influence function in advance. Therefore, we focus only on the runtime for the constrained influence maximization algorithm. \budgetmax runs on 64 cores of 2.4Ghz by using OpenMP to accelerate the first round of the optimization.
We report the allocation time for increasing number of products in Figure~\ref{allocation-speed}(a), which clearly shows a linear time complexity with respect to the size of the ground set. Figure~\ref{allocation-speed}(b) evaluates the runtime of allocation by varying the size of the network from 16,384 to 1,048,576 nodes.

\subsection{Experiments on Real Data}
In this section, we first quantify how well our proposed algorithm can estimate the true influence in the real-world dataset. Then, we evaluate the solution quality of the selected sources for influence maximization under different constraints. We have used the public MemeTracker datasets~\citep{LesBacKle09}, which contains more than 172 million news articles and blog posts from 1 million mainstream media sites and blogs.

\subsubsection{Influence Estimation}
\begin{figure}
  \centering
  \renewcommand{\tabcolsep}{0pt}
  {\small
  \begin{tabular}{c}
    \includegraphics[width=0.5\columnwidth]{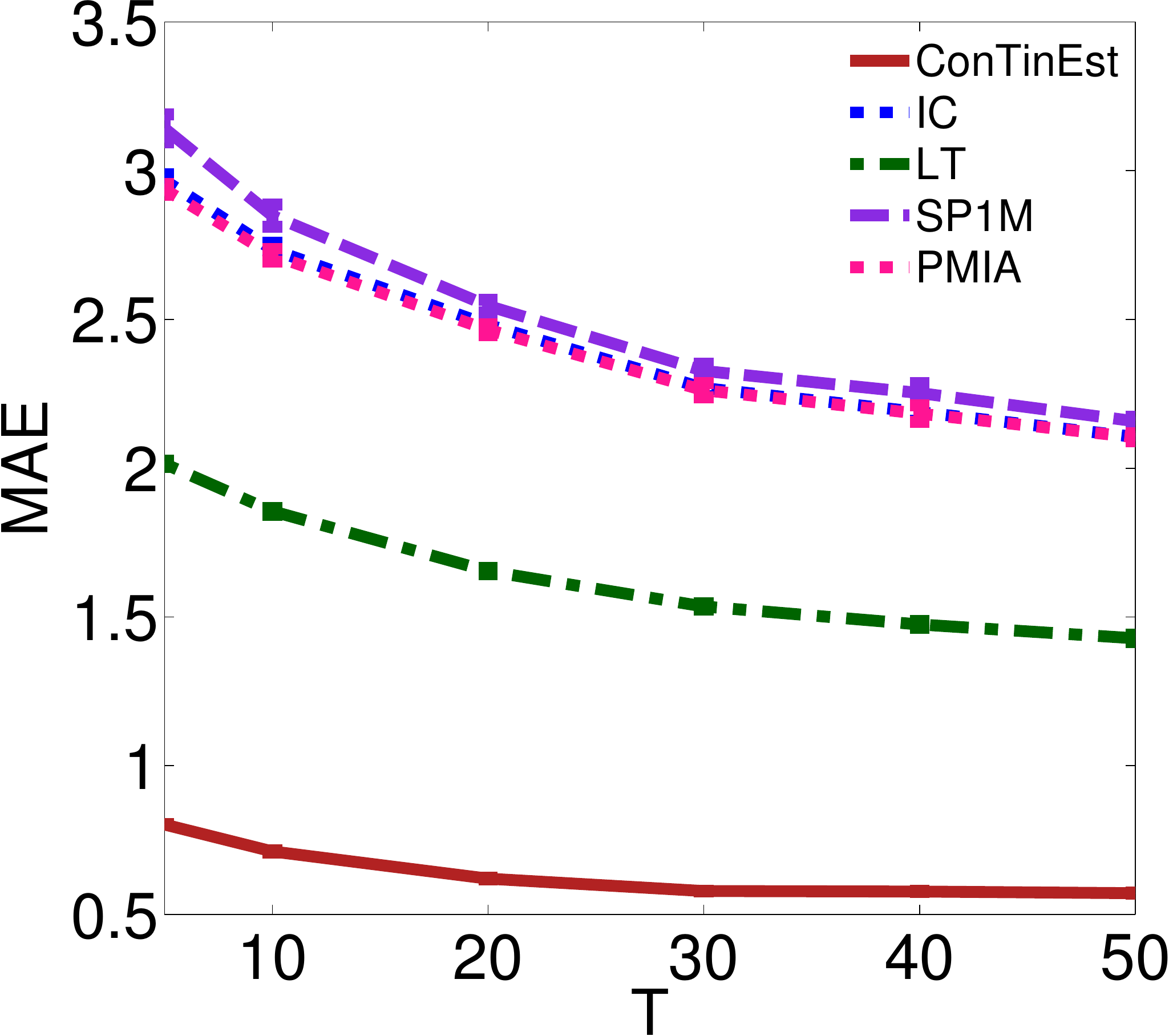}
  \end{tabular}
  }
  \caption{ \label{real_estimation_influence} In MemeTracker dataset, comparison of the accuracy of the estimated influence in terms of mean absolute error.}
 \end{figure}
We first trace the flow of information from one site to another by using the hyperlinks among articles and posts as in the work of~\citet{GomBalSch11, DuSonSmoYua12}. 
In detail, we extracted 10,967 hyperlink cascades among top 600 media sites. We then evaluate the accuracy of \continmax as follows.
First, we repeatedly split all cascades into a 80\% training set and a 20\% test set at random for five times. On each training set, we learn one continuous-time model, which we use for \continmax, and a discrete-time model, which we use for the competitive methods: IC, \spm, \pmia and MIAM-M.
For the continuous-time model, we opt for \netrate~\citep{GomBalSch11} with exponential transmission functions (fixing the shape parameter of the Weibull family to be one) to learn the diffusion networks by maximizing the likelihood of the observed cascades. 
For the discrete-time model, we learn the infection probabilities using the method by~\cite{NetPraSanSuj12}. 

Second, let $\Ccal(u)$ be the set of all cascades where $u$ was the source node. By counting the total number of distinct nodes infected before $T$ in $\Ccal(u)$, we can quantify the real influence of node 
$u$ up to time $T$.
Thus, we can evaluate the quality of the influence estimation by computing the average (across nodes) Mean Absolute Error (MAE) between the real and the estimated influence on the test set,
which we show in Figure~\ref{real_estimation_influence}. 
Clearly, \continmax performs the best statistically. Since the length of real cascades empirically conforms to a power-law distribution, where most cascades are very short (2-4 nodes), 
the gap of the estimation error is not too large. However, we emphasize that such accuracy improvement is critical for maximizing long-term influence since the estimation error for individuals 
will accumulate along the spreading paths. 
Hence, any consistent improvement in influence estimation can lead to significant improvement to the overall influence estimation and maximization task, which is further confirmed in the following 
sections.

\begin{figure}
  \centering
  \renewcommand{\tabcolsep}{5pt}
  {\small
  \begin{tabular}{cc}
    \includegraphics[width=0.45\columnwidth]{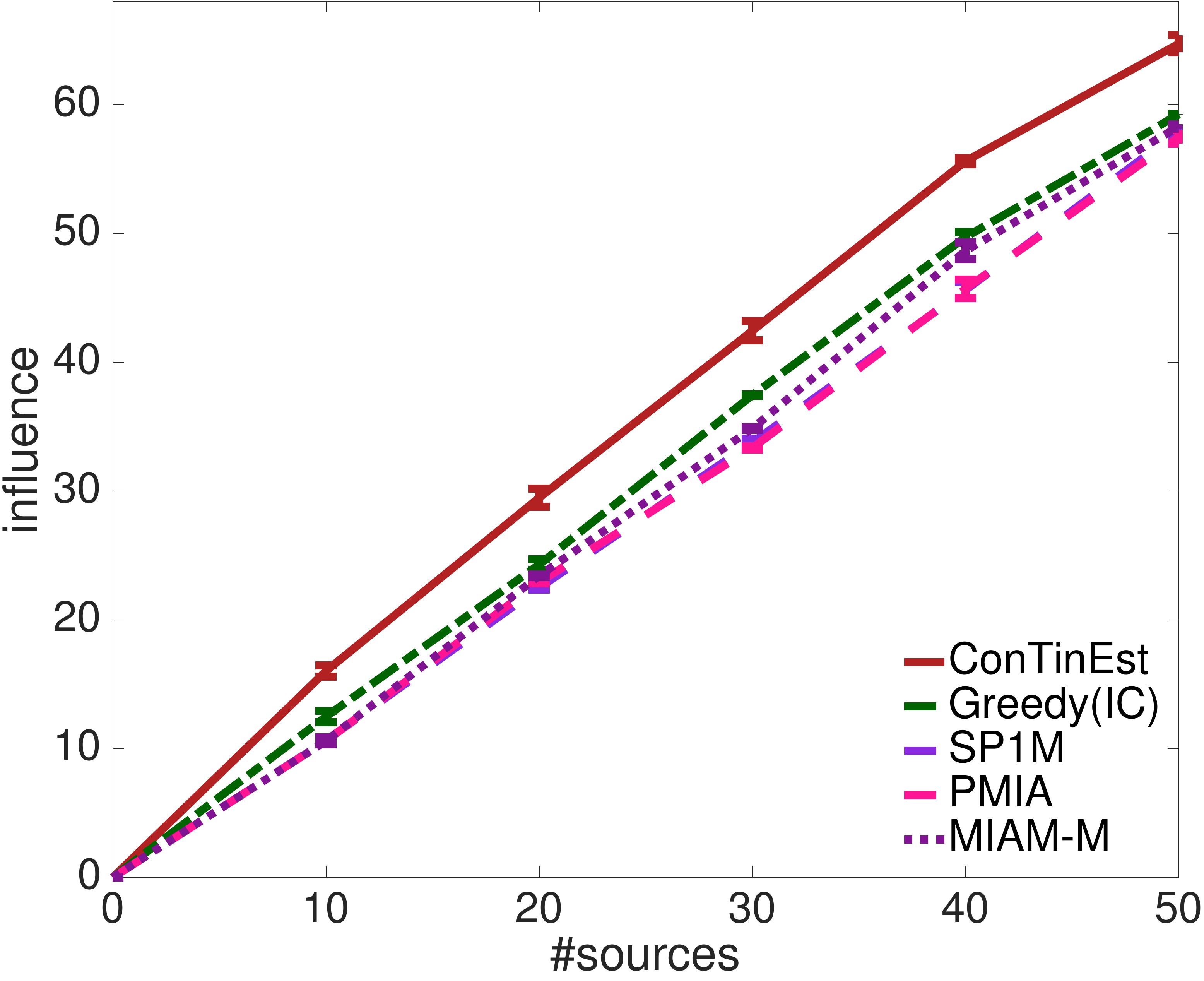} &
    ~\includegraphics[width=0.45\columnwidth]{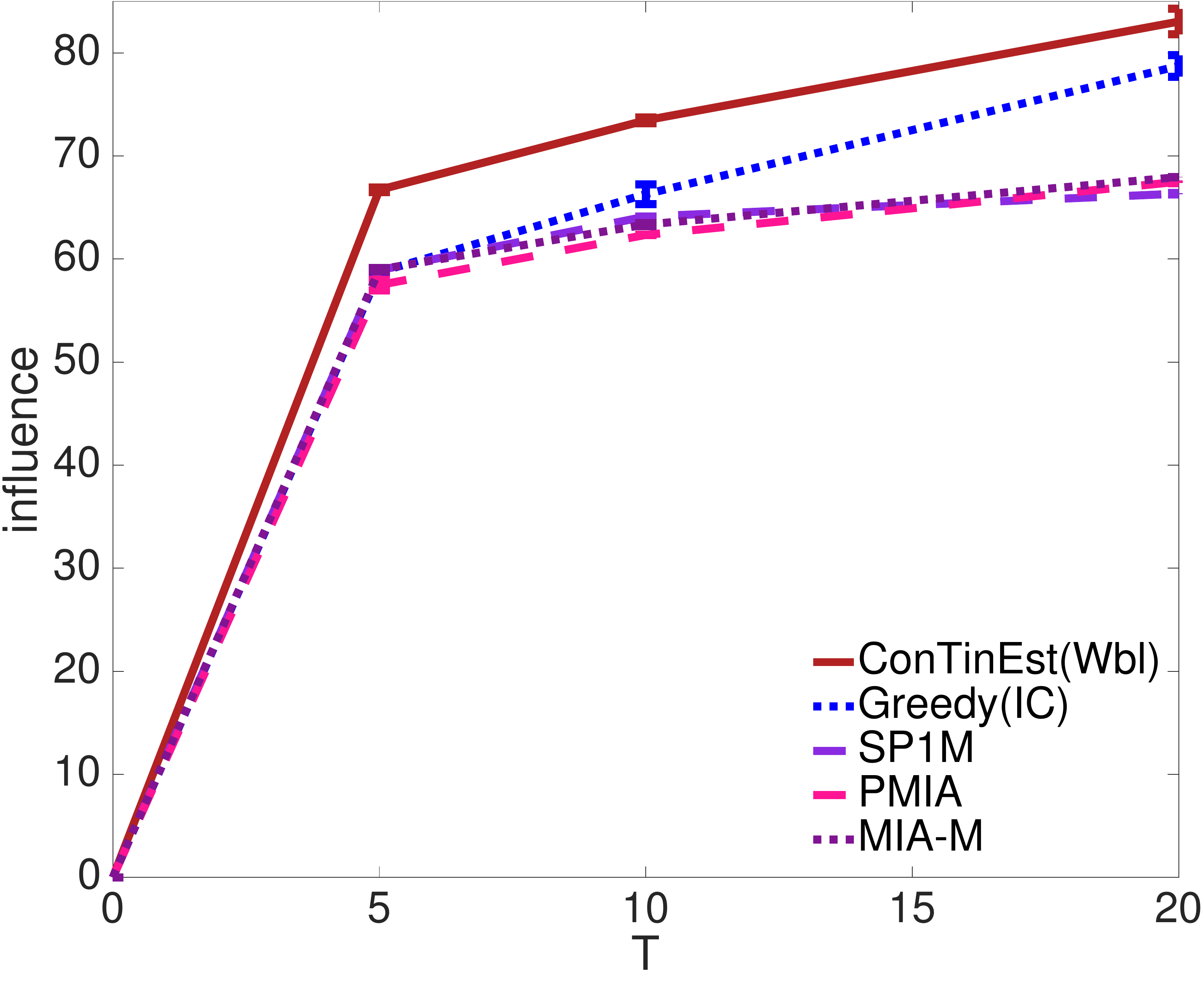} \\
    (a) Influence vs. \#sources  & (b) Influence vs. time\\

  \end{tabular}
  }
  \caption{ \label{real_influence_cardinality} In MemeTracker dataset, (a) comparison of the influence of the selected nodes by fixing the observation window $T=5$ and varying the number sources, and (b) comparison of the influence of the selected nodes by fixing the number of sources to 50 and varying the time window.}
 \end{figure}

%

\subsubsection{Influence Maximization with Uniform Cost}
We first apply \continmax to the continuous-time influence maximization task with the simple cardinality constraint on the users. We evaluate the influence of the selected nodes in the same spirit as influence estimation: the true influence is calculated as the total number of distinct nodes infected before $T$ based on $\Ccal(u)$ of the selected nodes. Figure~\ref{real_influence_cardinality} shows that the selected sources given by \continmax achieve the best performance as we vary the number of selected sources and the observation time window.

Next, we evaluate the performance of \budgetmax on cascades from Memetracker traced from quotes which are short textual phrases spreading through the websites. Because all published documents containing a particular quote are time-stamped, a cascade induced by the same quote is a collection of times when the media site first mentioned it.
In detail, we use the public dataset released by~\cite{GomLesSch2013a}, which splits the original Memetracker dataset into groups, each associated to a topic or real-world event. Each group 
consists of cascades built from quotes which were mentioned in posts containing particular keywords. 
We considered 64 groups, with at least 100,000 cascades, which play the role of products. Therein, we distinguish well-known topics, such as ``Apple" and ``Occupy Wall-Street", or real-world events, 
such as the Fukushima nuclear disaster in 2013 and the marriage between Kate Middleton and Prince William in 2011.

We then evaluate the accuracy of \budgetmax in the following way. First, we evenly split each group into a training and a test set and then learn one continuous-time model and a discrete-time model per group using the training sets.
As previously, for the continuous-time model, we opt for \netrate~\citep{GomBalSch11} with exponential transmission functions, and for the discrete-time model, we learn the infection 
probabilities using the method by~\cite{NetPraSanSuj12}, where the step-length is set to one.
Second, we run \budgetmax using both the continuous-time model and the discrete-time model. We refer to \budgetmax with the discrete-time model as the Greedy(discrete) 
method. 
Since we do have no ground-truth information about cost of each node, we focus our experiments using a uniform cost.
Third, once we have found an allocation over the learned networks, we evaluate the performance of the two methods using the cascades in the test set as follows: 
given a group-node pair $(i,j)$, let $\Ccal(j)$ denote the set of cascades induced by group $i$ that contains node $j$. Then, we take the average number of nodes coming after $j$ for 
all the cascades in $\Ccal(j)$ as a proxy of the average influence induced by assigning group $i$ to node $j$. 
Finally, the influence of an allocation is just the sum of the average influence of each group-node pair in the solution. In our experiments, we randomly select 128 nodes as our target users. 



\begin{figure}[t]
 \centering
 \renewcommand{\tabcolsep}{5pt}
 \begin{tabular}{cc}
\includegraphics[width=0.45\textwidth]{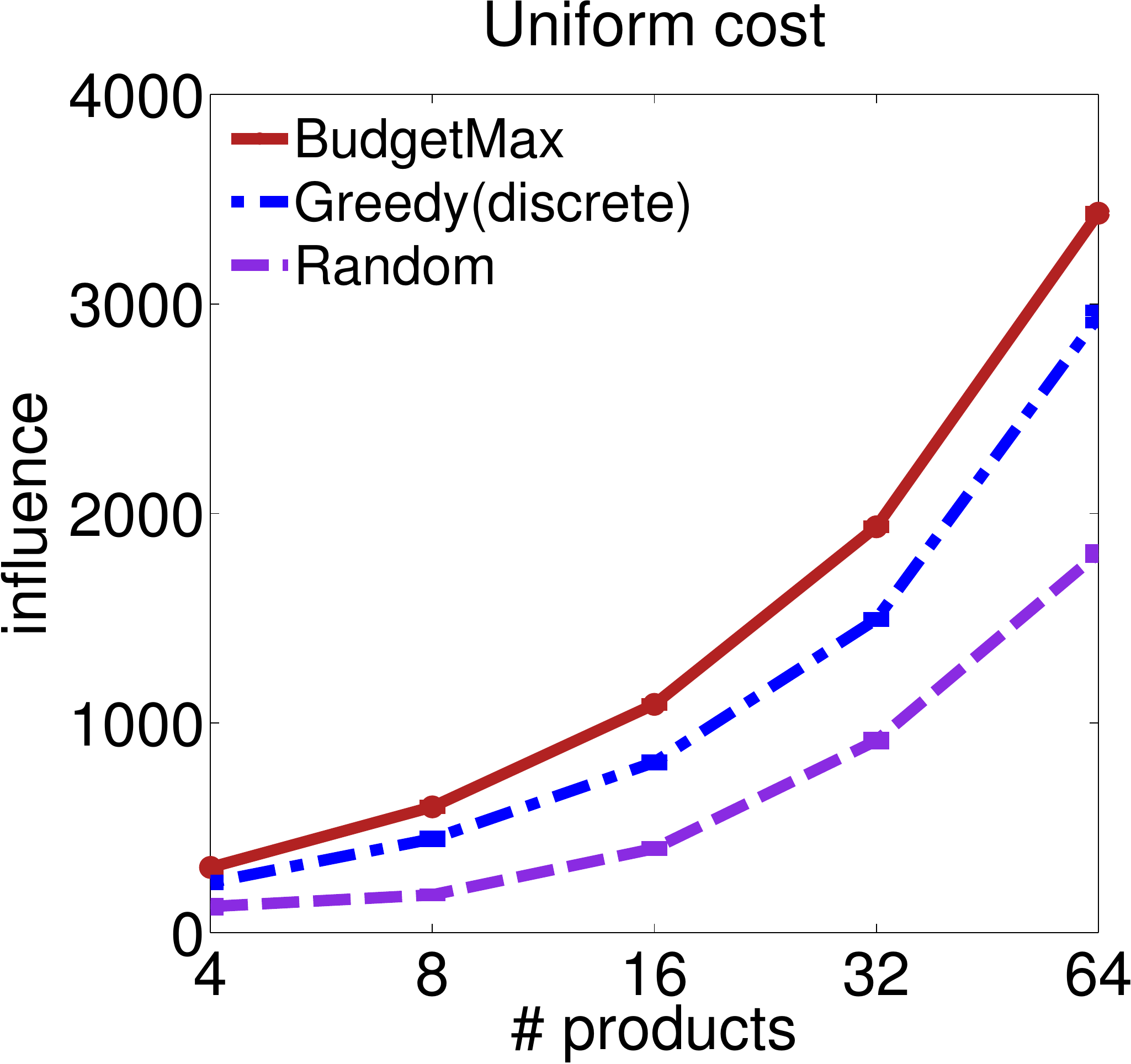}
& \includegraphics[width=0.45\textwidth]{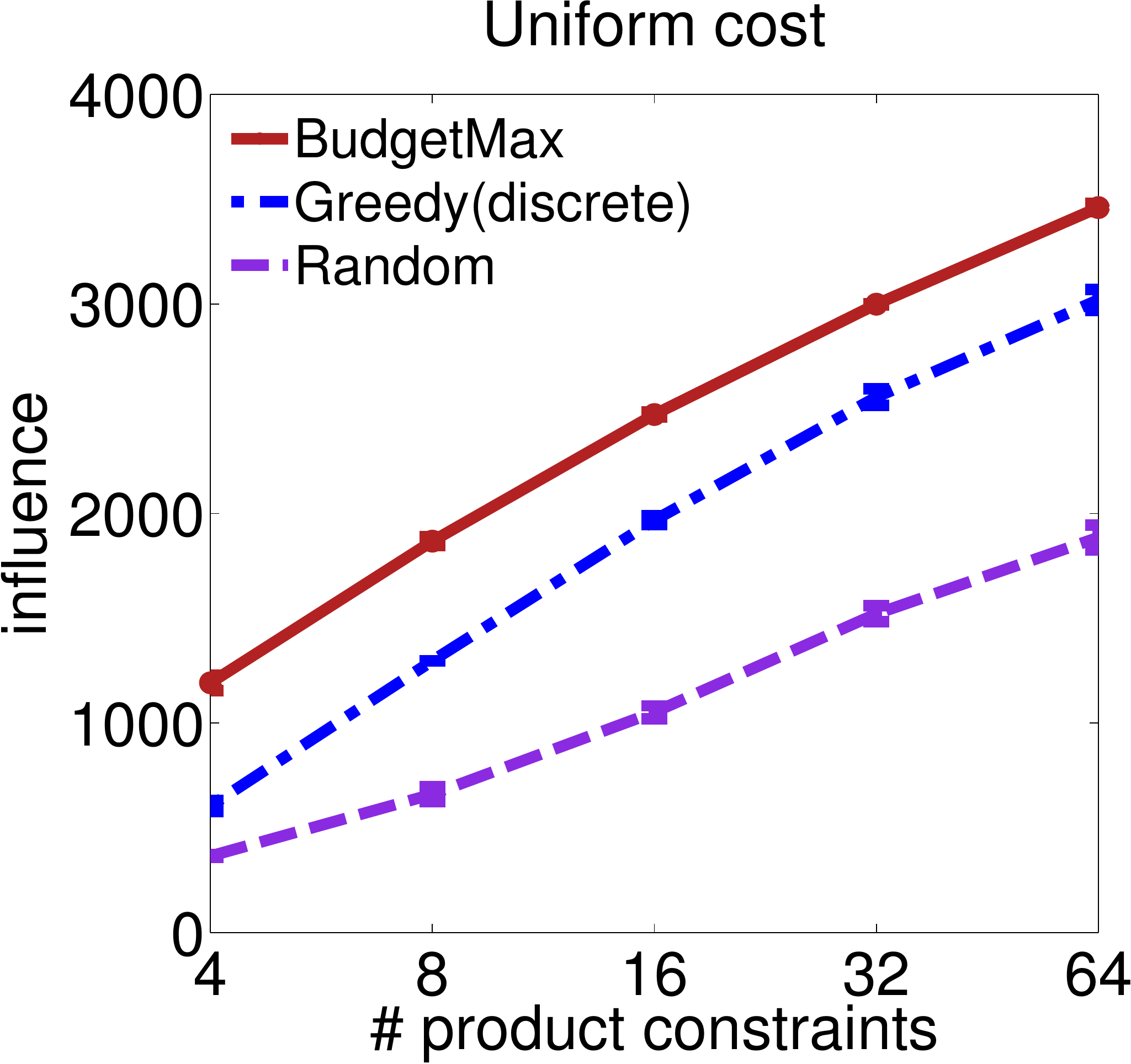}\\
(a) By products & (b)  By product constraints \\[5mm]
\includegraphics[width=0.45\textwidth]{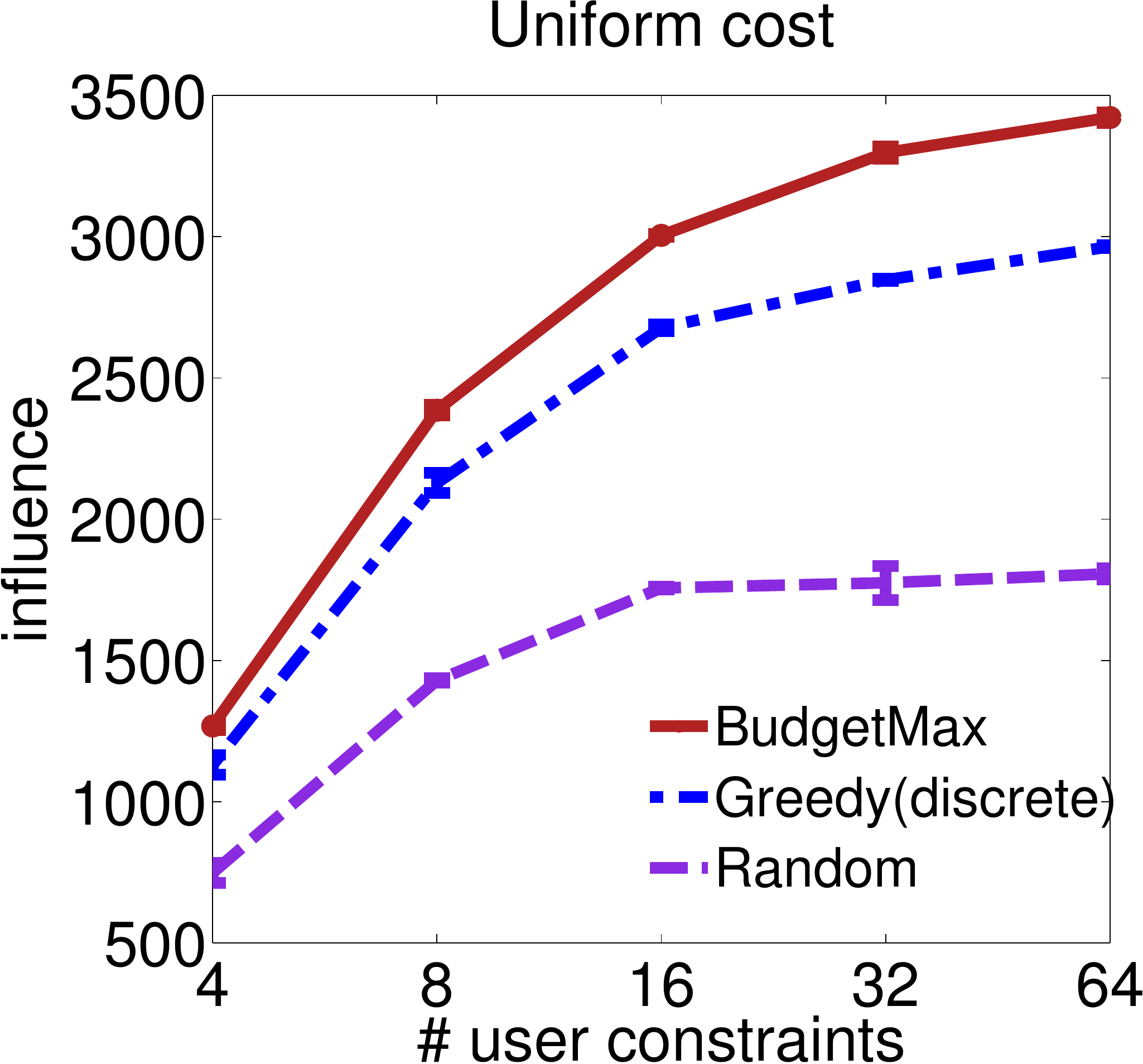}
&\includegraphics[width=0.45\textwidth]{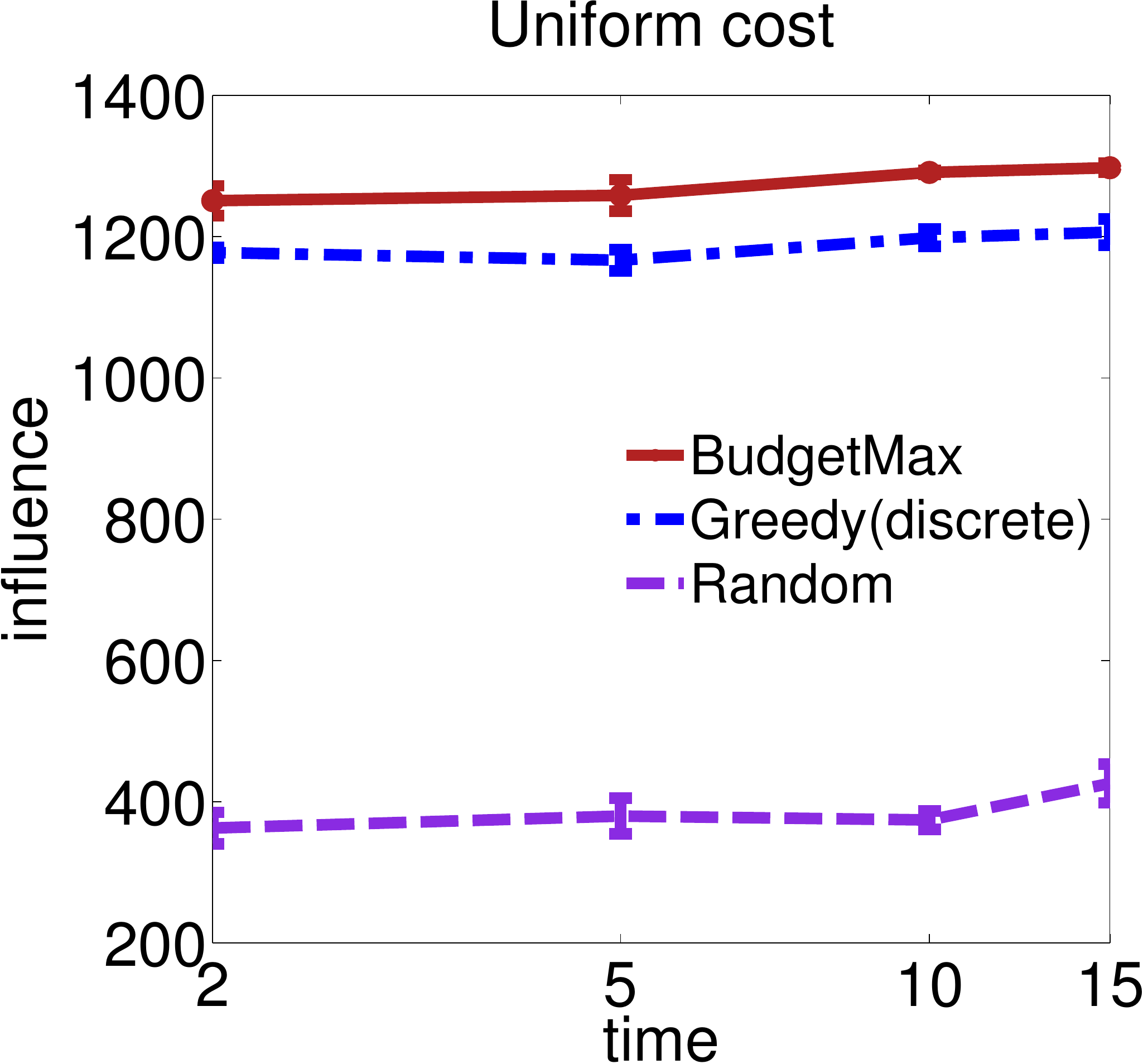} \\
(c) By user constraints & (d) By time
\end{tabular}
 \caption{\label{real_influence_uniform}  Over the inferred 64 product-specific diffusion networks, the true influence estimated from separated testing data (a) for increasing the number of products by fixing the product-constraint at 8 and user-constraint at 2; (b) for increasing product-constraint by fixing user-constraint at 2; (c) for increasing user-constraint by fixing product-constraint at 8; (d) for different time window T.}
\end{figure}


Figure~\ref{real_influence_uniform} summarizes the achieved influence against four factors: (a) the number of products, (b) the budget per product, (c) the budget per user and (d) the time 
window T, while fixing the other factors.
In comparison with the Greedy(IC) and a random allocation, \budgetmax finds an allocation that indeed induces the largest diffusion in the test data, with an average $20$-percent improvement overall.

In the end, Figure~\ref{demo-result} investigates qualitatively the actual allocations of groups (topics or real-world events; in red) and sites (in black). Here, we find examples that intuitively one could
expect: ``japantoday.com" is assigned to Fukushima Nuclear disaster or ``finance.yahoo.com" is assigned to ``Occupy Wall-street". 
Moreover, because we consider several topics and real-world events with different underlying diffusion networks, the selected nodes are not only very popular media sites such as nytimes.com or 
cnn.com but also several modest sites~\citep{BakHofMasWat11}, often specialized or local, such as freep.com or localnews8.com.

\begin{figure}[t]
 \centering
 \begin{tabular}{c}
\includegraphics[width=0.8\columnwidth]{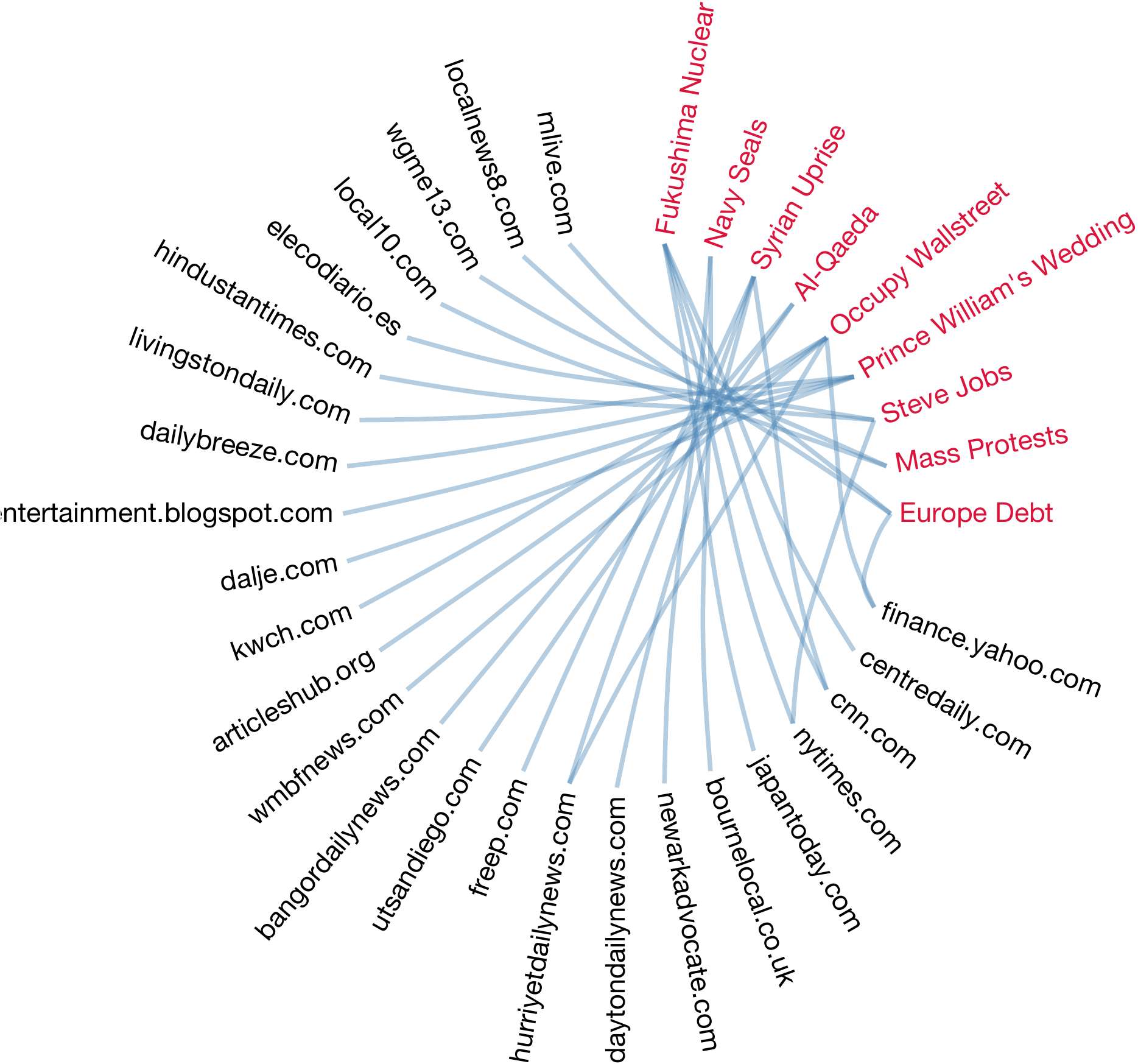}
\end{tabular}
 \caption{\label{demo-result}
The allocation of memes to media sites.}
\end{figure}

\section{Conclusion\label{sec:cl}}
We have studied the influence estimation and maximization problems in the continuous-time diffusion model. 
We first propose a randomized influence estimation algorithm \continmax, which can scale up to networks of millions of nodes while significantly improves over previous state of the art methods in terms of the accuracy of the estimated influence. 
Once we have a subroutine for efficient influence estimation in large networks, we then tackle the problem of maximizing the influence of multiple types of products (or information) in realistic continuous-time diffusion networks subject to various practical constraints: different products can have different diffusion structures; only influence within a given time window is considered; each user can only be recommended a small number of products; and each product has a limited campaign budget, and assigning it to users incurs costs. 
We provide a novel formulation as a submodular maximization under an intersection of matroid constraints and group-knapsack constraints, and then design an efficient adaptive threshold greedy algorithm with provable approximation guarantees, which we call \budgetmax. 
Experimental results show that the proposed algorithm performs remarkably better than other scalable alternatives in both synthetic and real-world datasets. There are also a few interesting open problems. For example, when the influence is estimated using \continmax, its error is a random variable. How does this affect the submodularity of the influence function? Is there an influence maximization algorithm that has better tolerance to the random error? These questions are left for future work.

\newpage

\acks{Nan Du is supported by the Facebook Graduate Fellowship 2014-2015. Maria-Florina Balcan and Yingyu Liang are supported in part by NSF grant CCF-1101283 and CCF-0953192, AFOSR grant FA9550-09-1-0538, ONR grant N00014-09-1-075, a Microsoft Faculty Fellowship, and a Raytheon Faculty Fellowship. Le Song is supported in part by NSF/NIH BIGDATA 1R01GM108341, ONR N00014-15-1-2340, NSF IIS-1218749, NSF CAREER IIS-1350983, Nvidia and Intel. Hongyuan Zha is supported in part by NSF/NIH BIGDATA 1R01GM108341, NSF DMS-1317424 and NSF CNS-1409635.}

%
%
%
%
%
%
%
%

\vskip 0.2in

\clearpage
\appendix

\section{Naive Sampling Algorithm}
\label{app:ns}

The graphical model perspective described in Section~\ref{sec:graphicalmodel} suggests a naive sampling (NS) algorithm for approximating $\sigma(\Acal,T)$:
\begin{itemize}
  \item[1.] Draw $n$ samples, $\cbr{\cbr{\tau_{ji}^l}_{(j,i)\in\Ecal}}_{l=1}^n$,~\iid~from the waiting time product distribution $\prod_{(j,i) \in \Ecal} f_{ji}(\tau_{ji})$;
  \item[2.] For each sample $\cbr{\tau_{ji}^l}_{(j,i)\in\Ecal}$ and for each node $i$, find the shortest path from source nodes to node $i$; count the number of nodes with $g_i\rbr{\cbr{\tau_{ji}^l}_{(j,i)\in\Ecal}}\leq T$;
  \item[3.] Average the counts across $n$ samples.
\end{itemize}

Although the naive sampling algorithm can handle arbitrary transmission function, it is not scalable to networks with millions of nodes. We need to compute the shortest path for each node and each sample, which
results in a computational complexity of $O(n |\Ecal| + n |\Vcal|\log|\Vcal|)$ for a single source node. The problem is even more pressing in the influence maximization problem, where we need to estimate the influence
of source nodes at different location and with increasing number of source nodes. To do this, the algorithm needs to be repeated, adding a multiplicative factor of $C |\Vcal|$ to the computational complexity ($C$ is
the number of nodes to select). Then, the algorithm becomes quadratic in the network size. When the network size is in the order of thousands and millions, typical in modern social network analysis, the naive sampling
algorithm become prohibitively expensive. Additionally, we may need to draw thousands of samples ($n$ is large), further making the algorithm impractical for large-scale problems.

\section{Least Label List}
\label{app:leastlabellist}

The notation ``$\text{argsort}((r_1,\ldots,r_{|\Vcal|}),\text{ascend})$'' in line 2 of Algorithm~\ref{a1} means that we sort the collection of random labels in ascending order and return the argument of the sort as an
ordered list.

\begin{algorithm}[h]
  \SetAlgoVlined
  \KwIn{a reversed directed graph $\Gcal=(\Vcal, \Ecal)$ with edge weights $\{\tau_{ji}\}_{(j,i)\in\Ecal}$, a node labeling $\cbr{r_i}_{i\in \Vcal}$}
  \KwOut{A list $r_\ast(s)$ for each $s\in \Vcal$}

  \lFor{each $s\in\Vcal$}{$d_s\leftarrow\infty, r_\ast(s)\leftarrow\emptyset$}

  \For{$i$ in $\text{argsort}((r_1,\ldots,r_{|\Vcal|}),\text{ascend})$}{
    empty heap ${\tt H}\leftarrow \emptyset$\;

    set all nodes except $i$ as unvisited\;

    push $(0,i)$ into heap ${\tt H}$\;

    \While{${\tt H} \ne\emptyset$}{
      pop $(d_\ast, s)$ with the minimum $d_\ast$ from ${\tt H}$\;

      add $(d_\ast, r_i)$ to the end of list $r_\ast(s)$\;

      $d_s\leftarrow d^\ast$\;

      \For{each unvisited out-neighbor $j$ of $s$}{
                set $j$ as visited\;

        \uIf{$(d, j)$ in heap ${\tt H}$}{
          Pop $(d, j)$ from heap ${\tt H}$\;

          Push $(\min\cbr{d,d_\ast + \tau_{js}}, j)$ into heap ${\tt H}$\;
        }
        \ElseIf{$d_\ast + \tau_{js} < d_j$}{
          Push $(d_\ast + \tau_{js}, j)$ into heap ${\tt H}$\;
        }
      }
    }
  }
  \caption{Least Label List}
  \label{a1}
\end{algorithm}

\begin{figure}[h!]
\begin{minipage}[b]{0.45\textwidth}
\includegraphics[width=1\textwidth,height=95pt]{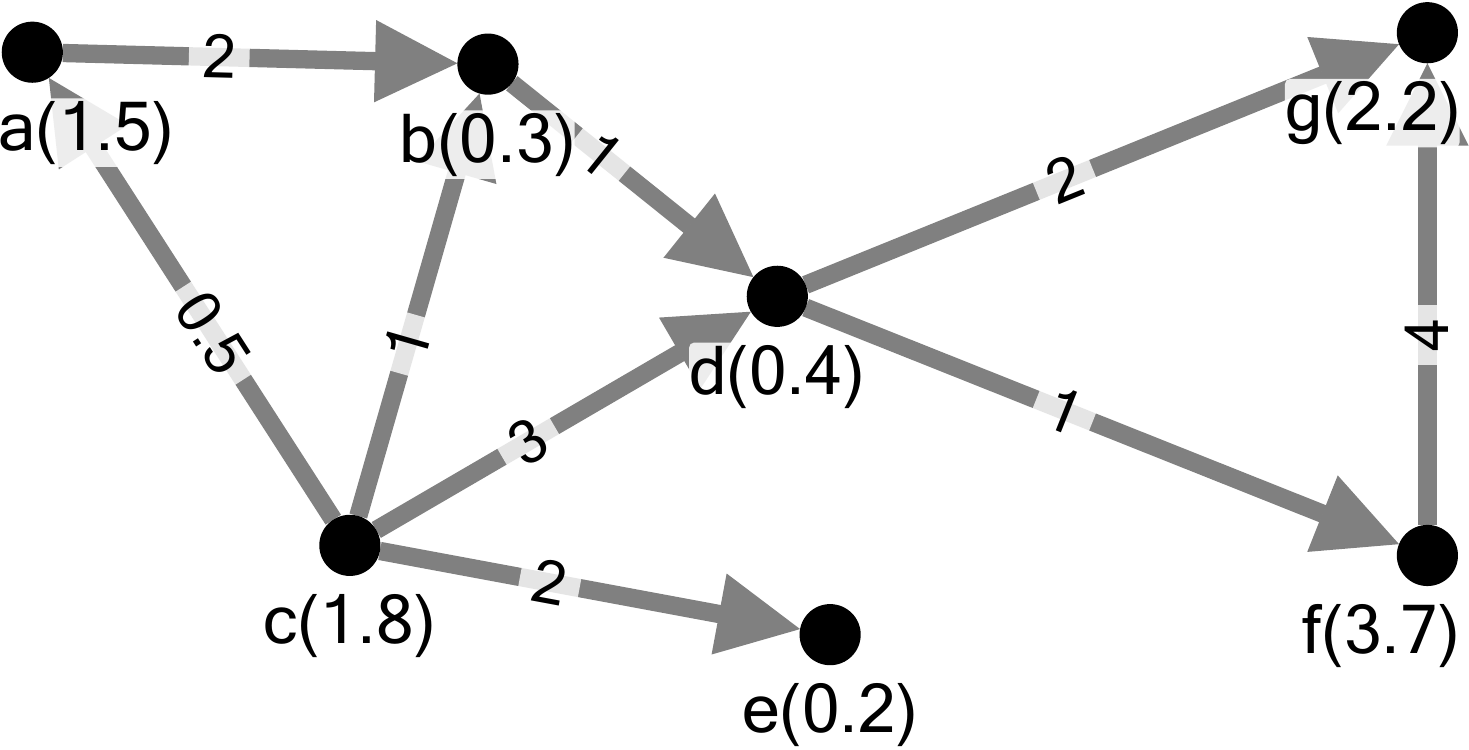}
\end{minipage}
\hspace{0.5cm}
\begin{minipage}[b]{0.45\textwidth}
{\small
  \begin{tabular}{l}
  $\bullet$ Node labeling : \\
  ~~~~$e(0.2) < b(0.3) < d(0.4) < a(1.5) < c(1.8) < g(2.2) < f(3.7)$\\\\
  $\bullet$ Neighborhoods: \\
  ~~~~$\Ncal(c,2)=\{a,b,c,e\};~\Ncal(c,3)=\{a,b,c,d,e,f\};$\\\\
  $\bullet$ Least-label list: \\
  ~~~~$r_\ast(c): (2,0.2),(1,0.3),(0.5,1.5),(0,1.8)$ \\\\
  $\bullet$ Query: $r_\ast(c,0.8) = r(a) = 1.5$ \\

  \end{tabular}
  }
\end{minipage}
\caption{\label{demo}Graph $\Gcal=(\Vcal,\Ecal)$, edge weights $\cbr{\tau_{ji}}_{(j,i)\in\Ecal}$, and node labeling $\cbr{r_i}_{i\in\Vcal}$ with the associated output from Algorithm~\ref{a1}.}
 \end{figure}

Figure~\ref{demo} shows an example of the Least-Label-List. The nodes from $a$ to $g$ are assigned to exponentially distributed labels with mean one shown in each parentheses. Given a query distance 0.8 for node $c$, we can binary-search its Least-label-list $r_*(c)$ to find that node $a$ belongs to this range with the smallest label $r(a) = 1.5$.

\section{Theorem 1}
\label{app:proof1}
\newtheorem{thm2}{Theorem}

\begin{thm2}
  Sample the following number of sets of random transmission times
  \begin{align*}
    n \geqslant \frac{C \Lambda}{\epsilon^2} \log\rbr{\frac{2 |\Vcal|}{\alpha}}
  \end{align*}
  where $\Lambda:= \max_{\Acal:\abr{\Acal}\leq C} 2\sigma(\Acal,T)^2 / (m-2) + 2Var(|\Ncal(\Acal,T)|)(m-1)/(m-2) + 2 a \epsilon /3$, $|\Ncal(\Acal, T)|\leqslant a$, and for each set of random transmission times, sample $m$ set of random labels. Then
  we can guarantee that
  $$
    \abr{\widehat{\sigma}(\Acal,T) - \sigma(\Acal,T)} \leqslant \epsilon
  $$
  simultaneously for all $\Acal$ with $\abr{\Acal}\leqslant C$, with probability at least $1 - \alpha$.
\end{thm2}
\begin{proof}
  Let $S_\tau:=\abr{\Ncal(\Acal,T)}$ for a fixed set of $\{\tau_{ji}\}$ and then $\sigma(\Acal,T) = \EE_{\tau}[S_\tau]$. The randomized algorithm with $m$ randomizations produces an unbiased estimator $\Shat_\tau=(m-1)/(\sum_{u=1}^m r_\ast^u)$ for $S_\tau$,~\ie,~$\EE_{r|\tau}[\Shat_\tau]=S_\tau$, with variance $\EE_{r|\tau}[(\Shat_\tau - S_\tau)^2] = S_\tau^2 / (m-2)$~\citep{Cohen1997}.

  Then $\Shat_\tau$ is also an unbiased estimator for $\sigma(\Acal,T)$, since $\EE_{\tau,r}[\Shat_\tau] = \EE_{\tau} \EE_{r|\tau} [\Shat_\tau] = \EE_{\tau} [S_\tau] = \sigma(\Acal,T)$. Its variance is
  \begin{align*}
    Var(\Shat_\tau)
    &:=\EE_{\tau,r}[(\Shat_\tau - \sigma(\Acal,T))^2] = \EE_{\tau,r}[(\Shat_\tau - S_\tau + S_\tau - \sigma(\Acal,T))^2] \\
    &= \EE_{\tau,r}[(\Shat_\tau - S_\tau)^2] + 2\, \EE_{\tau,r}[(\Shat_\tau - S_\tau)(S_\tau - \sigma(\Acal,T))] + \EE_{\tau,r}[(S_\tau - \sigma(\Acal,T))^2] \\
    &= \EE_{\tau}[S_\tau^2 / (m-2)] + 0\, + Var(S_\tau) \\
    &= \sigma(\Acal,T)^2 / (m-2) + Var(S_\tau)(m-1)/(m-2)
  \end{align*}

  Then using Bernstein's inequality, we have, for our final estimator $\widehat{\sigma}(\Acal,T) = \frac{1}{n}\sum_{l=1}^{n} \Shat_{\tau^l}$, that
  \begin{align}
    \Pr\cbr{\abr{\widehat{\sigma}(\Acal,T) - \sigma(\Acal,T)} \geqslant \epsilon } \leqslant  2 \exp\rbr{-\frac{n \epsilon^2}{2 Var(\Shat_\tau) + 2 a \epsilon / 3}}
    \label{eq:boundeddifference}
  \end{align}
  where $\Shat_\tau < a \leqslant |\Vcal|$.

  Setting the right hand side of relation~\eq{eq:boundeddifference}~to $\alpha$, we have that, with probability $1 - \alpha$, sampling the following number sets of random transmission times
  \begin{align*}
    n & \geqslant \frac{2 Var(\Shat_\tau) + 2 a \epsilon / 3}{\epsilon^2} \log\rbr{\frac{2}{\alpha}} \\
    & = \frac{2\sigma(\Acal,T)^2 / (m-2) + 2Var(S_\tau)(m-1)/(m-2) + 2 a \epsilon /3}{\epsilon^2} \log\rbr{\frac{2}{\alpha}}
  \end{align*}
  we can guarantee that our estimator to have error $\abr{\widehat{\sigma}(\Acal,T) - \sigma(\Acal,T)} \leqslant \epsilon$.


  If we want to insure that $\abr{\widehat{\sigma}(\Acal,T) - \sigma(\Acal,T)} \leqslant \epsilon$ simultaneously hold for all $\Acal$ such that $\abr{\Acal}\leqslant C \ll |\Vcal|$, we can first use union bound with relation~\eq{eq:boundeddifference}. In this case, we have that, with probability $1 - \alpha$, sampling the following number sets of random transmission times
  \begin{align*}
    n \geqslant \frac{C \Lambda}{\epsilon^2} \log\rbr{\frac{2 |\Vcal|}{\alpha}}
  \end{align*}
  we can guarantee that our estimator to have error $\abr{\widehat{\sigma}(\Acal,T) - \sigma(\Acal,T)} \leqslant \epsilon$ for all $\Acal$ with $\abr{\Acal}\leqslant C$.
  Note that we have define the constant $\Lambda:= \max_{\Acal:\abr{\Acal}\leq C} 2\sigma(\Acal,T)^2 / (m-2) + 2Var(S_\tau)(m-1)/(m-2) + 2 a \epsilon /3$.
\end{proof}

\section{Complete Proofs for Section~\ref{sec:influmaximization}}
\subsection{Uniform Cost}
\label{sec:influmaximization:proof:uniform}
In this section, we first prove a theorem for the general problem defined by Equation~\eq{pro:infMax}, considering a normalized monotonic submodular function $f(S)$ and general 
$P$ (Theorem~\ref{thm:dtgreedy}) and $k=0$, 
and then obtain the guarantee for our influence maximization problem (Theorem~\ref{thm:infMax_uni}).

Suppose $G=\cbr{g_1,\dots, g_{|G|}}$ in the order of selection, and let $G^t =\cbr{\g_1, \dots, \g_t}$. Let $C_t$ denote all those elements in $O \setminus G$ that satisfy the 
following: they are still feasible before selecting the $t$-th element $g_t$ but are infeasible after selecting $g_t$. Formally,
$$
  C_t=\cbr{z \in O \setminus G: \cbr{z} \cup G^{t-1} \in \Fcal, \cbr{z} \cup G^t \not\in \Fcal}.
$$

In the following, we will prove three claims and then use them to prove the Theorems~\ref{thm:dtgreedy} and~\ref{thm:infMax_uni}.
Recall that for any $i\in \Ground$ and $S \subseteq \Ground$, the marginal gain of $z$ with respect to $S$ is denoted as
$$
  f(z| S) := f(S\cup\cbr{z}) - f(S)
$$
and its approximation is denoted by
$$
  \widehat f(z| S) = \widehat  f(S\cup\cbr{z}) - \widehat f(S).
$$
Also, when $|f(S)- \widehat{f}(S)| \leq \epsilon$ for any $S \subseteq \Ground$, we have
$$
  |\widehat f(z| S)  - f(z| S) | \leq 2\epsilon
$$ for any $z\in \Ground$ and $S \subseteq \Ground$.

\smallskip
\noindent
\textbf{Claim~\ref{cla:size}. }
{\it
$\sum_{i=1}^t |C_i| \leq P t$, for $t =1, \dots, |G|$.
}

\begin{proof}
We first show the following property about matroids: for any $Q \subseteq \Ground$, the sizes of any two maximal independent subsets $T_1$ and $T_2$ of $Q$
can only differ by a multiplicative factor at most $P$.
Here, $T$ is a maximal independent subset of $Q$ if and only if:
\begin{itemize}
\item $T \subseteq Q$;
\item $T \in \Fcal = \bigcap_{i=1}^P \Ical_p$;
\item $T \cup \cbr{z} \not \in \Fcal$ for any $z \in Q \setminus T$.
\end{itemize}

To prove the property, note that for any element $ z \in T_1 \setminus T_2$,
$\cbr{z} \cup T_2$ violates at least one of the matroid constraints since $T_2$ is maximal.
Let $\{ V_i \}_{1 \leq i \leq P}$ denote all elements in $T_1 \setminus T_2$ that violate the $i$-th matroid,
and then partition $T_1 \cap T_2$ using these $V_i$'s so that they cover $T_1$.
Note that the size of each $V_i$ must be at most that of $T_2$, since otherwise by the Exchange axiom, there 
would exist $z \in V_i \setminus T_2$ that can be added to $T_2$ without violating the $i$-th matroid, leading 
to a contradiction.
Therefore, $|T_1|$ is at most $P$ times $|T_2|$.

Now we apply the property to prove the claim. Let $Q$ be the union of $G^{t}$ and $\bigcup_{i=1}^t C_t$.
On one hand, $G^{t}$ is a maximal independent subset of $Q$, since no element in $\bigcup_{i=1}^t C_t$ can be added to $G^t$ without violating the matroid constraints.
On the other hand, $\bigcup_{i=1}^t C_t$ is an independent subset of $Q$, since it is part of the optimal solution.
Therefore, $\bigcup_{i=1}^t C_t$ has size at most $P$ times $|G^t|$, which is $Pt$.
\end{proof}

\smallskip
\noindent
\textbf{Claim~\ref{cla:gain}. }
{\it
Suppose $g_t$ is selected at the threshold $\tau_t$. Then, 
$f(j|G^{t-1}) \leq (1+\delta) \tau_t + 4\epsilon + \frac{\delta}{\nGround} f(G), \forall j \in C_t$.
}

\begin{proof}
First, consider $\tau_t > w_{L+1} = 0$.
Since $g_t$ is selected at the threshold $\tau_t$, we have that $\widehat f( \g_t | \Greedy^{t-1}) \geq \tau_t$ and thus $f( \g_t | \Greedy^{t-1}) \geq \tau_t - 2\epsilon$.
Any $j \in C_t$ could have been selected at an earlier stage, since adding $j$ to $\Greedy^{t-1}$ would not have violated the constraint.
However, since $j \not\in \Greedy^{t-1}$, that means that $\widehat f( j| \Greedy^{t-1}) \leq (1+\delta) \tau_t$.
Then,
$$
  f(j|\Greedy^{t-1}) \leq (1+\delta) \tau_t + 2\epsilon.
$$

Second, consider $\tau_t = w_{L+1} = 0$. For each $j \in C_t$, we have $\widehat f(j|\Greedy) < \frac{\delta }{\nGround} d$.
Since the greedy algorithm must pick $g_1$ with $\widehat f(g_1) = d$ and $d \leq f(g_1)+\epsilon$, then
$$
  f(j|\Greedy) < \frac{\delta }{\nGround} f(G) + 4\epsilon.
$$
The claim follows by combining the two cases.
\end{proof}

\smallskip
\noindent
\textbf{Claim~\ref{cla:com}. }
{\it
The marginal gain of $O\setminus G$ satisfies
$$\sum_{j \in O\setminus G} f(j|\Greedy) \leq [(1+\delta) P + \delta] f(G)  + (6+2\delta) \epsilon P |G|. $$
}

\begin{proof}
Combining Claim~\ref{cla:size} and Claim~\ref{cla:gain}, we have:
\begin{align*}
 \sum_{j \in O\setminus G} f(j|\Greedy) = \sum_{t=1}^{|G|} \sum_{j \in C_t} f(j|\Greedy)
& \leq (1+\delta) \sum_{t=1}^{|G|} |C_t| \tau_t + \delta f(G)  + 4\epsilon \sum_{t=1}^{|G|} |C_t|\\
& \leq (1+\delta) \sum_{t=1}^{|G|} |C_t| \tau_t + \delta f(G)  + 4\epsilon P |G|.
\end{align*}
Further, $\sum_{t=1}^{|G|} |C_t| \tau_t \leq P \sum_{t=1}^{|G|} \tau_t$ by Claim~\ref{cla:size} and a technical lemma (Lemma~\ref{lem:seqsum}).
Finally, the claim follows from the fact that $f(G) = \sum_t f(g_t|G^{t-1}) \geq \sum_t (\tau_t - 2\epsilon)$.
\end{proof}

\begin{lemma}\label{lem:seqsum}
If $\sum_{i=1}^t \sigma_{i-1} \leq t$ for $t=1,\dots, K$ and $\rho_{i-1} \geq \rho_i$ for $i=1,\dots,K-1$ with $\rho_i, \sigma_i\geq 0$,
then $\sum_{i=1}^K \rho_i\sigma_i \leq \sum_{i=1}^K \rho_{i-1}$.
\end{lemma}
\begin{proof}
Consider the linear program
\begin{eqnarray*}
V&=&\max_{\sigma} \sum_{i=1}^K \rho_i\sigma_i\\
& \textrm{s.t.} & \sum_{i=1}^t \sigma_{i-1} \leq t, \ \ t=1,\dots,K,\\
&& \sigma_i\geq 0, \ \ i=1,\dots,K-1
\end{eqnarray*}
with dual
\begin{eqnarray*}
W&=&\min_{u} \sum_{i=1}^K t u_{t-1}\\
& \textrm{s.t.} & \sum_{t=i}^{K-1} u_t \geq \rho_i, \ \ i=0,\dots,K-1,\\
&& u_t\geq 0, \ \ t=0,\dots,K-1.
\end{eqnarray*}
As $\rho_i \geq \rho_{i+1}$, the solution $u_i = \rho_i - \rho_{i+1},i=0,\dots,K-1$ (where $\rho_K=0$)
is dual feasible with value $\sum_{t=1}^K t (\rho_{t-1}-\rho_t) = \sum_{i=1}^{K} \rho_{i-1}$.
By weak linear programming duality,  $\sum_{i=1}^K \rho_i\sigma_i \leq V \leq W \leq \sum_{i=1}^K \rho_{i-1}$.
\end{proof}

\begin{theorem}\label{thm:dtgreedy}
Suppose we use Algorithm~\ref{alg:greedyFixedDensity} to solve the problem defined by Equation~\eq{pro:infMax} with $k=0$, using $\rho=0$ and $\widehat f$ to estimate the 
function $f$, where $|\widehat f(S) - f(S)| \leq \epsilon$ for all $S \subseteq \Ground$.
It holds that the algorithm returns a greedy solution $\Greedy$ with
$$f(\Greedy) \geq \frac{1}{(1+2\delta)(P+1)} f(\Optimal) - \frac{4 P |\Greedy|}{P+\cur_f}\epsilon$$
where $O$ is the optimal solution, using $\Ocal(\frac{\nGround}{\delta}\log\frac{\nGround}{\delta})$ evaluations of $\widehat f$.
\end{theorem}

\begin{proof}
By submodularity and Claim~\ref{cla:com}, we have:
\begin{align*}
f(O) & \leq f(O\cup G) \leq f(G) + \sum_{j \in O\setminus G} f(j|\Greedy) \leq (1+\delta)(P + 1) f(G)  + (6+2\delta) \epsilon P |G|,
\end{align*}
which leads to the bound in the theorem.

Since there are $\Ocal(\frac{1}{\delta}\log \frac{\nGround}{\delta})$ thresholds, and there are $\Ocal(\nGround)$ evaluations at each threshold, the number of evaluations 
is bounded by $\Ocal(\frac{\nGround}{\delta}\log\frac{\nGround}{\delta})$.
\end{proof}

Theorem~\ref{thm:dtgreedy} essentially shows $f(G)$ is close to $f(O)$ up to a factor roughly $(1+P)$, which then leads to the following guarantee for our influence maximization problem. Suppose product $i \in \Item$ spreads according to diffusion network $\Graph_i = (\Node, \Edge_i)$, and let $i^*=\argmax_{i\in\Item}|\Edge_i|$.

\smallskip
\noindent
\textbf{Theorem~\ref{thm:infMax_uni}. }
{\it
In the influence maximization problem with uniform cost,  Algorithm~\ref{alg:greedyFixedDensity} (with $\rho=0$) is able to output a solution $G$ that satisfies
$
  f(\Greedy) \geq \frac{1-2\delta}{3} f(\Optimal)
$
in expected time $\widetilde\Ocal\left(\frac{|\Edge_{i^*}|+|\Node|}{\delta^2}  + \frac{|\Item||\Node|}{\delta^3} \right).$
}

\begin{proof}
In the influence maximization problem, the number of matroids is $P=2$. 
Also note that $|G| \leq f(G) \leq f(O)$, which leads to $4 |\Greedy|\epsilon \leq 4\epsilon f(O)$.
The approximation guarantee then follows from setting $\epsilon\leq \delta/16$ when using \continmax~\citep{DuSonZhaGom13} to estimate the influence.

The runtime is bounded as follows.
In Algorithm~\ref{alg:greedyFixedDensity}, we need to estimate the marginal gain of adding one more product to the current solution.
In \continmax~\citep{DuSonZhaGom13}, building the initial data structure takes time
$$
    \Ocal\left((|\Edge_{i^*}|\log|\Node| + |\Node|\log^2 |\Node|) \frac{1}{\delta^2} \log \frac{|\Node|}{\delta} \right)
$$
and afterwards each function evaluation takes time $$
    \Ocal\left(\frac{1}{\delta^2} \log \frac{|\Node|}{\delta} \log\log|\Node|\right).
$$
As there are $\Ocal\left(\frac{\nGround}{\delta}\log\frac{\nGround}{\delta}\right)$
evaluations where $N=|\Item||\Node|$, the runtime of our algorithm follows.
\end{proof}

\subsection{General case}
\label{sec:influmaximization:proof:general}
As in the previous section, we first prove that a theorem for the problem defined by Equation~\eq{pro:infMax} with general normalized monotonic submodular function $f(S)$ and general $P$ 
(Theorem~\ref{thm:densityEnu}), and then obtain the guarantee for our influence maximization problem (Theorem~\ref{thm:infMax}).

\begin{theorem}\label{thm:densityEnu}
Suppose Algorithm~\ref{alg:densityEnu} uses $\widehat f$ to estimate the function $f$
which satisfies $|\widehat f(S) - f(S)| \leq \epsilon$ for all $S \subseteq \Ground$.
Then, there exists a $\rho$ such that
\begin{align*}
f(S_\rho) \geq \frac{\max\cbr{1,|A_\rho|} }{(P+2k+1)(1+2\delta)} f(O) - 8\epsilon |S_\rho|
\end{align*}
where $A_\rho$ is the set of active knapsack constraints:
$$
 A_\rho = \{i: S_\rho  \cup \{z\} \not\in \Fcal,  \forall z \in \Ground_{i*}\} .
$$
\end{theorem}

\begin{proof}
Consider the optimal solution $O$ and set $\rho^* = \frac{2}{P+ 2k+1} f(O)$.
By submodularity, we have $d \leq f(O) \leq |\Ground| d$, so $\rho \in \sbr{ \frac{2d}{P+2k+1}, \frac{2|\Ground|d}{P+2k+1} }$,
and there is a run of Algorithm~\ref{alg:greedyFixedDensity} with $\rho$ such that $\rho^* \in [\rho, (1+\delta) \rho]$.
In the following we consider this run.

\smallskip
\noindent
{\bf Case 1} Suppose $|A_\rho| = 0 $. The key observation in this case is that since no knapsack constraints are active,
the algorithm runs as if there were only matroid constraints. Then, the argument for matroid constraints can be applied.
More precisely, let
$$
    O_+ := \cbr{z \in O\setminus S_\rho :  f(z|S_\rho) \geq c(z)\rho + 2\epsilon}
$$
$$
    O_- := \cbr{z \in O\setminus S_\rho : z \not\in O_+ }.
$$
Note that all elements in $O_+$ are feasible.
Following the argument of Claim~\ref{cla:com} in Theorem~\ref{thm:dtgreedy}, we have:
\begin{align}
    f(O_+|S_\rho) \leq ((1+\delta)P + \delta) f(S_\rho) + (4+2\delta) \epsilon P|S_\rho|.
\label{eqn:matroidcase1}
\end{align}
Also, by definition, the marginal gain of $O_-$ is:
\begin{align}
    f(O_-|S_\rho) \leq k \rho + 2\epsilon |O_-| \leq k \rho + 2\epsilon P |S_\rho|,
\label{eqn:matroidcase2}
\end{align}
where the last inequality follows from the fact that $S_\rho$ is a maximal independent subset, $O_-$ is an independent subset of $O\cup S_\rho$, and 
the sizes of any two maximal independent subsets in the intersection of $P$ matroids can differ by a factor of at most $P$.
Plugging (\ref{eqn:matroidcase1})(\ref{eqn:matroidcase2}) into $f(O) \leq f(O_+ |S_\rho) + f(O_- |S_\rho) + f(S_\rho)$, we obtain the bound
$$
    f(S_\rho) \geq \frac{f(O)}{(P+ 2 k + 1)(1+\delta)}  - \frac{(6+2\delta) \epsilon P|S_\rho| }{(P+ 1)(1+\delta)}.
$$

\smallskip
\noindent
{\bf Case 2}
Suppose $|A_\rho| > 0 $.
For any $i \in A_\rho$ (\ie, the $i$-th knapsack constraint is active), consider the step when $i$ is added to $A_\rho$.
Let $\Greedy_i = \Greedy \cap \Ground_{i*}$, and we have $c(\Greedy_i) + c(z) > 1$.
Since every element $g$ we include in $\Greedy_i$ satisfies $\widehat{f}(g|\Greedy)  \geq c(g) \rho$ with respect to the solution $\Greedy_i$ when $g$ is added.
Then $f(g|\Greedy) = f_i(g|\Greedy_i) \geq c(g) \rho - 2\epsilon$, and we have:
\begin{align}
f_i(\Greedy_i \cup \cbr{z}) & \geq \rho [c(\Greedy_i) + c(z)] - 2 \epsilon (|G_i| + 1) > \rho - 2 \epsilon (|G_i| + 1).\label{eqn:knapcase1}
\end{align}
Note that $\Greedy_i$ is non-empty since otherwise the knapsack constraint will not be active.
Any element in $\Greedy_i$ is selected before or at $w_{t}$, so $f_i(G_i) \geq w_{t} - 2\epsilon$.
Also, note that $z$ is not selected in previous thresholds before $w_{t}$, so $f_i( \cbr{z} | \Greedy_i) \leq (1+\delta) w_t + 2\epsilon$ and thus,
\begin{align}
f_i( \cbr{z} | \Greedy_i) \leq (1+\delta) f_i(\Greedy_i) + 2\epsilon (2+\delta). \label{eqn:knapcase2}
\end{align}
Combining Eqs. \ref{eqn:knapcase1} and \ref{eqn:knapcase2} into $ f_i(\Greedy_i \cup \cbr{z}) = f_i(\Greedy_i) + f_i( \cbr{z}  | \Greedy_i)$ leads to
\begin{align*}
f_i(\Greedy_i)  \geq \frac{\rho}{(2+\delta)} - \frac{2\epsilon(|G_i| + 3 + \delta) }{(2+\delta)}  & \geq \frac{1}{2(1+2\delta)} \rho^* - \frac{2\epsilon(|G_i| + 3 + \delta) }{(2+\delta)} \\
& \geq \frac{f(O)}{(P+2k+1)(1+2\delta)} - 5 \epsilon |G_i|.
\end{align*}
Summing up over all $i \in A_\rho$ leads to the desired bound.
\end{proof}

Suppose item $i \in \Item$ spreads according to
the diffusion network $\Graph_i = (\Node, \Edge_i)$. Let $i^*=\argmax_{i\in\Item}|\Edge_i|$.
By setting $\epsilon=\delta/16$ in Theorem~\ref{thm:densityEnu}, we have:

\smallskip
\noindent
\textbf{Theorem~\ref{thm:infMax}. }
{\it
In Algorithm~\ref{alg:densityEnu}, there exists a $\rho$ such that
$$
    f(S_\rho) \geq  \frac{\max\cbr{k_a, 1} }{(2|\Item|+2) (1+3\delta)} f(O)
$$
where $k_a$ is the number of active knapsack constraints.
The expected runtime to obtain the solution is $\widetilde\Ocal\left(\frac{|\Edge_{i^*}|+|\Node|}{\delta^2}  + \frac{|\Item||\Node|}{\delta^4} \right).$
}

\end{document}